\documentclass{lmcs} 
\pdfoutput=1

\usepackage{silence}
\usepackage{etoolbox}
\apptocmd{\sloppy}{\hbadness 10000\relax}{}{} 

\usepackage{lastpage}
\lmcsdoi{15}{3}{18}
\lmcsheading{}{\pageref{LastPage}}{}{}%
{Oct.~26,~2017}{Aug.~21,~2019}{}

\keywords{Decidable Logic, Quantifier Instantiation, EPR, Inductive Invariants, Deductive Verification, Decision Procedures}

\usepackage{hyperref}

\usepackage{amsmath}
\usepackage{url,alltt}
\usepackage{xspace}
\usepackage{caption}
\usepackage{subcaption}
\usepackage{graphicx}
\usepackage{wrapfig}
\usepackage{amssymb,amsmath,pifont,footnote}
\usepackage{latexsym,pdfpages}
\usepackage{array,verbatim}
\usepackage{stmaryrd,ifsym}
\usepackage{gastex} 
\usepackage{parcolumns}
\usepackage{listings}
\usepackage{color}
\usepackage{cleveref}
\usepackage{amsthm}
\usepackage{pdfpages}

\usepackage{alltt}
\usepackage{mathtools}
\usepackage{amsfonts}
\usepackage{amssymb}
\usepackage{textcomp}
\usepackage{epsfig}
\usepackage{changebar}
\usepackage{url}
\usepackage[ruled]{algorithm2e}
\usepackage{xspace}
\usepackage{datetime}
\usepackage{multirow}
\usepackage{stmaryrd}
\usepackage{fancyvrb}
\usepackage[flushleft]{threeparttable}
\usepackage[toc,page]{appendix}
\usepackage{comment}
\usepackage{hhline}
\usepackage{booktabs}

\newlist{compactenum}{enumerate}{3}
\setlist[compactenum]{topsep=0pt,partopsep=0pt,itemsep=0pt,parsep=0pt}
\setlist[compactenum,1]{label=\arabic*}
\setlist[compactenum,2]{label=\alph*}
\setlist[compactenum,3]{label=\roman*}

\newif\iflong%
\newif\iflonglong%
\newif\ifcomments%

\longtrue%
\longlongtrue%
\commentsfalse%

\newcommand{\para}[1]{\ultpara{#1}}
\newcommand{\sectionette}[1]{\ultpara{#1}}
\newcommand{\ultpara}[1]{\paragraph{\textbf{#1}}}

\newcommand{\til}{,\ldots,}

\newcommand{\Nat}{{\mathbb{N}}}

\newcommand{\ov}{\overline}

\renewcommand{\phi}{\varphi}

\newcommand{\A}{{\mathcal{A}}}        
\newcommand{\sland}{\;\land\;}

\renewcommand{\implies}{{\Rightarrow}}

\newcommand{\true}{{\textit{true}}}
\newcommand{\false}{{\textit{false}}}
\newcommand{\vocabulary}{{\Sigma}}

\newcommand{\struct}{{\sigma}}
\newcommand{\Dom}{{D}}
\newcommand{\Int}{{\mathcal{I}}}

\newcommand{\assgn}{v}

\newcommand{\Init}{\phi_0}
\newcommand{\PhiProp}{\phi_{\Property}}

\newcommand{\Property}{{P}}
\newcommand{\state}{s}

\newcommand{\Univ}{{\forall^*}}

\newcommand{\gV}[1]{\langle\textit{#1}\rangle}
\newcommand{\btuple}[1]{\overline{#1}}
\newcommand{\bvar}{\textbf{v}}
\newcommand{\brel}{\textbf{r}}
\newcommand{\bfun}{\textbf{f}}

\newcommand{\binsert}{\texttt{insert}}

\newcommand{\babort}{\texttt{abort}}
\newcommand{\bskip}{\texttt{skip}}
\newcommand{\baction}{\texttt{action}}
\newcommand{\bdo}{\texttt{do}}
\newcommand{\bwhile}{\texttt{while}}
\newcommand{\bassume}{\texttt{assume}}

\newcommand{\bif}{\texttt{if}}
\newcommand{\bthen}{\texttt{then}}
\newcommand{\belse}{\texttt{else}}
\newcommand{\bite}{\texttt{ite}}

\newcommand{\cinit}{\Cmd_{\texttt{init}}}
\newcommand{\cbody}{\Cmd_{\texttt{body}}}
\newcommand{\cfinal}{\Cmd_{\texttt{final}}}

\newcommand{\phiqf}{\varphi_{\texttt{QF}}}
\newcommand{\phiea}{\varphi_{\texttt{EA}}}
\newcommand{\phiae}{\varphi_{\texttt{AE}}}
\newcommand{\phiaf}{\varphi_{\texttt{AF}}}
\newcommand{\gVphiqf}{{\langle\phiqf\rangle}}
\newcommand{\gVphiea}{{\langle\phiea\rangle}}
\newcommand{\gVphiae}{{\langle\phiae\rangle}}
\newcommand{\gVphiaf}{{\langle\phiaf\rangle}}

\newcommand{\IF}{\texttt{if}}
\newcommand{\THEN}{\texttt{then}}
\newcommand{\ELSE}{\texttt{else}}

\newcommand{\ACTION}[2]{{\texttt{action #1}}(#2) \; \{} 
\newcommand{\ACTIONNAME}[1]{{\texttt{#1}}}
\newcommand{\ENDACTION}{\}}
\newcommand{\INSERTM}[2]{{#1} . \binsert\left({#2} \right)} 
\newcommand{\INSERTMU}[3]{{#1} . \binsert \left({#2}~|~{#3}\right)} 
\newcommand{\REMOVEM}[2]{#1 := #1 \setminus\{ #2 \}} 
\newcommand{\ASSUME}{\texttt{assume}}
\newcommand{\ASSUMEM}[1]{\ASSUME\,#1}

\newcommand{\ABORT}{\texttt{abort}}

\newcommand{\COMMENT}[1]{\texttt{\# #1}}
\newcommand{\INDENT}{\quad}
\newcommand{\CSEP}{;}

\newcommand{\INSTRCOMMENTS}{/@}
\newcommand{\INSTRCOMMENT}{\INSTRCOMMENTS \quad} 
\newcommand{\INSTRREPLACE}{\INSTRCOMMENTS \hookrightarrow} 
\newcommand{\UPDATECODE}[1]{\INSTRCOMMENT #1} 
\newcommand{\DERIVED}[2]{\INSTRCOMMENT #1 \equiv #2} 
\newcommand{\INVARIANT}{\mbox{Invariant}\;}

\newcommand{\mPending}{{\textit{pending}}}

\newcommand{\mSent}{{\textit{sent}}}
\newcommand{\mLeader}{{\textit{leader}}}

\newcommand{\mRingNext}{{\textit{ring\_next}}}

\newcommand{\mN}{{\textit{r}}}

\newcommand{\mReq}{{\textit{req}}}
\newcommand{\mResp}{{\textit{resp}}}
\newcommand{\mDBReq}{{\textit{db\_req}}}
\newcommand{\mDB}{{\textit{db}}}
\newcommand{\mDBResp}{{\textit{db\_resp}}}
\newcommand{\mT}{{\textit{t}}}
\newcommand{\mMatch}{{\textit{match}}}
\newcommand{\mvU}{\mathit{u}}
\newcommand{\mvR}{\mathit{q}}

\newcommand{\mvP}{\mathit{p}}
\newcommand{\mvID}{\mathit{id}}
\newcommand{\mvN}{\mathit{n}}
\newcommand{\mvNext}{\mathit{n}_0}
\newcommand{\mvM}{\mathit{m}}

\newcommand{\nstar}{{n^*}}

\newcommand{\EPR}{\textsc{EPR}\xspace}

\ifcomments%

\newcommand{\sharon}[1]{{\textcolor{blue}{SS\@: \emph{#1}}}}
\newcommand{\oded}[1]{{\textcolor{green}{OP\@: \emph{#1}}}}
\newcommand{\neil}[1]{{\textcolor{cyan}{NI\@: \emph{#1}}}}
\newcommand{\yotam}[1]{{\textcolor{magenta}{YF\@: \emph{#1}}}}
\newcommand{\TODO}[1]{{\textcolor{red}{TODO\@: \emph{#1}}}}
\newcommand{\commentout}[1]{}
\newcommand{\OMIT}[1]{}

\else

\newcommand{\sharon}[1]{\ignorespaces}
\newcommand{\oded}[1]{\ignorespaces}
\newcommand{\neil}[1]{\ignorespaces}
\newcommand{\yotam}[1]{\ignorespaces}
\newcommand{\TODO}[1]{\ignorespaces}
\newcommand{\commentout}[1]{\ignorespaces}
\newcommand{\OMIT}[1]{\ignorespaces}

\fi

\newcommand{\remove}[1]{\xspace}

\def\instr#1{\widehat{#1}}
\def\deltainstr{\widehat{\delta}}
\def\Sigmainstr{\widehat{\Sigma}}

\def\sk#1{{#1}_{S}}

\def\bhterms#1{\textnormal{BHT}_{#1}}
\def\closedterms{\bhterms{\infty}}

\def\instantiate#1#2{{#1}[\,#2\,]}

\def\reducemodel#1#2{{#1}_{\rvert_{#2}}}
\def\consts#1{\mathop{\textrm{const}}[#1]}

\def\FOL{\mbox{{\rm FOL}}}
\def\QF{\mbox{{\rm QF}}}
\def\AF{\mbox{{\rm AF}}}
\def\ONEALT{\mbox{{\rm 1-alternation}}}
\def\deltah{\widehat{\delta}}
\def\Ih{\widehat{I}}
\def\consts#1{\mathop{\textrm{const}}[#1]}
\renewcommand{\vec}[1]{\ov{#1}}   

\newcommand{\vecx}{{\vec{x}}}
\newcommand{\vecy}{{\vec{y}}}
\newcommand{\vect}{{\vec{t}}}
\newcommand{\vecc}{{\vec{c}}}

\newcommand{\E}{{\ensuremath{\exists^*}}}
\renewcommand{\AE}{{\ensuremath{\forall^*\exists^*}}}
\newcommand{\EA}{{\ensuremath{\exists^*\forall^*}}}

\newcommand{\bnfdef}{{::=}}

\newcommand{\Cmd}{\mathit{C}}
\newcommand{\Decls}{\mathit{decls}}
\newcommand{\vars}{{\ensuremath{\mathcal{V}}}}
\newcommand{\relations}{{\ensuremath{\mathcal{R}}}}

\newcommand{\axioms}{{\ensuremath{\mathcal{A}}}}

\renewcommand{\wp}{\mathit{wp}}
\newcommand\substitute[3]{{{#1} \left[ #3~/~#2 \right]}}

\newcommand{\nnf}[1]{\mathop{\textnormal{nnf}({#1})}} 

\newcommand{\timeout}{60 seconds}

\theoremstyle{plain} 

\newtheorem{theorem}[thm]{Theorem}
\newtheorem{definition}[thm]{Definition}
\newtheorem{lemma}[thm]{Lemma}
\theoremstyle{remark}\newtheorem{example}[thm]{Example}
\newtheorem{remark}[thm]{Remark}
\newtheorem{corollary}[thm]{Corollary}

\makeatletter
\newtheorem*{rep@theorem}{\rep@title}
\newcommand{\newreptheorem}[2]{%
\newenvironment{rep#1}[1]{%
 \def\rep@title{#2~\ref{##1}}%
 \begin{rep@theorem}}%
 {\end{rep@theorem}}}
\makeatother

\newreptheorem{theorem}{Theorem}
\newreptheorem{lemma}{Lemma}

\title[Bounded Quantifier Instantiation for Checking Inductive Invariants]{Bounded Quantifier Instantiation \texorpdfstring{\\}{} for Checking Inductive Invariants}
\titlecomment{{\lsuper*}A preliminary version of this paper appeared in~\cite{originalPaper}.}

\author[Y.M.Y.~Feldman]{Yotam M. Y. Feldman\rsuper{a}}
\author[O.~Padon]{Oded Padon\rsuper{a}}
\author[N.~Immerman]{Neil Immerman\rsuper{b}}
\author[M.~Sagiv]{\texorpdfstring{\\}{}Mooly Sagiv\rsuper{a}}
\author[S.~Shoham]{Sharon Shoham\rsuper{a}}

\address{\lsuper{a}Tel Aviv University, Tel Aviv, Israel}
\email{yotam.feldman@gmail.com}

\address{\lsuper{b}UMass, Amherst, USA}

\Crefname{defiC}{Definition}{Definitions}

\begin{document}

\begin{abstract}
We consider the problem of checking whether a proposed invariant $\varphi$ expressed in first-order logic with quantifier alternation is \emph{inductive}, i.e.\ preserved by a piece of code.
While the problem is undecidable, modern SMT solvers can sometimes solve it automatically.
However, they employ powerful quantifier instantiation methods that may diverge,
especially when $\varphi$ is not preserved. A notable difficulty arises due to counterexamples of infinite size.

This paper studies \emph{Bounded-Horizon instantiation}, a natural method for guaranteeing the termination of SMT solvers. 
The method bounds the depth of terms used in the quantifier instantiation process.
We show that this method is surprisingly powerful for checking quantified invariants in uninterpreted domains.
Furthermore, by producing partial models it can help the user diagnose the case when $\varphi$ is not inductive,
especially when the underlying reason is the existence of infinite counterexamples.

Our main technical result is that Bounded-Horizon is at least as powerful as \emph{instrumentation}, which is a manual method to guarantee convergence of the solver by modifying the program so that it admits a purely universal invariant.
We show that with a bound of 1 we can simulate a natural class of instrumentations, without the need to modify the code and in a fully automatic way.
%
%
We also report on a prototype implementation on top of Z3, which we used to verify several examples by Bounded-Horizon of bound 1.
\end{abstract}







\maketitle

\section{Introduction}

This paper addresses a fundamental problem in automatic program verification:  how to prove that a piece of code
preserves a given invariant.
In Floyd-Hoare style verification this means that we want to automatically prove the validity of the
Hoare triple $\{P\} C \{P\}$
where $P$ is an assertion and $C$ is a command.
Often this is shown by proving the unsatisfiability of a formula of the form $P(V) \land \delta(V, V') \land \neg P(V')$ (the \emph{verification condition})
where $P(V)$ denotes the assertion $P$ before the command,
$P(V')$ denotes the assertion $P$ after the command, and $\delta(V, V')$ is a two-vocabulary formula expressing the meaning of
the command $C$ as a transition relation between pre- and post-states.
When 
$C$ is a loop body, 
such a $P$ is an inductive invariant and can be used to prove safety properties of the loop
(if it also holds initially and implies the desired property).

For infinite-state programs, proving the validity of $\{P\}
C \{P\}$ is generally undecidable even when $C$ does not include
loops.
Indeed, existing Satisfiability Modulo Theory (SMT) solvers can diverge even for simple
assertions and simple commands.  
Recent attempts to apply program
verification to prove the correctness of critical system's design and
code~\cite{IronFleet} identify this as the main hurdle for using
program verification.

The difficulty is rooted in powerful constructs 
used in the SMT-based verification of interesting programs.
Prominent among these constructs are arithmetic and other program operations modelled using background theories, and logical quantifiers.
In this paper we target the verification of applications in which the problem can be modelled without interpreted theories.
This is in line with recent works that show that although reasoning about arithmetic is crucial for low-level code, in many cases the verification of high-level programs and designs can be performed by reasoning about quantification in uninterpreted theories.
Specifically, the decidable Effectively Propositional logic (EPR) has been successfully applied to application domains such as linked-list manipulation~\cite{CAV:IBINS13}, Software-Defined Networks~\cite{DBLP:conf/pldi/BallBGIKSSV14} and some distributed
protocols~\cite{pldi/PadonMPSS16,paxosEpr}.
Without interpreted theories it remains to address the complications induced by the use of quantifiers, and specifically by the use of alternating universal~($\forall$) and existential~($\exists$) quantifiers.

In the presence of quantifier alternation, the solver's ability to check assertions is hindered by the following issues:
\begin{enumerate}
	\item An \emph{infinite search space of proofs} that must be explored for correct assertions.
		  A standard form of proofs with quantified formulas is \emph{instantiation}, in which the solver attempts to replace universal quantifiers by a set of ground terms. The problem of exploring the infinite set of candidates for instantiation is sometimes manifested in matching loops~\cite{simplifyJACM05}.
	\item A difficulty of \emph{finding counterexamples} for invalid assertions, notably when counterexamples may be of infinite size.
		Current SMT techniques 
		often fail to produce models of satisfiable quantified formulas~\cite{ge2009complete,DBLP:conf/cade/ReynoldsTGKDB13}. This is somewhat unfortunate since one of the main values of program verification is the early detection of flaws in designs and programs.
		The possibility of infinite counterexamples is a major complication in this task, as they are especially difficult to find.
		In uninterpreted domains, infinite counterexamples usually do not indicate a real violation and are counterintuitive to programmers, yet render assertions invalid in the context of general first-order logic (on which SMT proof techniques are based). Hence infinite counter-models pose a real problem in the verification process.
\end{enumerate}

\noindent
Previous works on EPR-based verification~\cite{CAV:IBINS13,DBLP:conf/pldi/BallBGIKSSV14,pldi/PadonMPSS16} used universally quantified invariants with programs expressed by $\exists^*\forall^*$ formulas (EPR programs)\footnote{
$\exists^{*}\forall^{*}$ transition relations can be extracted from code by existing tools for C code manipulating linked lists~\cite{CAV:IBINS13,DBLP:conf/cav/ItzhakyBRST14,KarbyshevBIRS17} and for the modeling language RML~\cite{pldi/PadonMPSS16} which is Turing-complete.
}. In that setting, checking inductive invariants is decidable, hence problems (1) and (2) do not occur. In particular, EPR enjoys the finite-model property, and so counterexamples are of finite size. EPR programs are in fact Turing-complete~\cite{pldi/PadonMPSS16}, but universal invariants are not always sufficient to express the program properties required for verification.

For example,~\cite{IronFleet} describes a client-server scenario with the invariant that ``For every
reply message sent by the server, there exists a corresponding request
message sent by a client''. (See \Cref{ex:ae-inv} for further details.) This invariant is $\AE$ and thus leads to verification conditions with quantifier alternation. This kind of quantifier alternation may lead to divergence of the solver as problems (1) and (2) re-emerge.

\commentout{
\begin{example}%
\label{ex:ae-inv}
\begin{figure}[t]%
\label{fig:client-server}
\begin{tabular}{p{0.44\linewidth} p{0.50\linewidth}}
\begin{math}
   \begin{array}{l}
   \texttt{relation} \ \mReq(u,q)
   \\
   \texttt{relation} \ \mResp(u,p)
   \\
   \texttt{relation} \ \mMatch(q,p)
   \\
   \\
   \ACTION{new\_request}{\mvU} \\
     \INDENT
     \textbf{local} \; \mvR := * \CSEP \ \COMMENT{new request} \\
     \INDENT
     \ASSUMEM{\forall w,s. \; \neg\mReq(w,\mvR) \land {} \\ \hspace{27.5mm} \neg\mMatch(\mvR,s)} \CSEP \\
     \INDENT
     \INSERTM{\mReq}{(\mvU, \mvR)} \\
     \INDENT
      \UPDATECODE{\INSERTMU{\mN}{(\mvU, y)}{\mMatch(\mvR, y)}} \\
     \ENDACTION \\
   \ACTION{respond}{\mvU, \mvR} \\
     \INDENT
     \ASSUMEM{\mReq(\mvU, \mvR)} \CSEP \\
     \INDENT
     \textbf{local} \; \mvP := * \CSEP \ \COMMENT{new response} \\
     \INDENT
     \ASSUMEM{\forall w,x. \; \neg\mReq(w,\mvP) \land {} \\ \hspace{27.5mm}\neg\mMatch(x,\mvP)} \CSEP \\
     \INDENT
     \INSERTM{\mMatch}{(\mvR, \mvP)} \CSEP \\
     \INDENT
     \UPDATECODE{\INSERTMU{\mN}{(x, \mvP)}{\mReq(x, \mvR)}} \\
     \INDENT
     \INSERTM{\mResp}{(\mvU, \mvP)} \\
     \ENDACTION
  \end{array}
\end{math}
&
\begin{math}
  \begin{array}{l}
  \texttt{init} \ \forall u,q. \, \neg\mReq(u,q)
   \\
   \texttt{init} \ \forall u,p. \, \neg\mResp(u,p)
   \\
   \texttt{init} \ \forall u,p. \, \neg\mMatch(u,p)
   \\
   \\
   \\
   \ACTION{check}{\mvU, \mvP} \\
     \INDENT
     \IF \; \mResp(\mvU, \mvP) \land \forall \mvR. \  \mReq(\mvU, \mvR) \\ \hspace{1cm} \to \neg \mMatch(\mvR,\mvP) \\
     \INDENT
     \INSTRREPLACE
     \IF \; \mResp(\mvU, \mvP) \land \neg \mN(\mvU,\mvP) \\
     \INDENT \INDENT \THEN \; \ABORT \\
     \ENDACTION \\
     \\
     \INVARIANT
     I = \\
     \INDENT \forall \mvU, \mvP. \
     \mResp(\mvU, \mvP)
     \rightarrow \\
     \INDENT \INDENT
     \exists \mvR. \  \mReq(\mvU, \mvR) \land {} \mMatch(\mvR,\mvP) \\
     \DERIVED{\mN(x,y)}{\exists z. \  \mReq(x,z) \land \mMatch(z,y)}
     \\
     \INSTRCOMMENT
     \INVARIANT
     \Ih = \\
     \INDENT\INDENT\INDENT \forall \mvU, \mvP. \;
     \mResp(\mvU, \mvP)
     \rightarrow \mN(\mvU,\mvP)
   \end{array}
 \end{math}
\end{tabular}
 \caption{ Example demonstrating a \AE\ invariant that is provable with bound 1.
   The reader should
   first ignore the instrumentation code denoted by $\INSTRCOMMENTS$ (see \Cref{example:instr}).
   \iflong%
   This example models a simple client-server scenario, with the
   safety property that every response sent by the server was
   triggered by a request from a client. Verification of this example
   requires a $\AE$ invariant.
   This example is inspired by~\cite{IronFleet}.
   \fi%
   The complete program 
   is provided
   in~\cite{additionalMaterials} (files
   \texttt{client\_server\_ae.ivy},
   \texttt{client\_server\_instr.ivy}). }
\end{figure}


IronFleet~\cite{IronFleet} describes a client server scenario  where the invariant is ``For every
reply message sent by the server, there exists a corresponding request
message sent by a client''.
\Cref{fig:client-server} presents a
simple model of this scenario and its invariant:
\[
    I = \forall \mvU, \mvP. \;
    \mResp(\mvU, \mvP) \rightarrow \exists \mvR. \; \mReq(\mvU, \mvR)
    \land \mMatch(\mvR,\mvP).
\]
This invariant is $\AE$ and thus leads to verification conditions that include quantifier alternation. This kind of quantifier alternation may lead to divergence of the underlying solver when attempting to verify the preservation of the
invariant.

\commentout{
We first explain this example ignoring the annotations denoted by
``$\INSTRCOMMENTS$''.  This example models the system state using three
binary relations. The $\mReq$ relation stores pairs of users and
requests, representing requests sent by users.  The $\mResp$ relation
similarly stores pairs of users and replies, representing replies sent
back from the server.  The $\mMatch$ relation maintains the
correspondence between a request and its reply.

The action \texttt{new\_request} models an event where a user sends
a new request to the server. The action \texttt{respond} models an
event where the server responds to a pending request by sending a
reply to the user. The request and response are related by the
$\mMatch$ relation. The action \texttt{check} is used to verify the
safety property that every response sent by the server has a matching
request, by aborting the system if this does not hold.

A natural inductive invariant for this system is
\[
    I = \forall \mvU, \mvP. \;
    \mResp(\mvU, \mvP) \rightarrow \exists \mvR. \; \mReq(\mvU, \mvR)
    \land \mMatch(\mvR,\mvP).
\]
This invariant proves that the \texttt{then} branch in the \texttt{check}
action will never happen and thus the system will never abort.
This invariant is preserved under the execution of any of the system's actions.

However, this invariant is $\AE$ and thus leads to verification conditions that include quantifier alternation. This kind of quantifier alternation may lead to divergence of the underlying solver when attempting to verify the preservation of the
invariant.
} 
\end{example}
}

This paper aims to expand the applicability of the EPR-based verification approach to invariants of more complex quantification. We focus on the class of $\forall^*\exists^*$ invariants. $\AE$ invariants arise in interesting programs, but, as we show, checking inductiveness of invariants in this class is undecidable.
We thus study problems (1), (2) above for this setting using the notion of \emph{bounded quantifier instantiation}, a technique we term \emph{Bounded-Horizon}.

\para{Main results} This paper explores the utility of limited quantifier instantiations for checking $\forall^{*}\exists^{*}$ invariants, and for dealing with the problems that arise from quantifier alternation: divergence of the proof search and infinite counter-models. 

We consider instantiations that are \emph{bounded in the depth of terms}.
Bounded instantiations trivially prevent divergence while maintaining soundness. 
Although for a given bound the technique is not complete, i.e.\ unable to prove every correct invariant,
we provide completeness guarantees by comparing bounded instantiations to the method of \emph{instrumentation}, a powerful technique implicitly employed in previous works~\cite{CAV:IBINS13,KarbyshevBIRS17,pldi/PadonMPSS16}.
Instrumentation tackles a $\AE$ invariant by transforming the program in a way that allows the invariant to be expressed using just universal quantifiers, and, accordingly, makes the verification conditions fall in EPR\@.
We show that for invariants that can be proven using a typical form of instrumentation, bounded instantiations of a small bound are also complete,
meaning they are sufficiently powerful to prove the original program without modifications and in a fully automatic way.
This is encouraging since instrumentation is labor-intensive and error-prone while bounded instantiations are completely automatic.


This result suggests that in many cases correct $\AE$ invariants of EPR programs can be proven using a simple proof technique.
Typically in such cases existing tools such as Z3 will also manage to automatically prove the verification conditions. However, bounded instantiations guarantee termination a-priori even when the invariant is not correct.
In this case, when the bounded instantiation procedure terminates, it returns a logical structure which satisfies all the bounded instantiations. This structures is not necessarily a true counterexample but ``approximates'' one.
Interestingly, this capability suggests a way to overcome the 
problem of infinite models. This problem arises
when the user provides an invariant that is correct for finite models but is not correct in general first-order logic.
In such cases, state-of-the-art SMT solvers typically produce ``unknown'' or timeout as they fail to find infinite models. The user is thus left with very little aid from the solver when attempting to make progress and successfully verify the program.
In contrast, bounded quantifier instantiation can be used to find finite models with increasing sizes, potentially indicating the existence of an infinite model, and provide hints as to the source of the error.
This information allows the user to modify the program or the invariant to exclude the problematic models.
We demonstrate this approach on a real example in which such a scenario occurred in one of our verification attempts.
We show that the provided models assist in identifying and fixing the error, allowing the user to successfully verify the program.

We also implemented a prototype tool that performs bounded instantiations of bound 1, and
used it to verify several distributed protocols and heap-manipulating programs.
The implementation efficiently reduces the problem of checking inductiveness with bound 1 to a Z3 satisfiability check on which the solver always terminates,
thereby taking advantage of Z3's instantiation techniques while guaranteeing termination.

\commentout{
\ultpara{Main Results.}

\TODO{define bounded-horizon? discuss quantifier instantiation?}
We implemented a prototype tool that performs bounded instantiations of bound 1, and
used it to verify several distributed protocols.

We show that bounded instantiations is actually quite powerful, by comparing it to the method of \emph{instrumentation}, implicitly employed in previous works~\cite{CAV:IBINS13,KarbyshevBIRS17,pldi/PadonMPSS16}.
We show that whenever it is possible to transform the program and the invariant in a way that makes the invariant universal, bounded instantiations with a low bound suffice to prove the original program without modifications and in a fully automatic way.
This is encouraging since instrumentation is labor-intensive and error-prone, while bounded instantiations are completely automatic.

We believe that the actual power of bounded instantiations for an end-user can be in situations where the invariant is not inductive. The main reason is that in such cases sophisticated quantifier-instantiation based SMT-solvers often diverge and do not produce a counterexample~\cite{ge2009complete,DBLP:conf/cade/ReynoldsTGKDB13}.
This is especially true when the counter-model to the induction has infinite size.
We demonstrate this phenomenon by showing an interesting example of a non-inductive invariant which only has infinite counter-models. In this example the user can understand the problem and correct the invariant by staring at two successive finite partial counter-models produced by bounded instantiations.

}



\para{Outline}
The rest of the paper is organized as follows:
\Cref{sec:preliminaries} provides some technical background and notations.
In \Cref{sec:bounded-horizon} we define the Bounded-Horizon algorithm and discuss its basic properties.
\Cref{sec:bounded-horizon-instrumentation} defines the concept of instrumentation as used in this work, and shows that Bounded-Horizon with a low bound is at least as powerful.
\Cref{sec:appendix-instrumentation-revisited} relates instrumentation to bounded instantiation in the converse direction, showing that other forms of instrumentation can simulate quantifier instantiation of arbitrarily high depth.
In \Cref{sec:partial-models} we show how bounded instantiations can be used to tackle the problem of infinite counterexamples to induction when the verification conditions are not valid.
\Cref{sec:implementation} describes our implementation of Bounded-Horizon of bound 1, and 
evaluates its ability to prove some examples correct by bound 1 instantiation while ensuring termination. \yotam{Sharon, added ``while ensuring termination''. This reminds that the evaluation is also important to say that the approach is applicable to many examples}
\Cref{sec:related} discusses related work, and \Cref{sec:conclusion} concludes.
The discussion of the undecidability of checking inductiveness of $\AE$ invariants is deferred to \Cref{sec:undecidability}.

\sharon{rephrased ``initial evaluation'' since a reviewer complained}


\section{Preliminaries}%
\label{sec:preliminaries}

In this section we provide background and explain our notation.

\subsection{First-Order Logic}
We use standard relational first-order logic with
equality~\cite{Shoenfield67}.
\sectionette{Syntax} A first-order \emph{vocabulary}, denoted $\vocabulary$, consists of constant symbols $c_i$, relation symbols $r_j$, and function symbols $f_k$. A \emph{relational vocabulary} is a vocabulary without function symbols. In this paper we principally use relational vocabularies, and employ function symbols only when they are generated through Skolemization (see below).
%
%
Terms and formulas are constructed according to the syntax
\begin{align*}
t &\mathop{::=} c \ | \ x \ | \ f(t_1,\ldots,t_n) \\
f & \mathop{::=} r(t_1,\ldots,t_n) \ | \ t_1 = t_2 \ | \ \neg f \ | \ f_1 \lor f_2 \ | \ f_1 \land f_2 \ | \ f_1 \rightarrow f_2 \ | \ f_1 \leftrightarrow f_2 \ | \ \forall x. \, f \ | \ \exists x. \, f
\end{align*}
where $c$ is a constant symbol, $x$ is a variable, $f$ is an $n$-ary function symbol and $r$ is an $n$-ary relation symbol. 
\sharon{added terms}
We always assume terms and formulas are well-formed.
$\FOL(\Sigma)$ stands for the set of first-order formulas over $\Sigma$.
For a formula $\varphi$ we denote by $\consts{\varphi}$ the set of constants that appear in $\varphi$.
The set of free variables in a term or a
formula is defined as usual. A term without free variables is called a
\emph{closed term} or \emph{ground term}. A formula without free variables is called a
\emph{sentence} or a \emph{closed formula}. We sometimes write
$\varphi(x_1,\ldots,x_k)$ (respectively, $t(x_1,\ldots,x_k)$) to
indicate that $x_1,\ldots,x_k$ are free in $\varphi$ (respectively,
$t$). Substitution in $\varphi$ of terms $t_1,\ldots,t_n$ instead of free variables $x_1,\ldots,x_n$ is denoted by $\varphi[t_1 / x_1, \ldots, t_n / x_n]$.

\sectionette{Semantics} Given a vocabulary $\vocabulary$, a \emph{structure} of $\vocabulary$
is a pair $\struct = (\Dom, \Int)$: $\Dom$ is a
\emph{domain}, and $\Int$ is an \emph{interpretation}.
The domain $\Dom$ is a set (of elements). If this set if finite, we say that the structure $\struct$ is finite.
The interpretation $\Int$ maps each symbol in $\vocabulary$ to its
meaning in $\struct$: $\Int$ associates each $k$-ary relation
symbol $r$ with
%
%
a relation $\Int(r) \subseteq \Dom^k$,
and associates each $k$-ary function symbol $f$ with
%
%
a function $\Int(f): \Dom^k \to \Dom$.
\newpage 

We use the standard semantics of first-order logic.  Given
a vocabulary $\vocabulary$ and a structure $\struct = (\Dom,\Int)$, a \emph{valuation} $\assgn$ maps every logical variable $x$ to an
element in $\Dom$.  We write $\struct, \assgn \models \varphi$ to denote
that the structure $\struct$ and valuation $\assgn$ \emph{satisfy}
the formula $\varphi$. We write $\struct \models \varphi$ to mean that
$\struct, \assgn \models \varphi$ for any $\assgn$, and we will
reserve this for whenever $\varphi$ is a closed formula.
%

\sectionette{Satisfiability and Validity} We say that a formula $\varphi$ is \emph{satisfiable} if there are
some $\struct$ and $\assgn$ such that $\struct,\assgn \models
\varphi$. Otherwise, we say that $\varphi$ is \emph{unsatisfiable}. We
say that a formula $\varphi$ is valid if $\struct,\assgn \models
\varphi$ for any $\struct$ and $\assgn$. Note that $\varphi$ is valid
if and only if $\neg \varphi$ is unsatisfiable. A formula
$\varphi(x_1,\ldots,x_n)$, whose free variables are $x_1,\ldots,x_n$,
is valid if and only if the closed formula $\forall x_1,\ldots,x_n
. \; \varphi(x_1,\ldots,x_n)$ is valid; it is satisfiable if and only
of the closed formula $\exists x_1,\ldots,x_n . \;
\varphi(x_1,\ldots,x_n)$ is satisfiable.
Two formulas $\varphi$ and $\psi$ are \emph{equivalent} if
$\psi \leftrightarrow \varphi$ is valid.
We say that two formulas
$\varphi$ and $\psi$ are \emph{equisatisfiable} to mean that $\varphi$
is satisfiable if and only if $\psi$ is satisfiable. Note that if two
formulas are equivalent then they are also equisatisfiable, but the
converse does not necessarily hold.

\sectionette{Syntactical Classes of Formulas}
We say that a formula is in \emph{negation normal form} if negation is
only applied to its atomic subformulas, namely $r(t_1,\ldots,t_n)$ or
$t_1 = t_2$. Every formula can be transformed to an equivalent formula
in negation normal form, and the transformation is linear in the size
of the formula. For a formula $\varphi$ we denote by $\nnf{\varphi}$ an equivalent formula in negation normal form obtained by the standard procedure, pushing negation inwards when it is applied on a quantifier or a connective. \yotam{Sharon, can you review?} \sharon{added connective} For example, the negation normal form of $\neg \exists
x,y. \; r(x,y) \land x \neq y$ is $\forall x,y. \; \neg r(x,y) \lor
x=y$.

We say that a formula is \emph{quantifier-free} if it contains no
quantifiers. $\QF(\Sigma)$ denotes the set of quantifier-free formulas over $\Sigma$.
We say that a formula is in \emph{prenex normal form} (PNF) if it is of the
form $Q_1 . \cdots Q_n . \psi$ where $\psi$ is quantifier-free, and
each $Q_i$ is either $\forall x$ or $\exists x$ for some variable $x$. Every
formula can be transformed to an equivalent formula in prenex normal
form, and the transformation is linear in the size of the formula.
For example, the prenex normal form of $\forall x. \; r(x) \to \exists
y. \; p(x,y)$ is $\forall x. \; \exists y. \; r(x) \to p(x,y)$.

\sharon{removed (repeated in next paragraph): We write that $\varphi \in \exists^*(\Sigma)$ to mean that $\varphi$ is an \emph{existential} formula defined over vocabulary $\Sigma$.
Similarly, the class of \emph{universal} formulas is denoted by $\forall^*(\Sigma)$.}

We say a formula is \emph{universally quantified} or \emph{universal},
if it is in prenex normal form and has only universal quantifiers. $\forall^*(\Sigma)$ denotes the set of universal formulas over $\Sigma$.
An \emph{existentially quantified} or \emph{existential} formula is
similarly defined, and $\exists^*(\Sigma)$ denotes the set of existential formulas over $\Sigma$.
Whenever an existential quantifier is in the
scope of a universal quantifier or vice versa, we call this a
\emph{quantifier alternation}. A formula is \emph{alternation-free} if
it contains no quantifier alternations; namely, if it is a Boolean
combination of universal and existential formulas. The set of alternation-free formulas is denoted $\AF(\Sigma)$.
With quantifier alternations present, we denote by $\exists^{*}\forall^{*}(\Sigma)$ the set of formulas over $\Sigma$ in prenex normal form where all the existential quantifiers appear before the universal ones, and $\forall^{*}\exists^{*}(\Sigma)$ for the case where all universal quantifiers appear before the existentials (in both cases this is a single quantifier alternation).

\sectionette{EPR}%
\label{sec:prelim-epr}
The effectively-propositional (EPR) fragment of first-order logic, also known as the
Bernays-Sch\"onfinkel-Ramsey class, consists of $\exists^* \forall^*(\Sigma)$ sentences, which we denote by $\EPR(\Sigma)$.
Such sentences enjoy the
\emph{small model property}%
\iflong%
; in fact, a satisfiable EPR sentence
has a model of size no larger than the number of its constants plus existential quantifiers.
\else
.
\fi
Thus satisfiability of EPR sentences is decidable~\cite{Ramsey01011930}.

\sectionette{Skolemization}
\commentout{
Given any sentence $\phi\in \FOL(\Sigma)$, we let $\sk{\phi}\in \forall^*
(\Sigma \uplus \sk{\Sigma})$ be its Skolemization, where $\sk{\Sigma}$ consists of the Skolem constant and function symbols.
\sharon{shouldn't we include the definition?}
}
%
Let $\varphi(z_1,\ldots,z_n) \in \FOL(\Sigma)$. The
\emph{Skolemization} of $\varphi$, denoted $\sk{\varphi}$, is a
universal formula over $\Sigma \uplus \sk{\Sigma}$, where
$\sk{\Sigma}$ consists of fresh constant symbols and function symbols,
obtained as follows.  We first convert $\varphi$ to negation normal
form (NNF) using the standard rules.  For every
existential quantifier $\exists y$ that appears
under the scope of the universal quantifiers $\forall x_1, \ldots,
\forall x_m$, we introduce a fresh function symbol $f_y \in
\sk{\Sigma}$ of arity $n+m$. We replace each bound occurrence of $y$
by $f_y(z_1,\ldots,z_n,x_1, \ldots, x_m)$, and remove the existential
quantifier. If $n+m = 0$ (i.e., $\varphi$ has no free variables and
$\exists y$ does not appear in the scope of a universal quantifier) a
fresh constant symbol is used to replace $y$.
It is well known that
$\sk{\varphi}\rightarrow \varphi$ is valid
and that $\sk{\varphi},\varphi$ are equi-satisfiable.




\subsection{Transition Systems and Inductive Invariants}

\sectionette{Transition Relation}
A \emph{transition relation} is a sentence $\delta$ over a vocabulary $\Sigma \uplus \Sigma'$ where $\Sigma$ is a relational vocabulary used to describe the source (or pre-) state of a transition and $\Sigma' = \{a' \mid a \in \Sigma\}$ is used to describe the target (or post-) state.

\sectionette{Inductive Invariants}
A first-order sentence $I$ over $\Sigma$ is an \emph{inductive invariant} for $\delta$ if
$I \land \delta \rightarrow I'$ is valid,
or, equivalently, if $I \land \delta \land \neg I'$ is unsatisfiable\footnote{
\iflong%
In this paper, unless otherwise stated, satisfiability and validity refer to general models and are not restricted to finite models.
Note that for EPR formulas, finite satisfiability and general satisfiability coincide.
\else
In this paper, satisfiability and validity refer to general models, not restricted to finite models.
\fi
},
 where $I'$ 
 results from substituting every constant and relation symbol in $I$ by its primed version%
\iflong%
\ (i.e. $I' \in \FOL(\Sigma')$).
\else
.
\fi
Candidate invariants $I$ are always sentences (i.e.\ ground).

\begin{remark}
Often in the literature an inductive invariant is considered with respect to a set of initial states of the system $\Init$ and a safety property $\PhiProp$, requiring from an inductive invariant $I$ also that $\Init \rightarrow I$ and $I \rightarrow \PhiProp$ are valid. We refer to this setting in \Cref{sec:undecidability}. Elsewhere in the paper we focus on checking the validity of $I \land \delta \rightarrow I'$ which is typically the difficult part: when $\Init \in \exists^{*}\forall^{*}(\Sigma)$, if $I \in \AE(\Sigma)$ then the validity of $\Init \rightarrow I$ is decidable (\Cref{sec:prelim-epr}), and $\PhiProp$ is usually part of $I$. Furthermore, checking simply the inductiveness of $(\Init \lor I) \land \PhiProp$ suffices to establish safety.
\sharon{this remark appears before mentioning EPR transition systems. Better to avoid referring to classes of these formulas if possible.}
\end{remark}

\iflong%
\sectionette{Counterexample to Induction}
Given a first-order sentence $I$ over $\Sigma$ and transition relation $\delta$ (over $\Sigma \uplus \Sigma'$),
a \emph{counterexample to induction} is a structure $\A$ (over $\Sigma \uplus \Sigma'$) s.t.\ $\A \models I \land \delta \land \neg I'$.

\TODO{R1: add paragraph on translating RML to EPR transition systems. Later refer to it when explaining instrumentation}

%

\sectionette{EPR Transition Relation}

In this paper we focus on transition relations that are EPR sentences.
%
%
%
Namely, we specify a transition relation via an EPR sentence, $\delta$, over a vocabulary $\Sigma \uplus \Sigma'$ where $\Sigma$ is a relational vocabulary used to describe the source (or pre-) state of a transition and $\Sigma' = \{a' \mid a \in \Sigma\}$ is used to describe the target (or post-) state.

\sharon{motivate why: for alternation free invariants checking inductiveness is in EPR hence decidable}

\subsection{RML:\texorpdfstring{\@}{} Relational Modeling Language with Effectively Propositional Logic}

We now review a simple imperative modeling language, a variant of the \emph{relational modeling language} (RML)~\cite{pldi/PadonMPSS16,oded_thesis}, and its translation to EPR transition relations.

\subsubsection{RML Syntax and Informal Semantics}%
\label{sec:RMLSyntax}

\Cref{Fi:RMLSyntax} shows the abstract syntax of RML\@.
RML imposes two main programming limitations: \begin{enumerate}
\item the only data structures are uninterpreted relations,
\item program conditions and update formulas have restricted quantifier structure.
\end{enumerate}

\noindent
An RML program is composed of a set of actions.
Each action consists of loop-free code, and a transition of the RML program corresponds to (non-deterministically) selecting an action and executing its code atomically.
Thus, an RML program can be understood as a single loop, where the loop body is a non-deterministic choice between all the actions. The restriction of each action to loop-free code simplifies the presentation, and it does not reduce RML's expressive power, as nested loops can always be converted to a flat loop.

\para{Declarations and states}
The declarations of an RML program define a set
of relations \relations,
a set of program variables \vars, and a set of axioms \axioms\
in the form of (closed) $\exists^* \forall^*$-formulas.

A state of an RML program is a first-order structure over the vocabulary that contains a relation symbol
for every relation in \relations, and a constant symbol for every variable in
\vars, such that all axioms are satisfied.
The values of program
variables and relations are all mutable by the program.

\newcommand{\quador}{~~~|~~~}

\begin{figure}
\arraycolsep=2pt
\[
\begin{array}{cclr}
\gV{rml} & \bnfdef &
  \gV{decls} ~;~ \gV{actions}
\\[0.5em]
\gV{decls} & ::= &
  \epsilon \quador
  \gV{decls} ; \gV{decls}
\\
& | & \texttt{relation}~\brel
\\
& | & \texttt{variable}~\bvar
\\
& | & \texttt{axiom}~\phiea
\\
& | & \texttt{init}~\varphi
\\[0.5em]
\gV{actions} & \bnfdef &
  \epsilon \quador
  \gV{actions}~;~\gV{actions}
\\
& | & \textbf{action}~\baction~\{~\gV{cmd}~\}
\\[0.5em]
\gV{cmd} & ::= & \bskip & {\scriptstyle \text{do nothing}}
\\
& | & \babort & {\scriptstyle \text{terminate-abnormally}}
\\
& | & \brel\left( \btuple{x}\right)  : = \phiqf\left( \btuple{x}\right)  & {\scriptstyle \text{quantifier-free update of relation}~\brel}
\\
& | & \bvar := * & {\scriptstyle \text{havoc of variable}~\bvar}
\\
& | & \bassume~\phiea & {\scriptstyle \text{assume}~\exists^* \forall^*~\text{formula holds}}
\\
& | &  \gV{cmd}~;~\gV{cmd} & {\scriptstyle \text{sequential composition}}
\\
& | &  \gV{cmd}~|~\gV{cmd} & {\scriptstyle \text{non-deterministic choice}}
\\
\end{array}
\]
\caption{\label{Fi:RMLSyntax}%
Syntax of RML\@.
$\brel$ denotes a relation identifier.
$\bvar$ denotes a variable identifier.
$\baction$ denotes an action identifier.
$\btuple{x}$ denotes a vector of logical variables.
$\phiqf \left( \btuple{x} \right)$ denotes a quantifier-free formula with free logical variables $\btuple{x}$.
$\phiea$ denotes a closed formula with quantifier prefix $\exists^* \forall^*$.
The syntax of terms and formulas is of first-order logic.
}
\end{figure}

\para{Actions}
An RML program is composed of a set of actions. Each action consists of a name and a loop-free body,
given by an RML command. The transition relation formula of the
whole program will be given by the disjunction of the transition relation
formulas associated with each action of the program. Effectively, this
means that each transition is a non-deterministic choice between
all the actions of the program, and that each action is executed
atomically. Below we given an intuitive description of RML commands,
and \Cref{sec:prelim:rml:semantics} presents their axiomatic semantics
and explains how to translate a command $\Cmd$ to its associated
transition relation formula $\delta[\Cmd]$.

\para{Commands}
Each command investigates and potentially updates the state of the program.
The semantics of \bskip\ and \babort\ are standard.
The command $\brel(x_1,\ldots,x_n) := \phiqf(x_1,\ldots,x_n)$ is used to update the $n$-ary relation
$\brel$ to the set of all $n$-tuples that satisfy the quantifier-free formula $\phiqf$.
For example, $\brel(x_1,x_2) := (x_1=x_2)$ updates the binary relation $\brel$ to the identity relation;
$\brel(x_1,x_2) := \brel(x_2,x_1)$ updates $\brel$ to its inverse relation;
$\brel_1(x) := \brel_2(x,\bvar)$ updates $\brel_1$ to the set of all elements that are related by $\brel_2$ to the current value (interpretation) of program variable $\bvar$.
\sharon{we sometimes use elements, sometimes values, sometimes interpretations. Can we improve? }

The havoc command $\bvar := *$ performs a non-deterministic assignment
to $\bvar$.  The $\bassume$ command is used to restrict the executions
of the program to those that satisfy the given (closed)
$\exists^* \forall^*$-formula.  Sequential composition and
non-deterministic choice are defined in the usual way.

The commands given in \Cref{Fi:RMLSyntax} are the core of RML\@.
\Cref{fig:rml-sugar} provides several useful syntactic sugars for RML which we use in the examples in this paper,
including an \bif-\bthen-\belse\ command and
convenient update commands for relations and functions.

\begin{figure}
\[
\begin{array}{ll}
\textbf{Syntactic Sugar} &
\textbf{Desugared RML}
\\
\toprule
\textbf{local}~\bvar ~:=~ * &
\mbox{\bvar{} is syntactically declared inside the current scope}
\\
&
\bvar := *
\\
\midrule
\textbf{action}~\baction\left(\bvar_1,\ldots,\bvar_n\right)\{\Cmd\} &
\textbf{action}~\baction~\{
\\& \qquad \textbf{local}~\bvar_1:=*;
\\& \qquad \ldots
\\& \qquad \textbf{local}~\bvar_n:=*;
\\& \qquad \Cmd
\\& \}
\\
\midrule
\bif~\varphi~\Cmd_1 &
\{\bassume~\varphi~;~\Cmd_1\} ~|~
\{\bassume~\neg \varphi\}
\\
\bif~\varphi~\Cmd_1~\belse~\Cmd_2 &
\{\bassume~\varphi~;~\Cmd_1\} ~|~
\{\bassume~\neg \varphi~; \Cmd_2 \}
\\
\midrule
\brel . \binsert \left( \btuple{y} ~|~ \varphi_{QF}\left( \btuple{y}\right) \right)  &
 \brel\left( \btuple{x}\right)  := \brel\left( \btuple{x}\right)  \vee \left(\btuple{x} = \btuple{y} \land \varphi_{QF}\left( \btuple{y}\right)\right)
\\
\brel . \binsert \left( \btuple{g} \right)  &
 \brel . \binsert \left( \btuple{y} ~|~ \btuple{y} = \btuple{g} \right) \\
\bottomrule
\end{array}
\]
\caption[Syntactic sugars for RML]{\label{fig:rml-sugar}%
Syntactic sugars for RML\@.
In addition to using the notations of \Cref{Fi:RMLSyntax},
g denotes a ground term, $\btuple{y}$ denotes a tuple of terms where each $y_i$ is $x_i$ or a ground term, $\btuple{g}$ denotes a tuple of ground terms,
and equality and between tuples denotes the conjunction of the component-wise equalities.
\sharon{I don't see any formal notion of ``scope'' in RML\@. So I don't understand the translation of ``local''. Perhaps it's not really a syntactic sugar (as a syntactic sugar I would add it to the declarations but this is definitely not what we want)}
}
\end{figure}

\para{Turing-completeness}
To see that RML is Turing-complete, we can encode a (Minsky) counter
machine in RML\@. Each counter $c_i$ can be encoded with a unary relation
$r_i$. The value of counter $c_i$ is the number of elements in $r_i$.
Testing for zero, incrementing, and decrementing counters can all be
easily expressed by RML commands.

\subsubsection{Axiomatic Semantics}%
\label{sec:prelim:rml:semantics}%
\label{sec:prelim:rml:tr}

We now provide a formal semantics for RML by defining a weakest
precondition operator for RML commands with respect to assertions
expressed in first-order logic, which also allows us to define the
transition relation formula of an RML command.
We start with a formal definition of program states as structures, and program assertions as formulas in first-order logic.

\para{States}
Recall that an RML program declares a set of program variables $\vars$ and relations \relations.  We define a
first-order vocabulary $\vocabulary$, that contains a relation symbol
for every relation in \relations, and a constant symbol for every variable in
\vars. A state of the program is given by
a first-order structure over $\vocabulary$.
The states of an RML program are structures of $\vocabulary$ that satisfy all
the axioms $\axioms$ declared by the program.

\para{Assertions}
Assertions on program states are specified by sentences in first-order logic over
$\vocabulary$.
A state satisfies an assertion if it satisfies it in the
usual semantics of first-order logic.

\begin{remark}
Note that program variables, modeled
as constant symbols in $\vocabulary$, should not be confused with logical variables
used in first-order formulas.
\end{remark}

\para{Weakest precondition of RML commands}

\begin{figure}
\arraycolsep=2pt
\def\arraystretch{1.3}
$\begin{array}{rcl}
\wp \left( \bskip, Q \right)   & = & Q
\\
\wp \left( \babort, Q \right)   & =  & \false
\\
\wp \left( \brel \left( \btuple{x} \right)  := \phiqf \left( \btuple{x} \right) , Q \right)  & = & \substitute{\left( \axioms \to Q \right)}{\brel \left( \btuple{s} \right) }{\phiqf \left( \btuple{s} \right) }
\\
\wp \left( \bvar := *, Q \right)   & = & \forall x.\, \substitute{\left( \axioms \to Q \right)}{\bvar}{x}
\\
\wp \left( \bassume~\phiea, Q \right) & = &
\phiea \to Q 
\\
\wp \left( \Cmd_1~;~\Cmd_2, Q \right)  & = & \wp \left( \Cmd_1, \wp \left( \Cmd_2, Q \right)  \right)
\\
\wp \left( \Cmd_1~|~\Cmd_2, Q \right)  & = & \wp \left( \Cmd_1, Q \right)  \land \wp \left( \Cmd_2, Q \right)
\end{array}
$
\caption{\label{Fi:RMLWP}
Rules for $\wp$ for RML. 
$\btuple{s}$ denotes a vector of terms.
}
\end{figure}



\Cref{Fi:RMLWP} presents the definition of a \emph{weakest
  precondition} operator for RML, denoted $\wp$.  The weakest
precondition~\cite{Dijkstra76} of a command $\Cmd$ with respect to
an assertion $Q$, denoted $\wp(\Cmd, Q)$, is an assertion $Q'$ such
that every execution of $\Cmd$ starting from a state that satisfies
$Q'$ leads to a state that satisfies $Q$. Further, $\wp(\Cmd, Q)$ is
the weakest such assertion. Namely, $Q' \implies \wp(\Cmd, Q)$ for 
every $Q'$ as above.

The rule for $\wp$ of $\bskip$ is standard, as are the
rules for $\babort$, $\bassume$, sequential composition and non-deterministic
choice. The rules for updates of relations and functions and for havoc
are instances of Hoare's assignment rule~\cite{Hoare69}, applied to
the setting of RML and adjusted for the fact that state mutations are
restricted by the axioms $\axioms$.

\para{Transition relation formulas of RML commands}

\oded{is the following reasonable? I don't want to go into too much
  detail here}

The weakest precondition operator is closely related to transition
relation formulas. Recall that a transition relation formula has
vocabulary $\vocabulary \uplus \vocabulary'$, where the primed symbols
represent the state after executing the command.
%
%
Here, we use the connection between transition relations and the weakest precondition (pointed out by~\cite{EWD:EWD821,DBLP:journals/toplas/Nelson89}), to define the transition relation of a command $\Cmd$, denoted by $\delta[\Cmd]$, as follows:
\begin{equation*}
\delta[\Cmd] = \neg \wp\left(\Cmd, \neg \psi_{\vocabulary = \vocabulary'} \right)
\end{equation*}
where
\begin{equation*}
\psi_{\vocabulary = \vocabulary'} =
\bigwedge_{r \in \vocabulary} \forall \btuple{x}. \; r(\btuple{x}) \leftrightarrow r'(\btuple{x})
\;\land\;
\bigwedge_{c \in \vocabulary} c = c'
\end{equation*}
This makes a slight abuse of the definition of $\wp$, since it applies
$\wp$ to a formula over $\vocabulary \uplus \vocabulary'$. However, in
this context, the symbols in $\vocabulary'$ can be treated as
additional auxiliary symbols, without special meaning.

Intuitively, there is a transition from $\state$ to $\state'$ if and
only if $\state$ does not satisfy the weakest precondition of ``not
being $\state'$''. This is captured by the above connection, and ``not
being $\state'$'' is captured by $\neg \psi_{\vocabulary =
  \vocabulary'}$. Note further that non-deterministic choice between
commands results in a conjunction in the weakest precondition, and a
disjunction in the transition relation. Similarly, a havoc command
results in a universal quantifier in the weakest precondition, and an
existential quantifier in the transition relation.

The transition relation of the entire RML program $\delta$ is given by the disjunction of the transition relations of each action in the RML program, where the transition relation of each
action is computed from its body via the above definition. Formally, if the bodies of
actions are $\Cmd_1,\til,\Cmd_k$ then the transition relation of the
whole program $\delta$ is given by:
\begin{equation*}
\delta = \axioms \land \bigvee_{i} \delta[\Cmd_i]
\end{equation*}
\yotam{Added axioms to the big TR, should we remove them from the wp table? Should we explain this?}
Note that when using the syntactic sugar of \Cref{fig:rml-sugar} to define
actions with parameters, these parameters are existentially quantified
in the transition relation (through the negation of the weakest-precondition rule of a havoc command). \yotam{Oded, is the text in parenthesis correct?}

\TODO{add axioms to TR}

\para{RML Produces EPR Transition Relations.}
RML is designed so that the transition relations associated with RML programs are EPR transition relations.
This is formalized in the following claims:
\begin{lemma}%
\label{lemma:rml-wp-ae}
Let $\Cmd$ be an RML command. If $Q \in \forall^*\exists^*(\Sigma)$-formula, then so is the prenex normal form of $\wp(\Cmd,Q)$.
\end{lemma}
\begin{proof}
Straightforward from the rules of \Cref{Fi:RMLWP}, the fact that all assignments to relations use quantifier-free formulas, and all $\bassume$ commands and axioms are of formulas with $\exists^*\,\forall^*$ prenex normal form.
\end{proof}

\begin{corollary}
The transition relation $\delta$ of an RML program is an EPR transition relation.
\end{corollary}
\begin{proof}
Follows from \Cref{lemma:rml-wp-ae}, the construction of $\delta[\Cmd]$ as the negation of a weakest-precondition, and the transition relation of the entire program $\delta$ as the disjunction of transition relations of individual actions.
\end{proof}

\def\inducformula{\textit{Ind}}
\def\arity#1{\text{Arity}_{#1}}
\def\instdepth#1#2{\mathop{\text{depth}}(\instantiate{#1}{#2})} 

\section{Bounded-Horizon}%
\label{sec:bounded-horizon}
In this section, we define a systematic method of quantifier instantiation called \emph{Bounded-Horizon} as a way of checking the inductiveness of first-order logic formulas,
and explore some of its basic properties.

\para{Undecidability}
We first justify the use of sound but incomplete algorithms, such as the Bounded-Horizon algorithm, for the problem of checking inductiveness of $\AE$ formulas.
For a universal sentence $I \in \forall^*(\Sigma)$, the sentence $I \land \delta \land \neg I'$ is in EPR (recall that $\delta$ is specified in EPR). Hence, checking inductiveness amounts to checking the unsatisfiability of an EPR formula, and is therefore decidable. The same holds for $I \in AF(\Sigma)$. However, this is no longer true when quantifier alternation is introduced.
In \Cref{sec:undecidability} we show that checking inductiveness of $\AE$ formulas is indeed undecidable, even when the transition relation is restricted to $\EPR$ (see \Cref{thm:undecidability-infinite}).
Thus, an attempt to check inductiveness must sacrifice either soundness, completeness, or termination.
Techniques based on quantifier instantiation usually prefer completeness over termination.
In contrast, the Bounded-Horizon algorithm guarantees termination a-priori, possibly at the expense of completeness (but is surprisingly powerful nonetheless).
We now move to define the Bounded-Horizon algorithm for checking invariants with quantifier alternation and discuss its basic properties in checking inductiveness.

\para{Bounded-Horizon Instantiations}
Let $\delta \in \exists^{*}\forall^{*}(\Sigma, \Sigma')$ be an EPR transition relation and $I \in \FOL(\Sigma)$ a candidate invariant.
We would like to check the satisfiability of $I \land \delta \land \neg I'$, and equivalently of
$\inducformula = \sk{I} \land \sk{\delta} \land \sk{(\neg I')}$.
Recall that $\sk{\phi}$ denotes the Skolemization of $\phi$, and
note that $\sk{I}$ and $\sk{(\neg I')}$ possibly add Skolem functions to the vocabulary.
($\delta$ is an EPR sentence and so its Skolemization adds only constants.)
Roughly speaking, for a given $k \in \Nat$, Bounded-Horizon 
instantiates the universal quantifiers in $\inducformula$, while restricting the instantiations to produce ground-terms of function nesting at most $k$.
\iflong%
We then check if this (finite) set of instantiations is unsatisfiable; if it is already unsatisfiable then we have a proof that $I$ is inductive.
Otherwise we report that $I$ is not known to be inductive. 
The idea is to choose a (preferably small) number $k$ and perform instantiations bounded by $k$ instead of full-blown instantiation. As we will show, this algorithm is sound but not necessarily complete for a given $k$.
\fi

Below we provide the formal definitions%
\iflong%
.
\else
\ and discuss soundness and (in)completeness. 
\fi 
We start with the notion of instantiations, and recall Herbrand's theorem which establishes completeness of proof by unrestricted instantiations.
Suppose that some vocabulary $\tilde{\Sigma}$ including constants and function symbols is understood (e.g., $\tilde{\Sigma} = \Sigma \uplus \sk{\Sigma}$, where $\sk{\Sigma}$ includes Skolem constants and function symbols).

\begin{definition}[Instantiation]
Let $\varphi(\vec{x}) \in \forall^{*}(\tilde{\Sigma})$ be a universal formula with $n$ free variables and $m$ universal quantifiers. An \emph{instantiation} of $\varphi$ by 
a tuple $\vec{t}$ of $n+m$ ground terms, denoted by $\instantiate{\varphi}{\ov{t}}$, is obtained by substituting $\ov{t}$ for the free variables and the universally quantified variables, and then removing the universal quantifiers.
\end{definition}
Note that an instantiation is a quantifier-free sentence.

\begin{theorem}[Herbrand's Theorem]%
\label{herbrandsTheorem}
Let $\varphi \in \forall^{*} (\tilde{\Sigma})$. Then $\varphi$ is satisfiable iff the (potentially infinite) set
$\left\{
		\instantiate{\varphi}{\vec{t}} \mid \vec{t} \mbox{ is a tuple of ground terms over } \tilde{\Sigma}
	\right\}$
is satisfiable.
\end{theorem}

\begin{remark}%
\label{rem:epr-complete-inst}
Herbrand's theorem provides a simple proof of the decidability of the satisfiability of EPR sentences: 
Let $\varphi \in \EPR(\Sigma)$. Then its Skolemization $\sk{\varphi}$ may introduce constant symbols but not function symbols (since there is no $\AE$ quantification). Function symbols are not present in the vocabulary, so the set of possible instantiations is finite. From Herbrand's theorem, it suffices to check the satisfiability of the finite set of quantifier-free sentences $\{\instantiate{\sk{\varphi}}{\vec{c}}\}$, which is decidable.
\end{remark}

While EPR sentences introduce only instantiations on constant symbols (bound $0$ when considering the bound of function applications), arbitrary sentences may introduce instantiations of unbounded depths. We now turn to restrict the depth of terms used in instantiations.
\begin{definition}[Bounded-Depth Terms]
	For every $k \in \Nat$, we define $\bhterms{k}$ to be the set of ground terms over $\tilde{\Sigma}$ with
	function symbols nested to depth at most $k$. $\bhterms{k}$ is defined by induction over $k$, as follows.
	Let $C$ be the set of constants in $\tilde{\Sigma}$, $F$ the set of functions, and for every $f \in F$ let $\arity{f}$ be the arity of $f$. Then
\iflong%
	\begin{align*}
				\bhterms{0} &= C 
				\\
				\bhterms{k} &= \bhterms{k - 1} \cup \{f(t_1, \ldots t_m) \mid
														f \in F, \ m = \arity{f}, \ t_1, \ldots, t_m \in \bhterms{k-1} \}.
	\end{align*}
\else
$\bhterms{0} = C$ and for $k>0$:
	\begin{align*}
				\bhterms{k} &= \bhterms{k - 1} \cup \{f(t_1, \ldots t_m) \mid
														f \in F, \ m = \arity{f}, \ t_1, \ldots, t_m \in \bhterms{k-1} \}.
	\end{align*}
\fi
\end{definition}
We will also write
$\ov{t}\in\bhterms{k}$ for a tuple of terms $\ov{t}$, to mean that every entry of $\ov{t}$ is in $\bhterms{k}$ (the number of elements in $\ov{t}$ should be clear from the context).
Note that the set of ground terms is $\closedterms = \bigcup_{k \in \Nat}{\bhterms{k}}$.

\begin{definition}[Depth of Instantiation]
	Let $\varphi \in \forall^{*}(\tilde{\Sigma})$ and $\vec{t} \in \closedterms$. The \emph{depth of instantiation}, denoted $\instdepth{\varphi}{\vec{t}}$, is the smallest $k$ such that all ground terms that appear in $\instantiate{\varphi}{\vec{t}}$ are included in $\bhterms{k}$.
\end{definition}

We are now ready to define the algorithm and discuss its basic soundness and completeness properties.

\para{Bounded-Horizon algorithm}
Given a candidate invariant $I \in \FOL(\Sigma)$, a transition relation $\delta$ over $\Sigma \uplus \Sigma'$, and $k \in \Nat$, the Bounded-Horizon algorithm constructs the formula $\inducformula = \sk{I} \sland \sk{\delta} \sland \sk{(\neg I')}$, and
checks if the set
\begin{equation}
\label{eq:bounded-horizon-def}
\left\{
		\instantiate{\inducformula}{\vec{t}}
		\ | \
		\vec{t} \in \bhterms{k}, \ \instdepth{\inducformula}{\vec{t}} \leq k
	\right\}
\end{equation}
is unsatisfiable. If it is unsatisfiable, then $I$ is provably \emph{inductive} w.r.t.\ $\delta$ with \emph{Bounded-Horizon of bound $k$}.
Otherwise we report that $I$ is \emph{not known to be inductive}.

Note that the satisfiability check performed by Bounded-Horizon is decidable since the set of instantiations is finite, and each instantiation is a ground quantifier-free formula.

%

\sectionette{Bounded-Horizon for $\forall^{*}\exists^{*}$ Invariants}
We illustrate the definition of Bounded-Horizon in the case that $I \in \forall^{*}\exists^{*}(\Sigma)$.
Let $I = \forall \vec{x}. \ \exists \vec{y}. \ \alpha(\vec{x}, \vec{y})$ where $\alpha \in \QF$.
Then $\sk{I} = \forall \vec{x}. \ \alpha(\vec{x}, \vec{f}(\vec{x}))$ where $\vec{f}$ are new Skolem function symbols.
$\sk{\delta}$ introduces Skolem constants but no function symbols, and in this case so does $\sk{(\neg I')}$.
The Bounded-Horizon check of bound $k$ can be approximately\footnote{
	\Cref{eq:bounded-horizon-approx} is an under-approximation of the set of instantiations used for bound $k$; variables that do not appear in $\sk{I}$ under a function symbol can be taken from $\bhterms{k}$ in the conjunction without increasing the total depth of instantiation beyond $k$, and are therefore allowed in bounded instantiation of bound $k$.
	This approximation is illustrative nonetheless, and will be useful in the proofs in \Cref{sec:bounded-horizon-instrumentation}.
}
understood as checking the (un)satisfiability of 
\begin{equation}
\label{eq:bounded-horizon-approx}
			\bigl(\bigwedge_{\ov{t} \in \bhterms{k-1}}{\instantiate{\sk{I}}{\ov{t}}}\bigr)
	\sland 	\bigl(\bigwedge_{\ov{t} \in \bhterms{k}}{\instantiate{\sk{\delta}}{\ov{t}}}\bigr)
	\sland 	\bigl(\bigwedge_{\ov{t} \in \bhterms{k}}{\instantiate{\sk{\bigl(\neg I'\bigr)}}{\ov{t}}}\bigr).
\end{equation}

\sharon{the following may be understood as specific to the $\forall \exists$ case. Better add some sentence before retutning to the general scope (alternatively move the paragraph on $\forall \exists$ elsewhere --- not sure where..)} \yotam{added following:}
The Bounded-Horizon algorithm is sound for all $I \in \FOL(\Sigma)$, as formalized in the next lemma:
\begin{lemma}[Soundness]%
\label{lem:bh-sound}
For every $k \in \Nat$, Bounded-Horizon with bound $k$ is sound, i.e., if 
\iflong%
Bounded-Horizon of bound $k$
\else
it
\fi
reports that $I  \in \FOL(\Sigma)$ is inductive w.r.t.\ $\delta$, then $I$ is indeed inductive.
\end{lemma}
\begin{proof}
	Assume that $I$ is not inductive w.r.t.\ $\delta$, so there is a structure $\A$ such that $\A \models \sk{I} \sland \sk{\delta} \sland \sk{(\neg I')}$.
	In particular $\A \models \instantiate{\inducformula}{\vec{t}}$ for every $\vec{t} \in \closedterms$ and in particular for every $\vec{t} \in \bhterms{k}$ such that $\instdepth{\inducformula}{\vec{t}} \leq k$. Hence, Bounded-Horizon of bound $k$ will not report that $I$ is inductive.
\end{proof}

As the algorithm is sound for any $k$, the crucial question that remains is an appropriate choice of $k$.
A small $k$ is preferable for efficiency, but a larger $k$ could allow for proving more invariants.
In the following example, a bound of even 1 suffices for proving that the invariant is inductive.
We then show that for every correct invariant there is a suitable bound $k$, but a single choice of $k$ cannot prove all correct invariants.
Later, in \Cref{sec:bounded-horizon-instrumentation}, we show that bound of $1$ or $2$ is surprisingly powerful nonetheless.

\begin{example}%
\label{ex:ae-inv}

\Cref{fig:client-server} presents a
simple model of 
the client-server scenario described in~\cite{IronFleet}.
The program induces an EPR transition relation, and its invariant is provable by Bounded-Horizon of bound 1.
\TODO{change font of action parameters to avoid collision with quantified vars}

We first explain this example while ignoring the annotations denoted by
``$\INSTRCOMMENTS$''.  
The system state is modeled using three
binary relations. The $\mReq$ relation stores pairs of users and
requests, representing requests sent by users.  The $\mResp$ relation
similarly stores pairs of users and replies, representing replies sent
back from the server.  The $\mMatch$ relation maintains the
correspondence between a request and its reply.

The action \texttt{new\_request} models an event where a user $\mvU$ sends
a new request to the server. The action \texttt{respond} models an
event where the server responds to a pending request by sending a
reply to the user. The request and response are related by the
$\mMatch$ relation. The action \texttt{check} is used to verify the
safety property that every response sent by the server has a matching
request, by aborting the system if this does not hold.

A natural inductive invariant for this system is
\[
    I = \forall \mvU, \mvP. \;
    \mResp(\mvU, \mvP) \rightarrow \exists \mvR. \; \mReq(\mvU, \mvR)
    \land \mMatch(\mvR,\mvP).
\]
The invariant proves that the \texttt{then} branch in action \texttt{check}
will never happen and thus the system will never abort.
This invariant is preserved under execution of all actions, and this fact is provable by Bounded Horizon of bound 1.

\end{example}

\begin{lemma}[Completeness for some $k$]%
\label{bh:exists-k}
\iflong%
For every $I \in \FOL(\Sigma)$ and $\delta$ such that $I$ is inductive w.r.t.\ $\delta$
\else
If $I \in \FOL(\Sigma)$  is inductive w.r.t.\ $\delta$ then
\fi
 there exists
 \iflong%
 a finite
 \fi
 $k \in \Nat$ s.t. $I$ is provably inductive w.r.t.\ $\delta$ with Bounded-Horizon of bound $k$.
\end{lemma}
\begin{proof}
From \Cref{herbrandsTheorem} and compactness there is a finite unsatisfiable set $S$ of instantiations. Take $k$ to be the maximal depth of the instantiations in $S$.
\end{proof}

For example, if $I \in \forall^{*}$ then Bounded-Horizon of bound $0$ is complete.
However, as expected due to the undecidability of checking inductiveness (see \Cref{sec:undecidability}), Bounded-Horizon is \emph{not} complete for a given $k$ for arbitrary invariants.
\iflong%
\begin{example}
\begin{figure}[t]%
\label{fig:client-server-db}
\begin{tabular}{p{0.45\linewidth} p{0.5\linewidth}}
\begin{math}
   \begin{array}{l}
   \texttt{relation} \ \mReq(u,q)
   \\
   \texttt{relation} \ \mResp(u,p)
   \\
   \texttt{relation} \ \mDBReq(id,p)
   \\
   \texttt{relation} \ \mDBResp(id,p)
   \\
   \texttt{relation} \ \mT(id,u)
   \\
   \\
   \ACTION{new\_request}{\mvU} \\
     \INDENT
     \textbf{local} \; \mvR := * \CSEP \ \COMMENT{new request} \\
     \INDENT
     \ASSUMEM{\forall w,j. \; \neg\mReq(w,\mvR) \land {}\\ \hspace{27mm} \neg\mDBReq(\mvR,j)} \CSEP \\
     \INDENT
      \INSERTM{\mReq}{(\mvU, \mvR)} \\
     \ENDACTION \\
   \ACTION{db\_recv\_request}{\mvID, \mvR} \\
     \INDENT
     \ASSUMEM{\mDBReq(\mvID, \mvR)} \CSEP \\
     \INDENT
     \textbf{local} \; \mvP := * \CSEP \\
     \INDENT
     \ASSUMEM{\mDB(\mvR, \mvP)} \CSEP \\
     \INDENT
     \INSERTM{\mDBResp}{(\mvID, \mvP)} \\
     \ENDACTION \\
   \ACTION{check}{\mvU, \mvP} \\
     \INDENT
     \ASSUMEM{\mResp(\mvU, \mvP)} \CSEP \\
     \INDENT
     \IF \; \forall \mvR. \  \mReq(\mvU, \mvR) \to \neg \mDB(\mvR,\mvP) \\
     \INDENT \INDENT \THEN \; \ABORT \\
     \ENDACTION
  \end{array}
 \end{math}
  &
  \begin{math}
  \begin{array}{l}
  \texttt{init} \ \forall u,q. \, \neg\mReq(u,q)
  \\
  \texttt{init} \ \forall u,p. \, \neg\mResp(u,p)
  \\
  \texttt{init} \ \forall id,p. \, \neg\mDBReq(id,p)
  \\
  \texttt{init} \ \forall id,p. \, \neg\mDBResp(id,p)
  \\
  \texttt{init} \ \forall id,u. \, \neg\mT(id,u)
  \\
  \\
     \ACTION{server\_recv\_request}{\mvU, \mvR} \\
     \INDENT
     \ASSUMEM{\mReq(\mvU, \mvR)} \CSEP \\
     \INDENT
     \textbf{local} \, \mvID := * \CSEP \ \COMMENT{new DB request id}  \\
     \INDENT
     \ASSUMEM{\forall w. \; \neg\mT(\mvID, w)} \CSEP \\
     \INDENT
     \INSERTM{\mT}{(\mvID, \mvU)} \CSEP \\
     \INDENT
     \INSERTM{\mDBReq}{(\mvID, \mvR)} \\
     \ENDACTION
  \\
  \\
  \\
   \ACTION{server\_recv\_db\_response}{\mvID, \mvP} \\
     \INDENT
     \ASSUMEM{\mDBResp(\mvID, \mvP)} \CSEP \\
     \INDENT
     \INSERTMU{\mResp}{(x, \mvP)}{\mT(\mvID, x)} \\
     \ENDACTION
   \end{array}
 \end{math}
\\
\multicolumn{2}{c}{
 \begin{math}
 \begin{array}{l}
     \INVARIANT
     I =
     \forall \mvU, \mvP. \;
     \mResp(\mvU, \mvP) \to \exists \mvR. \;
     \mReq(\mvU, \mvR) \land \mDB(\mvR,\mvP) \sland \\
     %
     \hspace{1cm}
     \forall \mvID, \mvR. \;
     \mDBReq(\mvID, \mvR) \to \exists \mvU. \;
     \mT(\mvID, \mvU) \land \mReq(\mvU,\mvR) \sland \\
     \hspace{1cm}
     \forall \mvID, \mvP. \;
     \mDBResp(\mvID, \mvP) \to \exists \mvR. \;
     \mDBReq(\mvID, \mvR) \land \mDB(\mvR, \mvP) \sland \\
     \hspace{1cm}
     \forall \mvID, \mvU_1, \mvU_2. \;
     \mT(\mvID, \mvU_1) \land \mT(\mvID, \mvU_2) \to \mvU_1 = \mvU_2
 \end{array}
 \end{math}
 }

 \end{tabular}
 \caption{Example demonstrating a $\AE$ invariant that is provable
   only with bound 2.
   \iflong%
   The server anonymizes requests from clients to the database (DB) and forwards the answer to the client.
   The server performs a translation $\mT$ between clients' identity and an anonymous unique id.
   The safety property is that every response sent by the server to a client was
   triggered by a request from the client.
   The inductive invariant further states that every server request to the DB was triggered by a client's request from the server, and that every DB response was triggered by a server's request.
   \fi
   The complete program corresponding to this Figure appears in~\cite{additionalMaterials} (file
   \texttt{client\_server\_db\_ae.ivy}).
   }
\end{figure}


An example of a program and an inductive invariant 
for which a bound of 0 or 1 is insufficient appears in \Cref{fig:client-server-db}.
In this example the server operates as a middleman between clients and the database (DB), and is used to anonymize user requests before they reach the database.
The server performs a translation $\mvP$ between clients' identity and an anonymous unique id, sends a translated request to the DB, and forwards the DB's response to the clients.
The safety property is that every response sent by the server was
triggered by a request from a client. The inductive invariant states, in addition to the safety property, that every server request to the DB was triggered by a client's request from the server, and that every DB response was triggered by a server's request.
Proving that the invariant is inductive under the action \texttt{server\_recv\_db\_response} requires the prover to understand that for the response from the DB there is a matching request from the server to the DB, and that for this request to the DB there is a matching request from the client to the server. Every such translation requires another level of nesting in the instantiation.
In this example, a bound of 2 manages to prove inductiveness.
This example can be lifted to require an even larger depth of instantiation by adding more translation entities similar to the server, and describing the invariant in a similar, modular, way.
\end{example}
\else
An example for which a bound of 1 is insufficient appears in the extended version~\cite{extendedVersion}.
\fi

\sectionette{Small Bounded-Horizon for $\forall^{*}\exists^{*}$ Invariants}
Despite the incompleteness, we conjecture that  a small depth of instantiations typically suffices to prove inductiveness.
%
The intuition is that an EPR transition relation has a very limited ``horizon'' of the domain:
it interacts only with a small fraction of the domain, namely elements pointed to by program variables (that correspond to logical constants in the vocabulary).

When performing the Bounded-Horizon check with bound 1 on a $\forall^* \exists^*$ invariant $I = \forall \vec{x}. \ \exists \vec{y}. \ \alpha(\vec{x}, \vec{y})$, we essentially assume that the existential part of the invariant $\psi(\vec{x}) = \ \exists \vec{y}. \ \alpha(\vec{x}, \vec{y})$ holds on all program variables --- but not necessarily on all elements of the domain --- and try to prove that it 
holds on all elements of the domain after the transition.
We expect that for most elements of the domain, the correctness of $\psi$ 
is maintained simply because they were not modified at all by the transition.
For elements that are modified by the transition,
we expect the correctness after modification to result from
the fact that $\psi$ holds for the elements of the domain that are directly involved in the transition. 
If this is indeed the reason that $\psi$ 
is maintained, a bound of 1 sufficiently utilizes
$\psi$ 
in the pre-state to prove the invariant in the post-state, i.e.\ to prove that it is inductive.

This is the case in \Cref{ex:ae-inv}.
Additional examples are listed in \Cref{sec:implementation}.
\iflong%
The example of \Cref{fig:client-server-db} itself also admits a different invariant that is provable by bound 1.
\Cref{sec:bounded-horizon-instrumentation} further studies the power of Bounded-Horizon with a low bound.
\fi


\section{Power of Bounded-Horizon for Proving Inductiveness}%
\label{sec:bounded-horizon-instrumentation}
We now turn to investigate the ability of Bounded-Horizon to verify inductiveness.
In this section we provide sufficient conditions for its success by relating it to the notion of instrumentation (which we explain below).
We show that Bounded-Horizon with a low bound of 1 or 2 is as powerful as a natural class of sound program instrumentations, those that do not add existential quantifiers.
\Cref{sec:implementation} demonstrates the method's power on several interesting programs that we verified using Bounded-Horizon of bound 1.



\subsection{Instrumentation}%
\label{sec:instrumentation}
\sharon{slightly rephrased:}
In this section we present our view of an instrumentation procedure, a form of which was used in previous works~\cite{CAV:IBINS13,KarbyshevBIRS17,pldi/PadonMPSS16},
aiming to eliminate the need for quantifier-alternation, thus reducing the
verification task to a decidable fragment.
Generally speaking, instrumentation begins with a program that induces an EPR transition relation $\delta \in \EA(\Sigma \cup \Sigma')$. The purpose of instrumentation is to modify $\delta$ into another transition relation $\deltainstr$
that admits an inductive invariant with simpler quantification (e.g., universal, in which case it is decidable to check) in a sound way.
Instrumentation is generally a manual procedure.
We now describe the instrumentation procedure used in previous works~~\cite{CAV:IBINS13,KarbyshevBIRS17,pldi/PadonMPSS16},
but stress that the results of this paper do not depend on this specific recipe but on the semantic soundness condition below (\Cref{def:sound-instrumentation}).
This instrumentation procedure is also thoroughly described in a recent work~\cite{paxosEpr}.

The instrumentation procedure used previously~\cite{CAV:IBINS13,KarbyshevBIRS17,pldi/PadonMPSS16} consists of the following three steps:
\begin{enumerate}
  \item\label{it:identify}
  Identify a formula $\psi(\vecx) \in
    \FOL(\Sigma)$ (usually $\psi$ will be existential) that captures
    information that is needed in the inductive invariant. Extend the
    vocabulary with an \emph{instrumentation relation} $r(\vecx)$ that intentionally should capture the
    derived relation defined by $\psi(\vecx)$. Let
    $\Sigmainstr = \Sigma \cup \{r\}$ denote the extended vocabulary.

  \item\label{it:update}
  Add update code that updates $r$ when the original (``core'')
    relations are modified, and maintains the meaning of $r$ as
    encoding $\psi$. The update code must 
    not block executions of real code,
    and can possibly be a sound approximation. Sometimes 
    it can be generated
    automatically via finite differencing~\cite{TOPLAS:RepsSL10}.

  \item\label{it:rewrite}
    Modify the program to use $r$. Often this is performed by rewriting some  program conditions, keeping in mind that $r$ encodes $\psi$. This means replacing some quantified expressions by uses of $r$.


\end{enumerate}
\TODO{R!: relate to RML programs. Try to give intuition as to when such instrumentation is applicable}

\begin{example}%
\label{example:instr}
In the example of \Cref{fig:client-server}, to achieve a universal invariant
we add an instrumentation relation $\mN$ defined by $\mN(x,y) \equiv \exists z. \ \mReq(x,z) \land \mMatch(z,y)$ (step~\ref{it:identify}).
The simple form of $\psi$ allows us to obtain precise update code, which appears as annotations marked with $\INSTRCOMMENTS$ in lines that mutate $\mReq$ and $\mMatch$ (step~\ref{it:update}).
We also replace the $\IF$ condition in the action
\texttt{check} by an equivalent condition that uses
$\mN$ (step~\ref{it:rewrite}). The line marked with $\INSTRREPLACE$ in the \texttt{check}
action replaces the line above it.
The resulting program has the invariant $\Ih = \forall \mvU, \mvP. \; \mResp(\mvU, \mvP) \rightarrow \mN(\mvU,\mvP)$, which is universal.
\end{example}

 Let $\deltainstr \in \EA(\Sigmainstr \cup \Sigmainstr')$ denote the
transition relation induced by the modified program (modifications occur in steps~\ref{it:update},\ref{it:rewrite}).
The soundness of the instrumentation procedure is formalized in the
following connection between $\psi$, $\delta$, and $\deltainstr$:

\begin{definition}[Sound Instrumentation]%
\label{def:sound-instrumentation}%
\label{remark:bh-sound-instrumentation-substitution}
  $\deltainstr \in \EA(\Sigmainstr \cup \Sigmainstr')$ is a
  \emph{sound instrumentation} for $\delta \in \EA(\Sigma
  \cup \Sigma')$ and $\psi \in \FOL(\Sigma)$ if
    \begin{equation*}
      \delta \rightarrow \deltainstr[\psi / r, \psi' / r']
    \end{equation*}
  is valid.
\end{definition}

\TODO{R2: sound instrumentation is for $\delta$ and $\psi$. Fix and say that in the sequel we omit it}

\Cref{def:sound-instrumentation} requires that the instrumented program
includes at least all the behaviors of the original program, when $r$ is interpreted according to $\psi$.
Thus, if the instrumented program is safe, then it is sound to infer
that the original program is safe. 
The subtle point in instrumentation as opposed to, e.g., ghost code, is that instrumentation may affect the executions, for example by changing conditions in the code.
Soundness ensures that no executions are omitted.

\begin{example}%
\label{example:sound-intr-client-server}
In the example of \Cref{fig:client-server}, the instrumentation described in \Cref{example:instr} is a sound instrumentation, where the transition relation of the original program $\delta$ and that of the instrumented program $\deltainstr$ are produced from the example's code as in \Cref{sec:prelim:rml:tr}.
To see that  $\deltainstr$  forms a sound instrumentation for $\delta$ and $\psi(x,y) = \exists z. \  \mReq(x,z) \land \mMatch(z,y)$, consider a transition of $\delta$.  This induces a transition of $\deltainstr$ by interpreting the instrumentation relation $r$ according to $\psi$: the update code of $r$ in the instrumented program translates to a restriction in $\deltainstr$ relating $r,r'$, which holds since the code updates $r$ according to its meaning as $\psi$. Furthermore, if the transition of $\delta$ is of the action \texttt{check} and the condition of the $\IF$ statement holds, then $\deltainstr[\psi / r, \psi' / r']$ allows the matching transition---of that same action, with the same action parameters (formally, a valuation of the existential quantifiers in $\deltainstr$ as the valuation of the existentials for the action paramters in $\delta$). This is due to the fact that the rewritten $\IF$ statement in the instrumented program is equivalent to the original when interpreting $r$ as $\psi$.
\end{example}

\iflong%
\begin{remark}
Note that the definition of sound instrumentation ensures that $r$ is updated in a way that is consistent with its interpretation as $\psi$.
To see this, note that in the expression $\deltainstr[\psi / r, \psi' / r']$ the update code of $r$ in $\deltainstr$ becomes a constraint over the core relations in $\Sigma$. In a sound instrumentation this constraint is required to follow from the way the core relations are updated by $\delta$, essentially implying that the update code is correct.
\end{remark}
\fi

The instrumentation procedure does not require the user to know
an inductive invariant for the original program.
However, if a sound instrumentation which leads to an 
invariant exists,
then an inductive
invariant for the original $\delta$ can be produced by substituting
back the ``meaning'' of $r$ as $\psi$ (thus,  safety of the original program is implied):


\begin{lemma}%
\label{lem:original-sub-inductive}
Let $\deltainstr$ be a sound instrumentation for $\delta$ and $\psi$,
and $\Ih \in \FOL(\Sigmainstr)$  be an inductive invariant for $\deltainstr$.
Then $I = \Ih[\psi / r]$ is inductive w.r.t.\ $\delta$.
\end{lemma}
\begin{proof}
$\Ih \land \deltah \rightarrow \Ih'$ is valid,
thus, so is
$(\Ih \land \deltah \rightarrow \Ih')[\psi/r,\psi'/r']$.
$\deltah$ is a sound instrumentation for $\delta$, so (using \Cref{remark:bh-sound-instrumentation-substitution})
${I} \land {\delta} \rightarrow {I}'$ is valid.
\end{proof}
Note that typically the quantification structure of $I$ is more complex than that of $\Ih$.

\begin{example}
In the example of \Cref{fig:client-server}, the instrumented program has the inductive invariant $\Ih = \forall \mvU, \mvP. \; \mResp(\mvU, \mvP) \rightarrow \mN(\mvU,\mvP)$, which is universally quantified.
Substituting $\psi$ instead of $r$ we get an inductive invariant of the original program, $I = \forall \mvU, \mvP. \; \mResp(\mvU, \mvP) \rightarrow \exists \mvR. \  \mReq(\mvU, \mvR) \land \mMatch(\mvR,\mvP)$. This invariant is $\AE$.
\end{example}

\begin{remark}
In this paper we focus on instrumentation by a derived relation. It is also possible to consider instrumentations by a \emph{constant}, as used for example in handling the alternation-free invariant of the shared-tail example in~\cite{KarbyshevBIRS17}.
Instrumentation by a constant can be emulated with an instrumentation by a derived relation, conforming to the results of this paper.
This is performed by adding a unary relation $c(x)$ representing the constant $c$, and adding to the invariant a clause stating that $c(x)$ holds for exactly one element. The resultant invariant is alternation-free, and thus \Cref{bhinstr:af-small-bound} can account for the power of bounded instantiations in this case.
\end{remark}



%
%
%
%

\para{Instrumentation without additional existential quantifiers}
In order to relate instrumentation to Bounded-Horizon instantiations, we consider the typical case where the instrumentation process of $\delta$ does not add new existential quantifiers to $\deltainstr$. This happens when 
the update code does not introduce additional existential quantifiers.
To formally define this notion (\Cref{instrumentationNoNewExists}) we first define \emph{existential naming} to relate the existential quantifiers of two transition relations:
\begin{definition}[Existential Naming]%
  \label{def:existentialNaming}
  Let $\deltainstr = \exists z_1, \ldots z_m. \ \varphi (z_1, \ldots, z_m)$ where $\varphi \in \forall^{*}(\Sigmainstr,\Sigmainstr')$.
  An \emph{existential naming} $\eta$ for $(\deltainstr,\delta)$ is a mapping $\eta: \{z_1, \ldots, z_m\} \to \consts{\sk{\delta}} \cup \consts{{\deltainstr}}$.
  We define $\eta(\deltainstr)$ to be $\varphi[\eta(z_1) / z_1, \ldots, \eta(z_m) / z_m]$.
\end{definition}
An existential naming provides a Skolemization procedure for $\deltainstr$ which uses existing constants rather than fresh ones.
The existential naming fixes existentially quantified variables to constants of ${\deltainstr}$, constants of $\delta$, or existential quantifiers of $\delta$ (manifested as constants from $\sk{\delta}$).
Without further requirements, such a mapping $\eta$ always exists. However, we are interested in mappings such that
$\eta(\deltainstr)$ is a sound over-approximation of $\delta$, despite fixing the existential quantifiers of $\deltainstr$ according to $\eta$ (note that fixing the quantifiers makes the formula stronger, or, when viewing it as an over-approximation of $\delta$, tighter). Intuitively, this means that $\eta$ fixes the existential quantifiers in a sound way.
When such a mapping exists, we refer to the corresponding instrumentation as instrumentation without additional existentials:

\begin{definition}[Instrumentation Without Additional Existentials]%
 \label{instrumentationNoNewExists}
 $\deltainstr$ is a \emph{sound instrumentation without additional existentials} for $\delta$ and $\psi$ if there exists an existential naming $\eta$ such that
  \begin{equation*}
    \sk{\delta} \rightarrow\eta(\deltainstr)[\psi / r,\psi' / r']
  \end{equation*} is valid.
\end{definition}

\Cref{instrumentationNoNewExists} ensures that the
existential quantifiers in $\deltainstr$ have counterparts in (the Skolemized) $\delta$, which are identified by $\eta$, and suffice to establish soundness of $\deltainstr$.
\begin{example}
The instrumentation in \Cref{fig:client-server} results in $\deltainstr$ whose soundness can be established with an existential naming w.r.t.\ the original $\delta$. The existnential naming is as follows: existential quantifiers in $\deltainstr$ result (per the procedure in \Cref{sec:prelim:rml:tr}) only from action parameters and the havoc statements (see \Cref{fig:rml-sugar}).\footnote{Another potential source of existential quantifiers is existentially quantified {\bassume} statements, which are not present in the instrumented program.}
The existential naming maps these to the \emph{same} (Skolemized) action parameters and variables in the original program. Note that the update code of $r$ uses quantifier-free updates and does not utilize \emph{additional} existential quantifiers. A proof of soundness with the interpretation \sharon{should it be ``above naming'' instead of ``interpretation''?}  of the existential quantifiers of $\deltainstr$ is just as in \Cref{example:sound-intr-client-server}: for every transition of the original program there is a matching transition of the instrumented program when $r$ is interpreted according to $\psi$ \emph{and} the existential quantifiers of the instrumented program are interpreted to match the existential quantifiers of the original program as described here.
Therefore, $\deltainstr$ is a sound instrumentation without additional existentials.
\end{example}
Note that it is possible that $\deltainstr$ has in fact \emph{fewer} existential quantifiers than $\delta$, for example due to the rewriting of conditions (as happens in the example of \Cref{fig:client-server} --- see the \IF\ statement in action \ACTIONNAME{check}).


%

\subsection{From Instrumentation to Bounded-Horizon}
The results described in this section show that if there is an instrumentation without additional existentials, then Bounded-Horizon with a low bound is able to prove the original invariant, without specific knowledge of the instrumentation and without manual assistance from the programmer.
This is the case in the example of \Cref{fig:client-server}, which admits an instrumentation that transforms the invariant to a universal invariant (see \Cref{example:instr})
in a form that matches \Cref{bounded-horizon-psi-exists-positive}, and indeed the original invariant is provable by Bounded-Horizon of bound 1.

Interestingly, in case Bounded-Horizon with a small bound does not prove inductiveness 
\iflong%
\else
(see the example in the extended version~\cite{extendedVersion}),
\fi
the results imply that either the invariant is not inductive or \emph{no instrumentation} that does not add existential quantifiers can be used to
show that it is inductive (even with the programmer's manual assistance). 
\iflong%
This is the case in the example of \Cref{fig:client-server-db}, where a bound of 1 does not suffice.\footnote{Strictly speaking this shows that there is no such instrumentation where the instrumentation relation appears only positively in the invariant, which is the most common case. Examples that require an even larger bound (sketched above) do not admit any instrumentation without additional existential quantifiers that transforms the invariant to a universal form.}
\fi

While we show that instrumentation that does not add existentials is at most as powerful as Bounded-Horizon with a low bound, sound instrumentations that do add existentials to the program (thereby not satisfying \Cref{instrumentationNoNewExists}) can be used to simulate quantifier instantiation of an arbitrary depth. This topic is explored in \Cref{sec:appendix-instrumentation-revisited}.

In the remainder of this section we will assume that $\deltainstr$ is a sound instrumentation without additional existentials for $\delta$, and $\eta$ is the corresponding naming of existentials.
Further, we assume that $\Ih$ is an inductive invariant for $\deltainstr$ and denote $I = \Ih[\psi / r]$.
\commentout{
Soundness of $\deltainstr$ immediately implies that $I$ 
is an inductive invariant
for $\delta$ (\Cref{lem:original-sub-inductive}).
The additional condition that $\deltah$ does not add new existential quantifiers ensures that this fact is provable by Bounded-Horizon with a low bound.
}

\iflonglong%
\ultpara{Results.}
We now state the results whose proofs are presented in the rest of this section.
\Cref{bounded-horizon-psi-exists-positive} and \Cref{univ:psi-af-bound-2} consider $I \in \AE(\Sigma)$ that is transformed to $\Ih \in \Univ(\Sigmainstr)$.
In \Cref{bounded-horizon-psi-exists-positive} we show that a bound of 1 suffices to prove that $I$ is inductive for $\delta$ when $\psi \in \exists^{*}(\Sigma)$ (that is, the instrumentation defining formula is existential) and the instrumentation relation $r$ appears only positively in $\Ih$, or when $\psi \in \Univ(\Sigma)$ and $r$ appears only negatively in $\Ih$.
This is an attempt to explain the success of bound 1 instantiations in proving our examples (see \Cref{sec:implementation}).
In \Cref{univ:psi-af-bound-2} we show that a bound of 2 suffices in the more general setting of $\psi \in \AF(\Sigma)$ (with no restriction on appearances of $r$ in $\Ih$).

\Cref{bhinstr:af-small-bound} considers a generalization to $I$ that is 1-alternation and transformed to $\Ih \in \AF(\Sigmainstr)$. We show that a bound of 2 suffices in this case.

\ultpara{Proof idea.}
The rest of the section is devoted to proofs of these claims.
The idea of the proof concentrates around the instantiations necessary to prove inductiveness of the instrumented and original invariants.
To highlight the main points in the formal proof, the crux of the argument is as follows.
\begin{enumerate}
  \item Assume for the sake of contradiction that bounded instantiations of a low bound on the \emph{original} invariant $I$ do not suffice to prove it inductive w.r.t.\ the \emph{original} program $\delta$, and take a counterexample to induction of the instantiated $I$ (see \Cref{alpha-instant} in the proof of \Cref{bounded-when-instr:main-lemma-universal}).

  \item Exploiting properties of substitution, connect instantiations of the original and of the instrumented invariants through the assumption that $\deltah$ is an instrumentation without additional existentials for $\delta$ 
  to obtain a counterexample to induction for the instantiated $\Ih$ w.r.t.\ $\deltah$ (see \Cref{move-to-deltainstr}).

  \item Rely on the assumption that the instrumented invariant $\Ih$ is universal and $\deltah \in \EPR(\Sigmainstr)$. By \Cref{rem:epr-complete-inst}, this means that instantiations of bound 0---namely, with just the constants---suffice to prove $\Ih$ inductive w.r.t.\ the instrumented program. 
  In other words, a counterexample to induction obtained for the \emph{instantiated} $\Ih$ w.r.t.\ $\deltah$ is a \emph{true} counterexample to induction of $\Ih$ w.r.t.\ $\deltah$ (see \Cref{consec-instr}), in contradiction to the premise.
\end{enumerate}

\noindent
In essence, the proof translates an instantiation-based proof of the instrumented invariant to a proof of the original invariant by \emph{the same set of terms} instantiating the universal quantifiers; the set of constants, sufficient for the instrumented invariant, must thus be sufficient also for the original invariant, where this amounts to \Cref{eq:bounded-horizon-approx}, thus constituting a proof by bound 1 instantiations. \sharon{the last sentence is not clear in my opinion. Can we refer to \Cref{eq:bounded-horizon-approx} to explain the ``this amounts'' (or is it too specific to some of the proofs)? how can the set of constants amount to bound 1? it is only because you bring back the exists, right? i.e., forget about skolemization } \yotam{added the reference, me very like}

The formal proofs handle the fine details of relating between the universal quantifiers and the constants of the original and instrumented invariant, to complete the transformation of instantiations between them.
\fi

\begin{remark}%
\label{rem:multi}
The results of this section also apply when multiple instrumentation relations $\psi_1, \ldots, \psi_t \in \FOL(\Sigma)$ are simultaneously substituted for the relation symbols $r_1, \ldots, r_t$ in $\deltah$ and $\Ih$.
\end{remark}

\iflonglong%
\subsection{Power for \texorpdfstring{$\forall^* \exists^*$}{forall*/exists*} Invariants}%
\label{sec:ae-inst-implies-bound}

We now establish that low bounds are sufficient for the Bounded-Horizon check, assuming that a sound instrumentation without additional existentials exists, in the case of $\instr{I} \in \forall^*(\Sigmainstr)$ and $I \in \forall^* \exists^*(\Sigma)$.
To do so, we first prove the following lemma.

\begin{lemma}%
\label{bounded-when-instr:main-lemma-universal}
    Let $\deltainstr $ be a sound instrumentation of $\delta,\psi$
        without new existentials and with naming $\eta$, and let $\instr{I} \in \forall^*(\Sigmainstr)$ be an inductive invariant for $\deltainstr$.
    Write $\instr{I} = \forall
        \vec{x}. \ \instr{\alpha}(\vec{x})$ where $\instr{\alpha} \in
        \QF(\Sigmainstr)$ and let $\alpha = \nnf{\instr{\alpha}[\psi / r]}$.
	 Then,
	\begin{equation}
	\bigl(\bigwedge_{\ov{c} \in C^n}{\hspace*{-1ex}\alpha(\ov{c})}\bigr) \sland \sk{\delta} \sland \sk{(\neg I')}
	\end{equation}
	is unsatisfiable,
	where $C = \consts{\sk{\delta} \land \sk{(\neg I')}}$ and $n$ is the number of universal quantifiers in~$\instr{I}$.
\end{lemma}

\TODO{SHARON\@: R2\@: also for $\psi$?}

\begin{proof}
	Assume not, so there exists a structure ${\A}_0$ satisfying \Cref{eq:bounded-horizon-approx}, namely
	\begin{equation}
	\label{alpha-instant}
		{\A}_0 \models \bigl(\bigwedge_{\vec{c} \in C^n}{\hspace*{-1ex}\alpha(\vec{c})}\bigr) \sland
                \sk{\delta} \sland \sk{(\neg I')}\; .
	\end{equation}
	We will show that $\instr{I}$  is not inductive for $\deltah$.
	Let $\instr{C} = \consts{\eta(\deltah) \land \sk{(\neg \instr{I}')}}$. Then,
	\begin{equation}\label{alpha3-instant}
		{\A}_1 \models\;
			\bigl(\hspace*{-2ex}\bigwedge_{\vec{c} \in {(C \cup \instr{C})}^n}{\hspace*{-2ex}\alpha(\vec{c})}\bigr)
				 \sland \sk{\delta} \sland \sk{(\neg I')}
	\end{equation}
	where ${\A}_1$ is the same as ${\A}_0$ but also interprets any constant in $\instr{C}
        \setminus C$ as the interpretation of some  	arbitrary constant in $C$.
        Thus $\alpha(\vec{c})$ holds in ${\A}_1$ for the new constants as well.

	Removing some conjuncts from \Cref{alpha3-instant}, we get,
	\begin{equation}
		{\A}_1 \models
			\bigl(\bigwedge_{\vec{c} \in \instr{C}^n}{\hspace*{-1ex}\alpha(\vec{c})}\bigr)
				 \sland \sk{\delta} \sland \sk{(\neg I')}\; .
	\end{equation}

	By assumption (\Cref{instrumentationNoNewExists}), it follows that,
	\begin{equation}
	\label{move-to-deltainstr}
		{\A}_1 \models
			\bigl(\bigwedge_{\vec{c} \in \instr{C}^n}{\alpha(\vec{c})}\bigr)
				 \sland \eta(\deltah)[\psi / r,\psi' / r'] \sland \sk{(\neg I')}.
	\end{equation}

Recall that $I' = \Ih'[\psi'/r']$. Since $\A_1 \models \sk{(\neg \Ih'[\psi'/r'])}$, it follows that
$\A_1\models \sk{(\neg \Ih')}[\psi'/r']$.  In the latter formula, some existentially quantified variables
from $\psi$ or
$\lnot \psi$ may remain, whereas in the former formula they were replaced by Skolem constants. Thus
this is just a corollary of the fact that $\sk{\gamma} \rightarrow \gamma$ is valid for any $\gamma$.

Thus we have shown (recalling that $\alpha = \nnf{\instr{\alpha}[\psi / r]}$ and $\instr{\alpha}[\psi / r]$ are equivalent),
	\begin{equation}
		{\A}_1 \models \Bigl(\bigl(\bigwedge_{\vec{c} \in \instr{C}^n}{\instr{\alpha}(\vec{c})}\bigr)
			 \sland \eta(\deltah) \sland \sk{(\neg \instr{I}')}\Bigr)[\psi / r,\psi' / r'].
	\end{equation}
Now, consider the structure $\instr{{\A}}$ that expands ${\A}_1$ by interpreting $r$ and $r'$ the way that ${\A}_1$ interprets $\psi$ and $\psi'$, respectively. Then,
	\begin{equation}
	\label{instantiations-instrumentation}
		\instr{{\A}} \models
			\bigl(\bigwedge_{\vec{c} \in \instr{C}^n}{\instr{\alpha}(\vec{c})}\bigr)
				 \sland \eta(\deltah) \sland \sk{(\neg \instr{I}')}.
	\end{equation}

	The formula in \Cref{instantiations-instrumentation} is a universal sentence: $\instr{\alpha}(\vec{c})$ is quantifier free and closed, $\eta(\deltah) \in \Univ$ from the definition of existential naming, and $\neg \instr{I}' \in \exists^{*}$ and thus its Skolemization introduces only constants.
  It follows that the formula in \Cref{instantiations-instrumentation} is also satisfied by $\reducemodel{\instr{{\A}}}{\instr{C}}$, the substructure of $\instr{{\A}}$ with universe $\instr{C}^{\instr{\A}}$, i.e., $\instr{\A}$'s interpretation of the constant symbols $\instr{C}$; recall that $\instr{C} = \consts{\eta(\deltah) \land \sk{(\neg \instr{I}')}}$ so indeed $\instr{{\A}}$ provides an interpretation to every constant in the formula.
  Thus,
	\begin{equation}
	\label{move-to-universal-quantifiers}
		\reducemodel{\instr{{\A}}}{\instr{C}} \models
			\bigl(\forall \vec{x}. \ \instr{\alpha}(\vec{x})\bigr) \sland \eta(\deltah) \sland \sk{(\neg \instr{I}')}.
	\end{equation}

	Finally, since $\sk{\gamma} \rightarrow \gamma$ is valid and so is $\eta(\deltah) \rightarrow \deltah$ (for the same reasons), we know,
	\begin{equation}
	\label{consec-instr}
	\reducemodel{\instr{{\A}}}{\instr{C}} \models
		\instr{I} \sland \deltah \sland \neg \instr{I}'.
	\end{equation}
But this contradicts the fact that $\widehat{I}$ is inductive for $\deltah$.
\end{proof}

The following results are corollaries of \Cref{bounded-when-instr:main-lemma-universal}.
\else
The following theorems state our results for $I \in \AE$.

\begin{theorem}%
\label{bounded-horizon-psi-exists-positive}
	Let $\Ih \in \forall^*(\Sigmainstr)$ be an inductive invariant for $\deltainstr$, which is a sound instrumentation for $\delta$ without additional existentials.
  Assume $\psi \in \exists^{*}$ and $r$ appears only positively in $\instr{I}$, or $\psi \in \forall^{*}$ and $r$ appears only negatively in $\instr{I}$. Then $I = \instr{I}[\psi / r]$ is inductive for $\delta$ with Bounded-Horizon of bound $1$. (Note that $I \in \forall^* \exists^*$.) 
\end{theorem}
{
\renewcommand{\proofname}{Proof Sketch}
\begin{proof}
	Let $I = \forall \vec{x}. \ \alpha(\vec{x})$ where $\alpha \in \exists^*$.
	Assume for the sake of contradiction that $I$ is not inductive for $\delta$ with Bounded-Horizon of bound 1.
	By the assumptions on $\psi$ and $\instr{I}$, this means that there is a structure $\A$ such that
	\[
\textstyle
		\A \models \bigl(\bigwedge_{\vec{c}}{\alpha(\vec{c})}\bigr) \sland
		                \sk{\delta} \sland \sk{(\neg I')}\; .
	\]
	From the assumption (\Cref{instrumentationNoNewExists}) and properties of Skolemization, it follows that
	\[
\textstyle
		\A \models
			\bigl(\bigwedge_{\vec{c}}{\alpha(\vec{c})}\bigr)
				 \sland \bigl(\eta(\deltah)\bigr)[\psi / r,\psi' / r'] \sland \bigl(\sk{(\neg \instr{I}')}\bigr)[\psi/r, \psi'/r']\; .
	\]
	From the assumptions on the way $\psi$ appears in $\instr{I}$, when we write $\instr{I} = \forall \vec{x}. \ \instr{\alpha} (\vec{x})$ where $\instr{\alpha} \in \QF$ we have $\alpha = \instr{\alpha}[\psi/r]$.
	Thus, from properties of substitution (interpreting $r,r'$ according to $\psi,\psi'$ in $\A$) it follows that there is a structure $\instr{A}$ such that
	\[
\textstyle
		\instr{{\A}} \models
			\bigl(\bigwedge_{\vec{c}}{\instr{\alpha}(\vec{c})}\bigr)
				 \sland \eta(\deltah) \sland \sk{(\neg \instr{I}')}.
	\]
	By reducing $\instr{\A}$'s domain to the constants we have that
    $
		\bigl(\forall \vec{x}. \ \instr{\alpha}(\vec{x})\bigr) \land \eta(\deltah) \land \sk{(\neg \instr{I}')}
	$
	is satisfiable. (This is a use of complete instantiation for universal formulas.)

	This in turn implies (by properties of Skolemization) that
	$\instr{I} \land \deltainstr \land \neg\instr{I}'$ is satisfiable, which is a contradiction to the assumption that $\Ih$ is inductive for $\deltainstr$.
\end{proof}
}
\fi

\iflonglong%
\begin{theorem}\label{bounded-horizon-psi-exists-positive}
	Let $\Ih \in \forall^*(\Sigmainstr)$ be an inductive invariant for $\deltainstr$, which is a sound instrumentation for $\delta,\psi$ without additional existentials.
  Assume $\psi \in \exists^{*}(\Sigma)$ and $r$ appears only positively in $\instr{I}$, or $\psi \in \forall^{*}(\Sigma)$ and $r$ appears only negatively in $\instr{I}$. Then $I = \instr{I}[\psi / r]$ is inductive for $\delta$ with Bounded-Horizon of bound $1$. (Note that $I \in \forall^* \exists^*(\Sigma)$.) 
\end{theorem}
\begin{proof}
	Let $\instr{I} = \forall \vec{x}. \ \instr{\alpha} (\vec{x})$ where $\instr{\alpha} \in \QF$.
	In both cases of the claim ${\alpha} = \nnf{\instr{\alpha}[\psi / r]} \in \exists^{*}$, and so all the universal quantifiers in $I$ (more accurately, in $\nnf{I}$) are those of $\instr{I}$. This implies that the satisfiability check of \Cref{bounded-when-instr:main-lemma-universal} is simply the Bounded-Horizon satisfiability check with bound $1$, and it shows that the result must be unsatisfiable.

	More formally, assume for the sake of contradiction that $I$ is not inductive w.r.t.\ $\delta$ with Bounded-Horizon of bound $1$.
	Let $\alpha(\vec{x}) = \exists y_1, \ldots, y_m. \ \theta(\vec{x}, y_1, \ldots, y_m)$ where $\theta \in \QF$, and let
	\[
	\sk{\alpha}(\vec{x}) = \theta(\vec{x}, f_1(\vec{x}), \ldots, f_m(\vec{x}))
	\]
	be its Skolemization with fresh Skolem function symbols $f_1, \ldots, f_m$ (introduced for $y_1, \ldots, y_m$, respectively).
	Then there is a structure ${\A}$ satisfying \Cref{eq:bounded-horizon-approx}, which reads
	\begin{equation}
  \label{eq:thm-bnd-1-inst}
 				\bigl(\bigwedge_{\ov{t} \in \bhterms{0}}{\instantiate{\sk{\alpha}}{\ov{t}}}\bigr)
		\sland 	\bigl(\bigwedge_{\ov{t} \in \bhterms{1}}{\instantiate{\sk{\delta}}{\ov{t}}}\bigr)
		\sland 	\bigl(\bigwedge_{\ov{t} \in \bhterms{1}}{\instantiate{\sk{\bigl(\neg I'\bigr)}}{\ov{t}}}\bigr).
	\end{equation}
	Since $\sk{\alpha}$ has no universal quantifiers, the instantiation is just a substitution of the free variables, and ${\A}$ satisfies
	\begin{equation}
 				\bigl(\bigwedge_{\ov{t} \in \bhterms{0}}{\sk{\alpha}(\ov{t})}\bigr)
		\sland 	\bigl(\bigwedge_{\ov{t} \in \bhterms{1}}{\instantiate{\sk{\delta}}{\ov{t}}}\bigr)
		\sland 	\bigl(\bigwedge_{\ov{t} \in \bhterms{1}}{\instantiate{\sk{\bigl(\neg I'\bigr)}}{\ov{t}}}\bigr).
	\end{equation}
    By reducing ${\A}$ to the elements pointed to by $\bhterms{1}$ terms we have that
	\begin{equation}
  \label{eq:bh-univ-reduce-with-funcs}
		\reducemodel{{\A}}{\bhterms{1}} \models
			\bigl(\bigwedge_{\ov{t} \in \bhterms{0}}{\sk{\alpha}(\ov{t})}\bigr) \sland \sk{\delta} \sland \sk{\bigl(\neg I'\bigr)}
	\end{equation}
  Note that in $\reducemodel{{\A}}{\bhterms{1}}$ the interpretations of the Skolem functions are possibly partial functions.
  The functions appear in the formula of \Cref{eq:bh-univ-reduce-with-funcs} as closed terms, and applied on $\bhterms{0}$, and these cases the interpretations of the functions are defined.
  (In particular, they can be extended to total functions in an arbitrary way, and the resulting structure still satisfies \Cref{eq:bh-univ-reduce-with-funcs}.)

    We now move from the Skolem functions back to existential quantifiers.
	By the valuation that to every existentially quantified variable $y_i$ in $\alpha$ assigns the interpretation of $f_i(\ov{t})$ in $\reducemodel{{\A}}{\bhterms{1}}$ (recall that $f_i(\ov{t})$ appears in $\sk{\alpha}$ instead of the quantifier $\exists y_i$ in $\alpha$), we know that
    \begin{equation} \label{eq:bh-proof}
		\reducemodel{{\A}}{\bhterms{1}} \models
			\bigl(\bigwedge_{\ov{t} \in \bhterms{0}}{{\alpha}(\ov{t})}\bigr) \sland \sk{\delta} \sland \sk{\bigl(\neg I'\bigr)}.
	\end{equation}
    The set of constants referred to by $\bhterms{0}$ is the set of constants in \Cref{eq:thm-bnd-1-inst}, which is $\consts{\sk{\delta} \land \sk{\bigl(\neg I'\bigr)}}$. \TODO{is this right?}
    Therefore, \Cref{eq:bh-proof} can be rewritten as
	\begin{equation}
		\reducemodel{{\A}}{\bhterms{1}} \models
			\bigl(\bigwedge_{\ov{c} \in C^n}{\alpha(\ov{c})}\bigr) \sland \sk{\delta} \sland \sk{\bigl(\neg I'\bigr)}
	\end{equation}
	where $C^n = \consts{\sk{\delta} \land \sk{\bigl(\neg I'\bigr)}}^n$ and $n$
    is the number of universal quantifiers in $I$ (and $\instr{I}$).


	By \Cref{bounded-when-instr:main-lemma-universal} this is a contradiction to the assumption that $\instr{I}$ is inductive w.r.t.\ $\deltah$, and the claim follows.
\end{proof}
\fi

\begin{theorem}\label{univ:psi-af-bound-2}
	Let $\Ih \in \forall^{*}$. If $\psi \in \AF(\Sigma)$ then $I = \instr{I}[\psi / r]$ is inductive for $\delta$ with Bounded-Horizon of bound $2$. (Note that $I \in \forall^* \exists^*(\Sigma)$.)
\end{theorem}
\iflonglong%
\begin{proof}
	As before, let $\instr{I} = \forall \vec{x}. \ \instr{\alpha} (\vec{x})$ where $\instr{\alpha} \in \QF$, and let
    $\alpha = \nnf{\instr{\alpha}[\psi / r]}$. 
    \yotam{Please review --- used to assume a CNF form of $\alpha$, especially affecting \Cref{bounded-horizon-psi-af-2-reduced} and \Cref{bounded-horizon-psi-af-2-with-exists}}
    Since $\instr{\alpha} \in \QF$ and $\psi \in \AF$, $\alpha$ is a positive Boolean combination of formulas of the form $\forall \vec{v_i} \theta_{i,1}(\vec{x}, \vec{v_i})$ and $\exists \vec{z_j} \theta_{j,2}(\vec{x}, \vec{z_j})$ where $\vec{v_i} \cap \vec{z_j} = \emptyset$ for all $i,j$.
    In $\sk{\alpha}$, each formula $\exists \vec{z_j} \theta_{j,2}(\vec{x}, \vec{z_j})$ is replaced by $\theta_{j,2}(\vec{x}, \vec{g_i}(\vec{x}))$ where $\vec{g_i}$ are fresh Skolem function symbols. (Note that $\vec{x}$ is free in $\alpha$.)
\yotam{Replaced everywhere $\alpha = \instr{\alpha}[\psi / r]$ by the negation normal form. I think it's more correct (obviously subtle). See the added remarks in explanations ``thus we have shown''}
  \commentout{
	Write $\alpha(\vec{x})$ in CNF form, where each clause is alternation-free:
	\begin{align}\begin{split}
	\alpha(\vec{x}) =
		&\bigl(\forall \vec{v_1} \theta_{1,1}(\vec{x}, \vec{v_1}) \lor \exists \vec{z_1} \theta_{1,2}(\vec{x}, \vec{z_1})\bigr)
		\land
		\ldots
		\\
		\land
		&\bigl(\forall \vec{v_r} \theta_{r,1}(\vec{x}, \vec{v_r}) \lor \exists \vec{z_r} \theta_{r,2}(\vec{x}, \vec{z_r})\bigr)
	\end{split}\end{align}
	where $\theta_{1,1},\theta_{1,2},\ldots,\theta_{r,1},\theta_{r,2} \in \QF$.
  }

	Assume for the sake of contradiction that $I$ is not inductive w.r.t.\ $\delta$ with Bounded-Horizon of bound $2$.


	For brevity denote
	\def\transtorefute#1{{\xi}(#1)} 
	\[
	\transtorefute{k} =			\bigl(\bigwedge_{\ov{t} \in \bhterms{k}}{\instantiate{\sk{\delta}}{\ov{t}}}\bigr)
						\land 	\bigl(\bigwedge_{\ov{t} \in \bhterms{k}}{\instantiate{\sk{\bigl(\neg I'\bigr)}}{\ov{t}}}\bigr).
	\]

	By the assumption that inductiveness is not provable using Bounded-Horizon of bound $2$,
	\begin{equation}
	\label{bounded-horizon-2}
			\bigl(\bigwedge_{\ov{t} \in \bhterms{1}}{\instantiate{\sk{\alpha}}{\ov{t}}}\bigr)
		\land
			\transtorefute{2}
	\end{equation}
	is satisfiable by a structure ${\A}$.

	In particular
  \begin{equation}
  \label{bounded-horizon-psi-af-2-reduced}
  \bigl(\bigwedge_{\substack{\vec{c} \in \bhterms{0}, \\ \vec{d_1}, \ldots, \vec{d_r} \in \bhterms{1}}}{\instantiate{\sk{\alpha}}{\ov{t}}}\bigr)
    \land
      \transtorefute{1}
  \end{equation}
  \commentout{
	\begin{align}\label{bounded-horizon-psi-af-2-reduced}\begin{split}
			\bigwedge
						\Big{(} 
								&\bigl(\theta_{1,1}(\vec{c}, \vec{d_1}) \lor \theta_{1,2}(\vec{c}, \vec{g_1}(\vec{c}))\bigr)
								\land
								\\
								\ldots
								&\land
								\bigl(\theta_{r,1}(\vec{c}, \vec{d_r}) \lor \theta_{r,2}(\vec{c}, \vec{g_r}(\vec{c}))\bigr)
						 \Big{)} 
			\land \transtorefute{1}
	\end{split}\end{align}
  }
	is satisfied by ${\A}$ (the arity of $\vec{d_i}$ is understood from the number of universal quantifiers in the respective universal term).
	The reason that \Cref{bounded-horizon-psi-af-2-reduced} is satisfied by ${\A}$ is that \Cref{bounded-horizon-psi-af-2-reduced} differs from \Cref{bounded-horizon-2} by having \emph{fewer} conjuncts, as
		the Bounded-Horizon check with bound $2$ has conjuncts for each $\vec{c} \in \bhterms{1}$ and not just $\bhterms{0}$, and
		we take $\transtorefute{1}$ instead of $\transtorefute{2}$ ($\bhterms{1} \subseteq \bhterms{2}$ so the conjuncts of $\transtorefute{1}$ are included in those of $\transtorefute{2}$).

	Reduce ${\A}$ to the elements pointed by $\bhterms{1}$ terms, let $\A^{\downarrow} = \reducemodel{{\A}}{\bhterms{1}}$.

	Now,
  \begin{equation}
  \label{bounded-horizon-psi-af-2-with-exists}
  \bigl(\bigwedge_{\vec{c} \in \bhterms{0}}{\instantiate{\alpha}{\ov{t}}}\bigr)
    \land
      \transtorefute{1}
  \end{equation}
  \commentout{
	\begin{align}\label{bounded-horizon-psi-af-2-with-exists}\begin{split}
			\bigwedge_{\vec{c} \in \bhterms{0}}
						\Big{(} 
							&\bigl(\forall \vec{v_1} \theta_{1,1}(\vec{c}, \vec{v_1}) \lor \exists \vec{z_1} \theta_{1,2}(\vec{c}, \vec{z_1})\bigr)
							\land
							\ldots
							\\
							\land
							&\bigl(\forall \vec{v_r} \theta_{r,1}(\vec{c}, \vec{v_r}) \lor \exists \vec{z_r} \theta_{r,2}(\vec{c}, \vec{z_r})\bigr)
						 \Big{)} 
		\land \transtorefute{1}
	\end{split}\end{align}
  }
	is satisfied by $\A^{\downarrow}$.
	This is because:
	\begin{itemize}
		\item The universal quantifiers are semantically equivalent to a conjunction over all $\bhterms{1}$ elements because the domain was reduced, and
		\item The existential quantifiers are justified by the following valuation: the valuation assigns every $\vec{z_i}$ the interpretation of $\vec{g_i}(\vec{c})$.
	\end{itemize}
	With this valuation the conjunctions of formula~\ref{bounded-horizon-psi-af-2-with-exists} are all guaranteed by the conjunctions in formula~\ref{bounded-horizon-psi-af-2-reduced}.

	Now formula~\ref{bounded-horizon-psi-af-2-with-exists} exactly means that
	\begin{equation}
		\A^{\downarrow} \models
				\bigl(\bigwedge_{\ov{c} \in \bhterms{0}}{\alpha(\ov{c})}\bigr)
				\land \transtorefute{1}.
	\end{equation}
	As in the proof of \Cref{bounded-horizon-psi-exists-positive}, using \Cref{bounded-when-instr:main-lemma-universal}, this is a contradiction to the assumption that $\instr{I}$ is inductive w.r.t.\ $\deltah$, and the claim follows.
\end{proof}
\fi

\commentout{
\begin{remark}
It is possible to define the concept of bounded-horizon more semantically. \TODO{(Does anything here make any sense?)}
	We say that $\mathcal{H}$ is a partial Herbrand model of depth $k$ if the domain of $\mathcal{H}$ contains $\bhterms{k}$ and the interpretation of every function symbol may be a partial function but for every $t \in \bhterms{k}$, $t^{\mathcal{H}} = t$.

	An invariant $I$ is inductive using \emph{Semantic Bounded-Horizon of bound $k$} if there is no partial Herbrand structure of depth $k$, $\mathcal{H}$, such that $\mathcal{H}$ partially models $\sk{I} \land \sk{\delta} \land \sk{(\neg I')}$. \TODO{(Partial models --- something like in Sofronie-Stokkermans?)}

	In this formulation the result here is that if $\psi \in \AF$ then $I$ is inductive for $\delta$ with Semantic Bounded-Horizon of \emph{bound 1} (rather than $2$), because the proof essentially shows that a structure $\A$ reduced to have the domain $\bhterms{1}$ cannot be a partial counterexample-to-induction for the original invariant, because it would lead to a counterexample-to-induction for the original invariant.

	Nonetheless we prefer to use the proof-theoretic definition throughout the paper to elucidate the connection to quantifier instantiation.
\end{remark}
}

\iflonglong%
\subsection{Generalization to 1-Alternation Invariants}
We now generalize the results of \Cref{sec:ae-inst-implies-bound} to \emph{1-alternation invariants}. A formula is 1-alternation if it can be written as a Boolean combination of $\forall^* \exists^*$ formulas.
In the sequel, $\Ih \in \AF(\Sigmainstr)$ and $I = \Ih[\psi /r]\in \ONEALT(\Sigma)$. 
\else
The following theorem generalizes the above result to \emph{1-alternation invariants}. A formula is 1-alternation if it can be written as a Boolean combination of $\forall^* \exists^*$ formulas.
\fi

\iflonglong%
\begin{lemma}%
\label{bounded-when-instr:main-lemma-af}
	Let $\psi \in \FOL(\Sigma)$.
	Let $\delta \in \EPR(\Sigma \uplus \Sigma')$ and let $\deltainstr \in \EPR(\Sigmainstr \uplus \Sigmainstr')$ be a sound instrumentation of $\delta,\psi$.
	Let $\instr{I} \in \AF(\Sigmainstr)$ be an inductive invariant for $\deltainstr$, 
    and write $\sk{\Ih} = \forall \ov{x}. \ \instr{\alpha_1}(\ov{x})$ and
    $\sk{(\neg \Ih)} = \forall \ov{x}. \ \instr{\alpha_2}(\ov{x})$, where $\instr{\alpha_1}, \instr{\alpha_2}$ are quantifier free. 
    Let $\alpha_1 = \nnf{\instr{\alpha_1}[\psi / r]}$ and $\alpha_2 = \nnf{\instr{\alpha_2}[\psi / r]}$. 
     \quad Then,
	\begin{equation}
				\bigl(\bigwedge_{\ov{c_1} \in C^n}{\hspace*{-1ex}\alpha_1(\ov{c_1})}\bigr)
		\sland 	\sk{\delta}
		\sland	\bigl(\bigwedge_{\ov{c_2} \in C^m}{\hspace*{-1ex}\alpha_2(\ov{c_2})}\bigr)
	\end{equation}
	is unsatisfiable,
	where $C = \consts{\sk{\Ih} \land \sk{\delta} \land \sk{(\neg I')}}$, $n$ is the number of universal quantifiers in $\sk{\Ih}$ and $m$ is the number of universal quantifiers in $\sk{(\neg \Ih)}$.
\end{lemma}
\begin{proof}

	The proof is similar to that of \Cref{bounded-when-instr:main-lemma-universal}; this proof follows the same reasoning to transform satisfiable formulas, but transforms not only the quantifier and conjunctions related to the invariant in the pre-state, but also to its negation in the post-state.

	Assume not, i.e., there exists a structure ${\A}_0$ such that,
	\begin{equation}
	{\A}_0 \models
			\bigl(\bigwedge_{\ov{c_1} \in C^n}{\hspace*{-1ex}\alpha_1(\ov{c_1})}\bigr)
		\sland 	\sk{\delta}
		\sland	\bigl(\bigwedge_{\ov{c_2} \in C^m}{\hspace*{-1ex}\alpha_2(\ov{c_2})}\bigr).
	\end{equation}

	We will show that $\instr{I}$  is not inductive for $\deltah$.\quad
	Let $\instr{C} = \consts{\sk{\Ih} \land \eta(\deltah) \land \sk{(\neg \instr{I}')}}$. Then,
	\begin{equation}
		{\A}_1 \models\;
			\bigl(\bigwedge_{\ov{c_1} \in \instr{C}^n}{\hspace*{-1ex}\alpha_1(\ov{c_1})}\bigr)
		\sland 	\sk{\delta}
		\sland	\bigl(\bigwedge_{\ov{c_2} \in \instr{C}^m}{\hspace*{-1ex}\alpha_2(\ov{c_2})}\bigr).
	\end{equation}
	where ${\A}_1$ is the same as ${\A}_0$ but also interprets any constant in $\instr{C}
    \setminus C$ as some arbitrary constant in $C$.

   	By the assumption (\Cref{instrumentationNoNewExists}), it follows that,
   	\begin{equation}
		{\A}_1 \models\;
			\bigl(\bigwedge_{\ov{c_1} \in \instr{C}^n}{\hspace*{-1ex}\alpha_1(\ov{c_1})}\bigr)
		\sland 	\eta(\deltah)[\psi / r,\psi' / r']
		\sland	\bigl(\bigwedge_{\ov{c_2} \in \instr{C}^m}{\hspace*{-1ex}\alpha_2(\ov{c_2})}\bigr).
	\end{equation}
	where $\eta$ is the existential naming between $\delta,\deltah$.

	Thus we have shown (recalling that $\alpha_1 = \nnf{\instr{\alpha_1}[\psi / r]}$ and $\instr{\alpha_1}[\psi / r]$ are equivalent, and similarly for $\alpha_2$),
	\begin{equation}
		{\A}_1 \models
				\Bigl(		\bigl(\bigwedge_{\vec{c_1} \in \instr{C}^n}{\instr{\alpha_1}(\vec{c_1})}\bigr)
					 \sland \eta(\deltah)
					 \sland \bigl(\bigwedge_{\vec{c_2} \in \instr{C}^m}{\instr{\alpha_2}(\vec{c_2})}\bigr)
				\Bigr)[\psi / r,\psi' / r'].
	\end{equation}
	Now, consider the structure $\instr{{\A}}$ that expands ${\A}_1$ by interpreting $r$ and $r'$ the way that ${\A}_1$ interprets $\psi$ and $\psi'$, respectively. Then,
	\begin{equation}
	\label{af:instantiations-instrumentation}
		\instr{{\A}} \models
					\bigl(\bigwedge_{\vec{c_1} \in \instr{C}^n}{\instr{\alpha_1}(\vec{c_1})}\bigr)
			 \sland \eta(\deltah)
			 \sland \bigl(\bigwedge_{\vec{c_2} \in \instr{C}^m}{\instr{\alpha_2}(\vec{c_2})}\bigr)
	\end{equation}

The formula in \Cref{af:instantiations-instrumentation} is a universal sentence: $\instr{\alpha_1}(\vec{c}), \instr{\alpha_2}(\vec{c})$ are quantifier free and closed, and $\eta(\deltah) \in \Univ$.
  It follows that the formula in \Cref{af:instantiations-instrumentation} $\reducemodel{\instr{{\A}}}{\instr{C}}$, the substructure of $\instr{{\A}}$ with universe $\instr{C}^{\instr{\A}}$, i.e., $\instr{\A}$'s interpretation of the constant symbols $\instr{C}$; recall that $\instr{C} = \consts{\sk{\Ih} \land \eta(\deltah) \land \sk{(\neg \instr{I}')}}$ (and $\sk{\Ih} = \forall \ov{x}. \ \instr{\alpha_1}(\ov{x})$ and
    $\sk{(\neg \Ih)} = \forall \ov{x}. \ \instr{\alpha_2}(\ov{x})$) so indeed $\instr{{\A}}$ provides an interpretation to every constant in the formula.
Thus,
	\begin{equation}
	\label{af:move-to-universal-quantifiers}
		\reducemodel{\instr{{\A}}}{\instr{C}} \models
					\bigl(\forall \vec{x}. \ \instr{\alpha_1}(\vec{x})\bigr)
			\sland 	\eta(\deltah)
			\sland 	\bigl(\forall \vec{x}. \ \instr{\alpha_2}(\vec{x})\bigr).
	\end{equation}
	Recall that $\instr{C}$ was defined as $\instr{C} = \consts{\Ih \land \eta(\deltah) \land \sk{(\neg \instr{I}')}}$.

	Finally, since $\sk{\gamma} \rightarrow \gamma$ is valid and for the same reasons $\eta(\deltah) \rightarrow \deltah$ is valid, we know,
	\begin{equation}
	\label{af:consec-instr}
	\reducemodel{\instr{{\A}}}{\instr{C}} \models
		\instr{I} \sland \deltah \sland \neg \instr{I}'.
	\end{equation}
	But this contradicts the fact that $\widehat{I}$ is inductive for $\deltah$.
\end{proof}

The following result is a corollary of \Cref{bounded-when-instr:main-lemma-af}.
\fi
\begin{theorem}\label{bhinstr:af-small-bound}
	Let $\Ih \in \AF(\Sigmainstr)$ an inductive invariant for $\deltainstr$ which is a sound instrumentation of $\delta$ without additional existentials. If $\psi \in \AF(\Sigma)$ then $I = \instr{I}[\psi / r]$ is inductive for $\delta$ with Bounded-Horizon of bound $2$. (Note that $I \in \ONEALT(\Sigma)$.)
\end{theorem}
\iflonglong%
\begin{proof}
	$\psi \in \AF$ implies that $\alpha_1(\ov{x}) = \nnf{\instr{\alpha_1}[\psi / r]} \in \AF$ and $\alpha_2(\ov{x}) = \nnf{\instr{\alpha_2}[\psi / r]} \in AF$ (recall that $\instr{\alpha_1},\instr{\alpha_2} \in \QF$).

	By the assumption that inductiveness is not provable using Bounded-Horizon of bound~$2$,
	\begin{equation}
	\label{af:bounded-horizon-2}
				\bigl(\bigwedge_{\ov{t} \in \bhterms{1}}{\instantiate{\sk{(\alpha_1)}}{\ov{t}}}\bigr)
		\land
				\bigl(\bigwedge_{\ov{t} \in \bhterms{2}}{\instantiate{\sk{\delta}}{\ov{t}}}\bigr)
		\land 	\bigl(\bigwedge_{\ov{t} \in \bhterms{1}}{\instantiate{\sk{(\alpha_2)}}{\ov{t}}}\bigr)
	\end{equation}
	is satisfiable, where $\sk{(\alpha_1)}, \sk{(\alpha_2)}$ introduce Skolem functions. Let ${\A}$ be such a satisfying structure, and $\A^{\downarrow} = \reducemodel{{\A}}{\bhterms{1}}$.

	Because $\alpha_1 \in \AF$, in the same way as in the proof of \Cref{univ:psi-af-bound-2},
	\begin{equation}
		\A^{\downarrow} \models
			\bigwedge_{\ov{c} \in \bhterms{0}}{\alpha_1(\ov{c})}
	\end{equation}
	and in the same way, since $\alpha_2 \in \AF$ as well, the same structure has
	\begin{equation}
		\A^{\downarrow} \models
			\bigwedge_{\ov{c} \in \bhterms{0}}{\alpha_2(\ov{c})}.
	\end{equation}
	Note that from \Cref{af:bounded-horizon-2}, $\A^{\downarrow} \models \deltainstr$. Overall we have
	\begin{equation}
		\A^{\downarrow} \models
					\bigl(\bigwedge_{\ov{c} \in \bhterms{0}}{\alpha_1(\ov{c})}\bigr)
			\sland	\sk{\delta}
			\sland	\bigl(\bigwedge_{\ov{c} \in \bhterms{0}}{\alpha_2(\ov{c})}\bigr)
	\end{equation}
	and by \Cref{bounded-when-instr:main-lemma-af} this is a contradiction to the assumption that $\Ih$ a is inductive w.r.t.~$\deltainstr$.
\end{proof}
\fi


\section{Instrumentation for High Depth Instantiations}%
\label{sec:appendix-instrumentation-revisited}

In this section we discuss the connection between quantifier
instantiation and program instrumentation in the converse direction,
i.e.\ simulating quantifier instantiation by the process of instrumentation.
In \Cref{sec:bounded-horizon-instrumentation} we showed that instrumentation without adding existential quantifiers is at most as powerful as bounded instantiations with a low bound.
In this section we show that allowing additional existentials does increase the power of instrumentation in proving $\AE$ invariants.
In particular, we show how to systematically construct instrumented programs in a way
that corresponds to quantifier instantiation for $\AE$-invariants:
performing the required instantiations within the program allows expressing the invariant in a form that falls in the decidable fragment.
Together with \Cref{sec:bounded-horizon-instrumentation}, this makes the point that, in the context of invariant checking, the form of instrumentation by a derived relation studied in this paper directly corresponds to quantifier instantiation.

We are again interested in an original program $\delta$ with a $\AE$ inductive invariant $I \in \AE(\Sigma)$. Our goal is to prove that $I$ is inductive w.r.t.\ $\delta$ by a reduction to a decidable class. We therefore identify some existential formula $\psi(\vecx) \in \exists^{*}(\Sigma)$ that
expresses needed information, and encode it using an instrumentation
relation $r$, with the meaning that ``$r(\vecx) \equiv \psi(\vecx)$''.
Instrumentation then modifies $\delta$ to produce an instrumented program $\deltainstr$ with a universal inductive invariant $\Ih$ such that $I = \Ih[\psi / r]$ (at least up to redundant tautological clauses in $I$). 

Intuitively, adding the instrumentation relation $r$ lets $\Ih \in \Univ(\Sigmainstr)$
express existential information by referring to $r$ instead. The modifications
to the program must encode ``enough'' of ``$r(\vecx) \equiv \psi(\vecx)$''
to make $\Ih$ inductive.
In \Cref{sec:instrumentation}, this was performed by adding update code and rewriting program conditions.
In this section, we will take a different approach to instrumentation:
instantiating the correspondence (using $\ASSUME$
  statements) between $\psi$ and
  $r$ for specific variables in the program---relating $r(\vect)$ to $\psi(\vect)$ where $\vect$ is a specific tuple of closed terms. While different from the process discussed in \Cref{sec:instrumentation}, the result is still a sound instrumentation per \Cref{def:sound-instrumentation}.
  We show this approach---closely following quantifier instantiation---can systematically construct instrumentations $\deltainstr, \Ih$ whose effect is to prove that $I$ is inductive w.r.t.\ $\delta$, by encoding the necessary instantiations.
  We begin by defining the instrumentation process we use in this section.

\ultpara{Instrumentation by Local Instantiations}%
\label{sec:local-instantiations}
We would like to enforce $r$ to be interpreted according to $\psi$ in
the pre-state, i.e., to enforce
$\forall \vecx. \;\psi(\vecx) \leftrightarrow r(\vecx)$. The
direction $\forall \vecx. \;\psi(\vecx) \rightarrow r(\vecx)$ is
an EPR formula (since $\psi \in \E$), and thus we can simply conjoin
it to the verification conditions without sacrificing decidability.

The converse implication, $\forall \vecx. \;r(\vecx) \to \psi(\vecx)$,
is a $\AE$ formula, and adding it to the verification condition will
lead to a formula that does not belong to the decidable EPR class. Note that Bounded-Horizon with bound 1 is
analogous to enforcing $r(\vect) \to \psi(\vect)$ for every $\vect$ that is a tuple of program variables. Inspired by this, we define the following
instrumentation that lets the user locally enforce the definition of
$r$ for program variables.

\begin{definition}[Local Instantiation]%
\label{def:local-instantiation}
Let $\psi(\vecx) = \exists \vecy. \;\varphi(\vecx, \vecy)$ where $\varphi \in \QF(\Sigma)$.
To generate a \emph{local instantiation} of $\forall \vecx. \;
r(\vecx) \to \exists \vecy. \;\varphi(\vecx, \vecy)$ on some
tuple of program variables $\vect$, we instrument the program by
adding new local program variables $\vecc$, and inserting the following code:
\begin{equation}
\label{eq:local-inst}
\textbf{\textup{local}}~\vecc := *;\; \textup{\ASSUME} \; {r(\vect) \rightarrow \varphi(\vect, \vecc)}.
\end{equation}
\end{definition}
The code in \Cref{eq:local-inst} uses the havoc statement, which sets the value of $\vecc$ to
arbitrary values, followed by an assume statement that restricts the
execution such that if $r(\vect)$ holds, then $\varphi(\vect, \vecc)$
holds.
Thus, this code realizes the restriction that
$r(\vect) \to \exists \vecy. \;\varphi(\vect, \vecy)$, and also assigns to the new program variables $\vecc$ the ``witnesses'' for the existential quantifiers. \yotam{chose to be inexact with ``witness'', otherwise seems cumbersome to me} \sharon{changed $\vecc$ to $\vecy$ in $r(\vect) \to \exists \vecy. \;\varphi(\vect, \vecc)$}
Note that the new variables translate to existential quantifiers in the transition relation (through the semantics of havoc---see~\Cref{Fi:RMLWP}, recalling that the transition relation is obtained by negation).
We call this addition to the program
a local instantiation, as it imposes the connection between $r$
and $\psi$ locally for some program variables $\vect$.

\begin{lemma}[Soundness of Local Instantiations]%
\label{lem:local-instantiation}
If $\deltainstr$ is obtained from $\delta$ by a local instantiation
then $\deltainstr$ is a sound instrumentation with $\psi$ per \Cref{def:sound-instrumentation}.
\end{lemma}
\begin{proof}
It suffices to show that the logical constraint generated in $\deltainstr$ as a result of the $\textup{\ASSUME}$ command does not exclude any transition allowed by $\delta$.
The code added by a local instantiation for $\vect$ translates to a new constraint in $\deltainstr$ of the
form $\gamma = \exists \vecy. \;
r(\vect) \rightarrow \varphi(\vect, \vecy)$ through the semantics of havoc and $\bassume$ (see \Cref{Fi:RMLWP}). Note that the variables added through local instantiation are translated to existential quantifiers, not new constants.
Since $(\forall \vecx. \;
r(\vecx) \leftrightarrow \psi(\vecx)) \rightarrow \gamma$ is valid (recall that $\psi(\vecx) = \exists \vecy. \;\varphi(\vecx, \vecy)$), we have
$((\forall \vecx. \;
r(\vecx) \leftrightarrow \psi(\vecx)) \sland \delta) \rightarrow \deltainstr$ 
is valid,
which implies the condition of \Cref{def:sound-instrumentation}.
\end{proof}

\begin{remark}
The combination of adding $\forall \vecx. \;\psi(\vecx) \rightarrow
r(\vecx)$ to the verification condition and allowing the user to
perform local instantiations on the program variables is at least as powerful as rewriting
program conditions, since any rewrite of $\psi(\vect)$ to $r(\vect)$
can be simulated by a local instantiation on $\vect$.
\end{remark}

\ultpara{Instantiations for the Invariant}

The mechanism of local instantiations is designed to support
instantiations of $\forall \vecx. \;\psi(\vecx) \leftrightarrow
r(\vecx)$ required to prove that the invariant $I = \Ih[\psi / r]$
is preserved by the program. This proof is carried out by showing that
$\left(\forall \vecx. \;\psi(\vecx) \leftrightarrow
r(\vecx)\right) \land \Ih \land \delta \sland \neg (\Ih[\psi / r])'$ is
unsatisfiable.
This may require instantiating
the definition of $r$ on Skolem constants that come from the negation
of the invariant in the post-state. Thus, we extend our
instrumentation method by adding new ``program variables'' that
represent the elements of the domain on which the invariant is
potentially violated in the post-state.  For every existentially
quantified variable $x$ in $\neg \Ih$, we add a special program variable
$sk_x$ which can be used in local instantiations, enhancing their
power to prove that $\Ih$ is inductive.

To this end we formally use a slightly modified inductiveness check, termed \emph{Skolemization-aware inductiveness}.
In this check, existentially quantified variables can be shared between $\deltainstr$ and $\sk{(\neg \Ih)}$, facilitating local instantiations on Skolem constants coming from the negation of the invariant.
Formally, denote $\Ih = \forall \vec{x}. \ \theta(\vec{x})$ where $\theta(\vec{x}) \in \QF(\Sigmainstr)$ (recall that $\Ih \in \Univ(\Sigmainstr)$). Let $\vec{s}$ be free variables, intended to be shared  between $\deltainstr$ and $\sk{(\neg \Ih)}$.
Then $\vec{s}$ replace the existentially quantified variables $\vec{x}$ in $\neg \Ih$, and also replace the existentially quantified variables in $\deltainstr$  introduced by local instantiations performed on $\overline{sk}_x$, thus linking them.
The inductiveness check now translates to the unsatisfiability of
\begin{equation}
\left(\forall \vec{x}. \ \theta(\vec{x})\right) \land \exists \vec{s}. \ \left(\deltainstr \land \neg\theta'(\vec{s})\right).
\end{equation}


\ultpara{Obtaining Deep Instantiations}
Applying local instrumentation on a tuple $\vect$ that consists of
original program variables, or variables that represent Skolem
constants, corresponds to instantiations of Bounded-Horizon with bound
1. However, once a local instantiation is performed, new program
variables $\vecc$ are added. Performing a local instantiation on these
new variables now corresponds to instantiation of bound 2. By
iteratively applying local instantiations, where each iteration adds
new program variables (corresponding to existential quantifiers in the transition relation), we can thus simulate quantifier instantiations of
arbitrary depth.

\bigskip
Overall, the procedure leads to the following claim:
\begin{lemma}
Let $I \in \AE(\Sigma)$ be an inductive invariant for $\delta \in \EPR(\Sigma \uplus \Sigma')$, provable using Bounded-Horizon of bound $k$. Let $\psi \in \E(\Sigma)$ be its existential sub-formula, i.e., $I = \forall \vecx.\ \psi(\vecx)$.
Then it is possible to construct $\deltainstr \in \EPR(\instr{\Sigma} \uplus \instr{\Sigma}')$ where $\instr{\Sigma} = \Sigma \uplus \instr{S} \uplus \{r\}$, $\instr{S}$ being the Skolem constants from $\sk{(\neg I)}$ and $r$ being a fresh relation symbol, such that
    \begin{itemize}

        \item $\deltainstr$ is a sound instrumentation of $\delta$ and $\psi$, and

        \item $\instr{I} = \left(\forall \vecx.\ r(\vecx)\right) \land \left(\forall \vecx. \ \psi(\vecx) \rightarrow r(\vec{x})\right)$ is an inductive invariant for $\deltainstr$ with a Skolemization aware check.

\end{itemize}
\end{lemma}
\begin{proof}[Proof (sketch)]
Construct $\deltainstr$ by performing local instantiations, iteratively constructing variables corresponding to the terms used by the instantiations of $I = \forall x. \ \psi(x)$ required to prove that $\sk{I} \land \sk{\delta} \land \sk{(\neg I')}$ is unsatisfiable. We translate instantiations for the Skolem constants from $\sk{(\neg I')}$ to the corresponding Skolem constants of $\sk{(\neg \Ih)}$: each existential quantifier in $\neg I'$ comes from a universal quantifier in $I$, which is also present in $\forall \vec{x}. \ r(\vec{x})$ which is part of $\Ih$.
Instantiations of higher depth are obtained through local instantiation over the terms corresponding to the base terms, i.e.\ the instantiation for a term $f(t_1,\ldots,t_n)$ is simulated by a local instantiation over variables $c_{t_1},\ldots,c_{t_n}$ which were introduced for the instantiation of $t_1,\ldots,t_n$. We also conjoin to $\deltainstr$ the clause $\forall \vec{x}. \ \psi'(\vec{x}) \rightarrow r'(\vec{x})$.

The proof is by contradiction: assume that $\Ih \land \deltainstr \land \neg\Ih$ is satisfiable. Without loss of generality\footnote{
The theory of equality, which is universally quantified, can be conjoined to the verification conditions, and thus the rest of the reasoning can be performed in first-order logic without equality.
}, take a Herbrand model $\instr{\A}$ of this formula. From this we construct a model $\A$ satisfying the instantiations of $\sk{I} \land \sk{\delta} \land \sk{(\neg I')}$: the domain and interpretation of constants and relations (omitting $r,r'$) are without change. The key point is that the Skolem functions are interpreted according to the hierarchy between variables introduced through local instantiation.
(Note that $\sk{(\neg I)}$ is satisfied by $\A$, which reads $\A \models \sk{(\exists \vec{x}. \ \neg\psi'(\vec{x}))}$, because $\instr{A} \models \sk{(\exists \vec{x}. \ \neg r'(\vec{x}))}$ and $\instr{A} \models \forall \vec{x}. \ \neg r'(\vec{x}) \rightarrow \neg\psi'(\vec{x})$ (the latter is part of $\Ih$ and guaranteed to hold through $\deltainstr$).)
\yotam{Sharon, how incomprehensible is this?\ldots}
\end{proof}

\ultpara{Illustrating Example}

\begin{figure}[t]
 \[
   \begin{array}{l}
     \DERIVED{r_1(x, y)}{\exists z. \; \mReq(x, z) \land \mDB(z,y)}
     \\
     \DERIVED{r_2(x, y)}{\exists z. \; \mT(x, z) \land \mReq(z,y)}
     \\
     \DERIVED{r_3(x, y)}{\exists z. \; \mDBReq(x, z) \land \mDB(z, y)}
     \\
     \INSTRCOMMENT
     \INVARIANT
     \Ih =
     \forall \mvU, \mvP. \;
     \mResp(\mvU, \mvP) \to r_1(\mvU, \mvP) \sland \\
     \INSTRCOMMENT
     %
     \hspace{1cm}
     \forall \mvID, \mvR. \;
     \mDBReq(\mvID, \mvR) \to r_2(\mvID, \mvR) \sland \\
     \INSTRCOMMENT
     \hspace{1cm}
     \forall \mvID, \mvP. \;
     \mDBResp(\mvID, \mvP) \to r_3(\mvID, \mvP) \sland \\
     \INSTRCOMMENT
     \hspace{1cm}
     \forall \mvID, \mvU_1, \mvU_2. \;
     \mT(\mvID, \mvU_1) \land \mT(\mvID, \mvU_2) \to \mvU_1 = \mvU_2
     \\
   \ACTION{server\_process\_db\_response}{\mvID, \mvP} \\
     \INDENT
     \COMMENT{instantiate $r_3$ on $(\mvID, sk_p)$ (depth 1)} \\ 
     \INDENT\INSTRCOMMENT
     c_1 := * \CSEP \\
     \INDENT\INSTRCOMMENT
     \ASSUMEM{r_3(\mvID, sk_p) \to \mDBReq(i,c_1) \land \mDB(c_1,sk_p)} \CSEP \\
     \INDENT
     \COMMENT{instantiate $r_2$ on $(\mvID, c_1)$ (depth 2)} \\ 
     \INDENT\INSTRCOMMENT
     c_2 := * \CSEP \\
     \INDENT\INSTRCOMMENT
     \ASSUMEM{r_2(\mvID,c_1) \to \mT(\mvID,c_2) \land \mReq(c_2,c_1)} \CSEP \\
     \INDENT
     \ldots \\
     \ENDACTION \\

   \ACTION{check}{\mvU, \mvP} \\
     \INDENT
     \COMMENT{instantiate $r_1$ on $(\mvU, \mvP)$ (depth 1)} \\ 
     \INDENT\INSTRCOMMENT
     c_3 := * \CSEP \\
     \INDENT\INSTRCOMMENT
     \ASSUMEM{r_1(\mvU, \mvP) \to \mReq(\mvU,c_3) \land \mDB(c_3,\mvP)} \CSEP \\
     \INDENT
     \ldots \\
     \ENDACTION
   \end{array}
 \]
 \caption{An illustration of instrumentation by local instantiations
   for the example of \Cref{fig:client-server-db}.  The
   instrumentation adds three instrumentation relations $r_1,r_2,r_3$,
   and performs three local instantiations in order to prove that the
   invariant is inductive. Note that an instantiation depth of 2 is
   used, in accordance with the fact that the original invariant is
   provable using bound 2 but not bound 1. The complete model corresponding to
   this Figure appears in~\cite{additionalMaterials} (file \texttt{client\_server\_db\_instr.ivy}).}%
 \label{fig:client-server-db-instr}
\end{figure}


\Cref{fig:client-server-db-instr}
illustrates the local instantiation procedure on the example
of \Cref{fig:client-server-db}.  Recall that the $\AE$ invariant $I$
of \Cref{fig:client-server-db} is not provable using Bounded-Horizon
of bound 1, but is provable using Bounded-Horizon of bound 2. The
instrumentation presented in \Cref{fig:client-server-db-instr}
introduces three instrumentation relations to encode the existential
parts of $I$, thereby producing the instrumented universal invariant
$\Ih$.

To prove the inductiveness of $\Ih$, we use local instantiations in
the actions \texttt{check} and \texttt{server\_process\_db\_response}.
In \texttt{server\_process\_db\_response},
we instantiate the definition of $r_3$ on $(\mvID, sk_p)$, and assign
the existential witness to $c_1$. Intentionally, $c_1$ gets the
request that was sent from the server to the DB that led to the
response $sk_p$ being sent from the DB to the server. $sk_p$ is the
response that supposedly causes a violation of the invariant when the
action \texttt{server\_process\_db\_response} is executed (the
instantiations are used to prove that a violation does not occur).
This instantiation is of depth 1. Next, we make an instantiation of
the definition of $r_2$ on $(\mvID,c_1)$ and obtain a new existential
witness $c_2$. The use of $c_1$ here makes this instantiation depth 2.
The \texttt{check} action includes another instantiation of depth 1,
which is simply used to prove that the $\ABORT$ cannot happen
(similarly to rewriting a program condition).

The reader can observe that the local instantiations introduced during
the instrumentation process closely correspond to the instantiations
required to prove the original $\AE$ invariant $I$
(see \Cref{sec:bounded-horizon}).

It is important to note that the process of instrumentation by local instantiations
discussed here is different in spirit from those
of \Cref{sec:bounded-horizon-instrumentation}: instrumentation by
local instantiation consists of almost nothing but adding existential
quantifiers to the transition relation, as opposed to the condition
in \Cref{instrumentationNoNewExists} where we do not allow the
instrumentation to add new existential quantifiers.

\section{Partial Models for Understanding Non-Inductiveness}%
\label{sec:partial-models}


When conducting SMT-based deductive verification (e.g., using Dafny~\cite{dafny}), the user constructs both the formal representation of the system and its invariants.
In many cases, the invariant $I$ is initially  not inductive w.r.t.\ the given program, due to a bug in the program or in the invariant. Therefore, deductive verification is typically an iterative process in which the user attempts to prove inductiveness, and when this fails the user adapts the program, the invariant, or both.

In such scenarios, it is extremely desirable to present the user with a \emph{counterexample to induction} in the form of a state that satisfies $I$ but makes a transition to a state that violates it. Such a state can be obtained from a model of the formula $\inducformula = I \wedge \delta \wedge \neg I'$ which is used to check inductiveness. It explains the error, and guides the user towards fixing the program and/or the invariant~\cite{dafny,DBLP:conf/pldi/FlanaganLLNSS02}.
However, in many cases where the check involves quantifier alternation, current SMT solvers are unable to produce counterexamples.
Instead, SMT solvers usually diverge or report ``unknown''~\cite{ge2009complete,ReynoldsCAV13}.
In such cases, Bounded-Horizon instantiations can be used to present a concrete logical structure which is comprehensible to the user, obtained as a model of the (finite) instantiations of the formula $\inducformula$. 
While this structure is not a true counterexample (as it is only a model of a subset of the instantiations of the formula), it can still guide the user in the right direction towards fixing the program and/or the invariant.
We illustrate this using two examples.

\begin{figure}[t]
\begin{center}
\begin{tabular}{p{0.45\textwidth} p{0.45\textwidth}}
\begin{math}
   \begin{array}{l}
   \texttt{relation} \ \mPending(m,n)
   \\
  \texttt{init} \ \forall m,m. \, \neg\mPending(m,n)
  \\
   \ldots \; \COMMENT{ring topology} \\
   \\
   \ACTION{send\_packet}{\mvN} \\
     \INDENT
     \ASSUMEM{\mRingNext}{(\mvN, \mvM)} \\
     \INDENT
     \INSERTM{\mPending}{(\mvN, \mvM)} \\
     \INDENT
     \INSERTM{\mSent}{\mvN} \\
   \ENDACTION
   \end{array}
\end{math}
&
  \begin{math}
  \begin{array}{l}
   \\
   \ACTION{receive\_packet}{\mvN, \mvM} \\
     \INDENT
     \ASSUMEM{\mPending(\mvM, \mvN)} \\
     \INDENT
     \REMOVEM{\mPending}{(\mvM, \mvN)} \\
     \INDENT
     \IF \; \mvM = \mvN \; \THEN\\
     \INDENT \INDENT
	  \INSERTM{\mLeader}{\mvN} \\
     \INDENT
     \ELSE \\
     \INDENT \INDENT
     \IF \; \mvN < \mvM \; \THEN \\
     \INDENT \INDENT \INDENT
     \ASSUMEM{\mRingNext}{(\mvN,\mvNext)} \\
     \INDENT \INDENT \INDENT
     \INSERTM{\mPending}{(\mvM, \mvNext)} \\
     \INDENT \INDENT
     \ELSE \quad \COMMENT{do not forward} \\
     \ENDACTION \\
   \end{array}
\end{math}
\end{tabular}
\end{center}
\caption{
  A sketch of the leader election protocol discussed in this section (the complete program appears in~\cite{additionalMaterials}, file \texttt{ring\_leader\_termination.ivy}).
}%
\label{fig:leader-code}
\end{figure}

\begin{figure}[t]
\centering
\begin{subfigure}[b]{0.30\textwidth}
    \centering
    \includegraphics[trim={10pt 3pt 10pt 3pt},clip,scale=0.4]{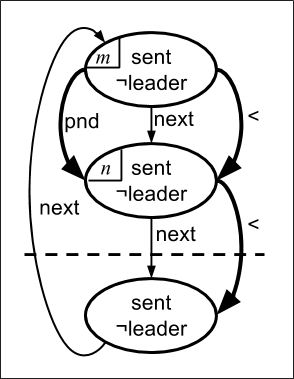}
    \caption{}\label{tacas-leader-1}
\end{subfigure}
\begin{subfigure}[b]{0.30\textwidth}
    \centering
    \includegraphics[trim={10pt 3pt 10pt 3pt},clip,scale=0.4]{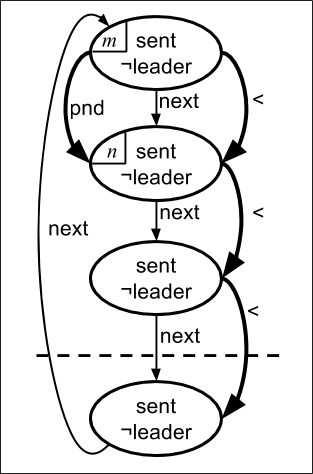}
    \caption{}\label{tacas-leader-2}
\end{subfigure}
\begin{subfigure}[b]{0.30\textwidth}
    \centering
    \includegraphics[trim={10pt 3pt 10pt 3pt},clip,scale=0.4]{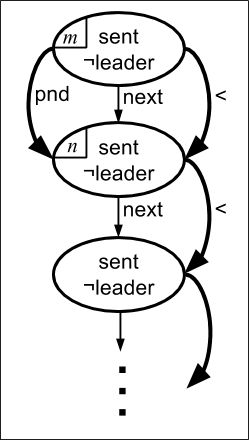}
    \caption{}\label{tacas-leader-3}
\end{subfigure}
\caption{\label{fig:leader-partial-models}
Leader-election in a ring protocol as an illustration of the use of partial models for incorrect programs and invariants.
(\textsc{a}), (\textsc{b}) show partial models of bound 1 and 2, respectively, and (\textsc{c}) illustrates an infinite structure that explains the root cause of the non-inductiveness.
 }
\end{figure}


\subsection{Leader Election in a Ring} Our first example is a simple leader-election protocol in a ring~\cite{chang1979improved}, whose model is presented in \Cref{fig:leader-code}.
The protocol assumes that nodes are organized in a directional ring topology with unique IDs, and elects the node with the highest ID as the leader.
Each node sends its own ID to its successor, and forwards messages when they contain an ID higher than its own ID\@.
A node that receives its own ID is elected as leader.
We wish to prove a termination property which states that once all nodes have sent their ID, and there are no pending messages in the network, there must be an elected leader.
To verify this we use a relational model of the protocol, similar to~\cite{pldi/PadonMPSS16}, and specify the property via the following formula:
\begin{equation}
\label{pm:le-term}
\left(\exists n. \ \mLeader(n)\right) \lor \left(\exists n_1,n_2. \ \neg\mSent(n_1) \lor \mPending(n_1,n_2)\right)
\end{equation}
A natural attempt of proving this using an inductive invariant is by conjoining \Cref{pm:le-term} (which is not inductive by itself) with the following property (this was the authors' actual next step in proving this termination property):
\begin{equation}
\label{pm:ae-inv}
\forall n_1. \ \mSent(n_1) \land \neg\mLeader(n_1) \rightarrow \left(\left(\exists n_2.\ \mPending(n_1,n_2)\right) \lor \left(\exists n_2.\ n_1 < n_2\right)\right)
\end{equation}
meaning that if a node has sent its own ID but has not (yet) become leader, then there is either a message pending in the network with the node's ID, or a node with a higher ID\@.

Alas, the conjunction of \Cref{pm:le-term,pm:ae-inv} is still not an inductive invariant for the protocol (as we explain below).
Since \Cref{pm:ae-inv} contains $\AE$ quantification, the associated 
inductiveness check is outside of the decidable EPR fragment.
Indeed, Z3 diverges when it is used to check $\inducformula$.
This is not surprising since the formula has no satisfying finite structures, but has an infinite model (a scenario that is not unusual for $\AE$ formulas).

On the other hand, applying Bounded-Horizon (with any bound) to $\inducformula$  
results in a formula that has finite models. 
These concrete models 
are \emph{partial models} of $\inducformula$.
%
Figs.~\ref{fig:leader-partial-models}(a) and (b) show  partial models (restricted to the pre-states) obtained with bounds of 1 and 2, respectively, on this example\footnote{Ivy~\cite{DBLP:conf/sas/McMillanP18}, a language and tool for the verification of distributed protocols, can visualize counterexamples with relations of arity up to 2~\cite{pldi/PadonMPSS16}. The dashed line marks the boundary between elements on which all $\AE$ clauses holds and those on which it does not. It is meaningful when the $\AE$ part consists of a single universal quantifier over the existential formula, and is readily computed automatically.}: they show (finite) rings in which all elements have advertised themselves (sent), none are considered leaders, and there is a message pending to $n$ with the ID of $m$, which $n$ would not forward as $n$'s ID is greater than $m$'s. Therefore, once the message is received there are no pending messages and no leader, which contradicts the safety property. 

These models are not true counterexamples to induction: the sub-formula of \Cref{pm:ae-inv} residing under the universal quantifier does not hold for all the elements of the domain.
It does, however, hold for all elements with which the quantifier was instantiated, which are the elements above the dashed line in the figure. For these, there exists a node with a higher ID, which is mandated by the invariant of \Cref{pm:ae-inv} (unless corresponding messages are present).
Intuitively, these elements have all sent their own ID, which was blocked by their successor that has a higher ID, so none of them is the leader.
In a finite model, this has to end somewhere, because one of the nodes must have the highest ID\@. Hence, no finite counter-model exists. However, extrapolating from \Cref{fig:leader-partial-models}(a) and (b), 
we can obtain the infinite model depicted in \Cref{fig:leader-partial-models}(c). This model represents an infinite (``open'') ring in which each node has a lower ID than its successor. 
This model is a true
model of the formula $\inducformula$ generated by the invariant in \Cref{pm:le-term,pm:ae-inv}, but the fact that it is infinite prevented Z3 from producing it.

\TODO{R1: turn the following into a paragraph trying to highlight the general methodology}
Since we use tools that check general (un)satisfiability, which is not limited to finite structures, the only way to prove that an invariant is inductive  is to exclude infinite counterexamples to induction as well. 
Using Bounded-Horizon instantiations, we are able to obtain meaningful partial models that provide hints to the user of what is missing.
In this case, the solution is to add an axiom to the system model which states that there is a node with maximal ID:\@
$\exists n_1. \ \forall n_2. \ n_2 \leq n_1$. With this additional assumption, the 
formula $\inducformula$ is unsatisfiable so the invariant is inductive, and this is proven both by Z3's instantiation heuristics and by Bounded-Horizon with a bound of 1.
This illustrates the usefulness of Bounded-Horizon when the invariant is not inductive. 

\subsection{Trusted Chain}
Another example which exhibits similar behavior is that of a \emph{trusted chain}.
In this example, nodes are organized in a linear total order between nodes $s$ and $t$.
A node is marked trusted when it receives a message. When a node processes a message it propagates the message to its successor.
The safety property requires that all nodes are trusted when there are no pending messages.
The invariant further states that if a node is trusted then its successor is trusted or has a pending message. Referring to the successor in the linear total order introduces $\AE$ quantification\footnote{This can be understood as saying ``\emph{for every} pair of nodes, either \emph{exists} a node between them or\ldots''.}, as well as that if a message is pending to a node after $s$ then $s$ is already trusted.
This invariant is \emph{not} inductive due to a counterexample with an \emph{infinite} chain of trusted nodes and an $\omega$ untrusted node (similar to the construction in \Cref{thm:undecidability-infinite}).
Partial models for this case would have the form of a finite chain in which all nodes are trusted apart from $t$ and nodes lying ``beyond the horizon''---meaning, $\AE$ invariant was not instantiated on them---including $t$'s predecessor.\footnote{Ivy seeks a counterexample over a domain of minimal cardinality~\cite{pldi/PadonMPSS16}, in which case $t$ and $t$'s predecessor would be the \emph{only} untrusted nodes.}
In this case, the invariant can be made inductive by adding an axiom expressing an induction principle of being trusted over the order (``if $s$ is trusted, and every successor of a trusted node is trusted, then all nodes are trusted''). This fix causes the invariant to become inductive, and this invariant is now provable using Bounded-Horizon of bound 1.

\bigskip
\yotam{Sharon, rephrased this}
Generally, using the Bounded Horizon technique with increasing bounds produces a sequence of partial counter-models. This is unlike the usual scenario upon the solver's divergence, where the user is typically provided with no evidence of what went wrong. By analyzing this sequence of partial counter-models, the user can extrapolate to an infinite counter-model, and identify a strengthening (e.g., an axiom) that will overrule it.

%

\commentout{

\ultpara{Meaning of Partial Models}

While incomplete, presenting these models to the user can help them to understand the problem: extrapolating on the elements shown in the structures, the ``blurry'' part of the domain \TODO{explain what this means} must have that $\mSent$ is true and that there are no pending messages, otherwise the property in \Cref{pm:le-term} would hold.
It therefore must hold that every node has a node with a higher id. \TODO{explain how to see in model}

This points in the direction of a true, infinite counterexample in which there is no node with a maximal id, and therefore no leader is elected by the protocol.
Note that the invariant does hold for finite structures, explaining the difficulty of automatically finding a true counterexample in this case.

Following this understanding the user may strengthen the invariant with the existential property, meaning that there is a leader with a maximal id in the network, produces an inductive invariant which proves the termination property.

In contrast, when the invariant is not inductive Bounded-Horizon always terminates, with the result that the invariant may not be inductive (\Cref{lem:bh-sound}). However, models of Bounded-Horizon instantiations may not be true models of the original formula, but \emph{partial-models}.
In a partial model the invariant does not hold for all elements of the domain, but does hold for (tuples of) elements pointed to by constants.
Intuitively, increasing the bound expands the ``zone of confidence'', on which the invariant holds, to wider and wider distances in the domain (measured by applications of the Skolem functions).
\yotam{Probably very confusing, problem with informality?}
\yotam{this doesn't really explain 1-alternation partial models, in which also the AE may hold in the post state although it shouldn't}

}




\section{Implementation and Evaluation}%
\label{sec:implementation}
In this section, we describe our implementation of Bounded-Horizon of bound 1 and evaluate it on several correct and incorrect examples.
While discussing the examples, we note that certain types of ghost code that records properties of the history of an execution can be viewed as a special case of instrumentation, and demonstrate how the Bounded-Horizon technique circumvents the need for augmenting the program with such ghost code.

\subsection{Implementation}%
\label{sec:implementation-encoding}
We implemented a prototype of Bounded-Horizon of bound 1 on top of Z3~\cite{de2008z3} and used it within Ivy~\cite{DBLP:conf/sas/McMillanP18} and the framework of~\cite{CAV:IBINS13}.

Our implementation works by adding ``guards'' that restrict the range of universal quantifiers to the set of constants where necessary.
Technically, recall that we are considering the satisfiability of $\inducformula = \sk{I} \land \sk{\delta} \land \sk{(\neg I')}$.\footnote{
	Skolemization is performed via Z3, taking advantage of heuristics that reduce the number of different Skolem functions.}
 Let $\forall x. \ \theta$ be a subformula of $\inducformula$. If $\theta$ contains function symbol applications\footnote{
	This in fact implements the approximation as of \Cref{eq:bounded-horizon-approx}. The exact bound 1 per \Cref{eq:bounded-horizon-def} can be implemented by a more careful consideration of which universally quantified variables should be restricted, but this was not necessary for our examples.
}, we transform the subformula to
$\forall x. \ \bigl(\bigvee_{c}{x = c}\bigr) \rightarrow \theta$ where $c$ ranges over $\consts{\inducformula}$.
The resulting formula is then dispatched to the solver.
This is a simple way to encode a bound of 1 on the depth of instantiations performed (instead of trying instantiations of higher and higher depth),
while leaving room for the solver to perform the necessary instantiations cleverly.
The translation enlarges the formula by $O(\text{\#}\text{Consts} \cdot \text{\#}\forall)$ although the number of bounded instantiations grows exponentially with $\text{\#}\forall$.
The exponential explosion is due to combinations of constants in the instantiation, a problem we defer to the solver.

Our encoding restricts the constants necessary to decide satisfiability of the formula while avoiding an exponential blowup in the translation. The task of exploring the space of possible instantiations is left for Z3's instantiation heuristics. It not immediate, however, that Z3 is guaranteed to terminate on the resultant formula, which syntactically does not fall into a decidable class (despite of encoding bounded instantiations). Fortunately, employing Model-Based Quantifier Instantiation~\cite{ge2009complete}, Z3 is guaranteed to terminate on this formula, as desired.
This is because during the Model-Based Quantifier Instantiation process every instantiation of a universal formula has the same truth value in the model as an instantiation using one of the existing ground terms
(constants and then $\bhterms{1}$ terms).
Z3's instantiation engine produces instantiations using existing terms rather than create superfluous new terms, and so must terminate on the formulas our procedure produces~\cite{nikolajPersonal}.

An alternative approach of implementation is to integrate the instantiation bound with the solver's heuristics more closely (see~\cite{decidinglocalCAV15}).
\yotam{Sharon, what do you think about this sentence?}
It would also be interesting to efficiently integrate this approach with resolution based first-order theorem-provers such as Vampire~\cite{DBLP:conf/cav/KovacsV13}.

\commentout{
	\subsection{Benchmarks}

	Before we describe the examples we used in our evaluation, we note that certain types of ghost code that records properties of the history of an execution can be viewed as a special case of instrumentation.
	Our examples include a program where this idea is applied.


	\TODO{move the subsection on derived relations over the history to repalce the paragraph after the list of examples, and re-augment it with the firewall example}
}

\subsection{Examples}
We applied the procedure to the incorrect examples of \Cref{sec:partial-models}, and also successfully verified several correct programs and invariants using bound 1. These examples are (the examples' code can be found in~\cite{additionalMaterials}):

\begin{itemize}
    \item The client-server example of \Cref{fig:client-server}.

    \item List reverse~\cite{CAV:IBINS13}, where the invariant states
      that the $n$ edges (``next'' pointers) are reversed.  The
      invariant is $\forall^* \exists^*$ due to the encoding of $n$
      via its (reflexive) transitive closure $\nstar$ as explained in~\cite{CAV:IBINS13}.

    \item Learning switch~\cite{DBLP:conf/pldi/BallBGIKSSV14}, in which the routing tables of the switches in the network are automatically constructed through observing the source packets arriving on each switch's links.
    The invariant states that in every routing path, every routing node has a successor.

    \item Hole-punching firewall~\cite{DBLP:conf/pldi/BallBGIKSSV14}, in which the firewall allows a packet from an external host to enter the network only if the host is recorded as trusted by the firewall.
      The invariant states that if a packet from an external host then is allowed, then there previously was an internal host that contacted the external host.
      We explored two modeling alternatives: using a ghost history relation, or existentially
      quantifying over time. We elaborate on this topic and this example below.

    \item Trusted chain, as explained in \Cref{sec:partial-models}.

    \item Leader election in a ring~\cite{chang1979improved,pldi/PadonMPSS16} with the
      invariant discussed in \Cref{sec:partial-models}. (See \Cref{sec:partial-models} for full details.)
\end{itemize}


\ultpara{Derived Relations Over the History}
\def\tnow{t_\text{now}}
\def\tsort{\textit{Time}}
\def\tbefore{\textit{before}}
Sometimes expressing verification conditions requires modifications to the program in which the state is augmented with additional relations, but these relations cannot be defined as derived relations over the core program relations.
Often the reason is that the property of interest depends not only on the current program state, but also on previous states in the execution history.
Such a scenario occurs in the example of the \emph{hole-punching firewall}.
In this example, a firewall controls packets entering and leaving the organization's network. Packets from outside the organization are allowed if they originate from a host that is considered trusted. A host $\textit{host}_{o}$ is considered trusted if some host $\textit{host}_{i}$ in the organization's network has previously sent a packet to $\textit{host}_{o}$.
This correctness condition depends on previous states in the execution history, and cannot be expressed in an inductive invariant without changing the vocabulary.

A common method to overcome this problem is to introduce \emph{ghost state} to record historical information, and \emph{ghost code} to mutate the ghost state, implementing this book-keeping. In the hole-punching firewall example, the user can add a relation $\textit{ever-pending}(\textit{host}_1,\textit{host}_2)$, and add ghost code that adds the tuple $(\textit{host}_1,\textit{host}_2)$ to $\textit{ever-pending}$ whenever $\textit{host}_1$ sends a packet to $\textit{host}_2$.

We observe that it is sometimes possible to think of such ghost relations as standard derived relations, defined over a vocabulary in which past states are available, and view the procedure of adding appropriate ghost code as an instrumentation by a derived relation as defined in this paper (\Cref{def:sound-instrumentation}). The idea is as follows:

In our first-order setting, we can lift the vocabulary to encode the time explicitly with a classical encoding (see e.g.~\cite{Abadi:1989:PTP:66447.66449}), by adding a new parameter $t$
to every relation and constant symbol.
In the transition relation, references to $p(\cdot)$ are replaced by $p(\tnow, \cdot)$ and $p'(\cdot)$ by $p(\tnow', \cdot)$, where
$\tnow$ is a new constant that represents the current time. $\tnow$ is incremented in every transition (according to some total order on time).
The transition relation also needs to include the requirement that the previous states are not modified by the current transition, meaning that
$\forall t < \tnow. \ \forall \vec{x}. \ p(t,\vec{x}) \leftrightarrow p'(t,\vec{x})$.
Call this modified transition relation $\delta_t$.


We now proceed to express the ghost relation as a derived relation over the history.
In the example of the hole-punching firewall, the relation $\textit{ever-pending}$ can now be expressed by the derived relation
$\psi(\textit{host}_1, \textit{host}_2) \equiv \exists t. \ t \leq \tnow \land \textit{pending}(t, \textit{host}_1, \textit{host}_2)$.

With this construction, we can directly use the derived relation in the inductive invariant instead of the ghost relation.
Viewing the addition of the ghost relation and ghost code as a sound instrumentation,
\Cref{lem:original-sub-inductive} implies that an inductive invariant for the program with ghost state induces an inductive invariant for $\delta_t$. This invariant is defined over the vocabulary that records the entire history, but without the ghost relation symbol.

The value of this point of view is that the results of this paper imply that in certain cases the user can prove an invariant expressed over the history using bounded instantiations with a low bound, \emph{without resorting to ghost code manipulations}.
The reasoning is as follows:
If adding the ghost code does not add existential quantifiers, then augmenting the program with the ghost code can be thought of as instrumentation without additional existentials.
If, additionally, the derived relation --- now defined over the entire execution history --- can be expressed as a combination of universal and existential properties, \Cref{univ:psi-af-bound-2} applies, showing that Bounded-Horizon with a low-bound is guaranteed to prove the inductiveness of the invariant expressed for $\delta_t$, with no need to add the ghost code.

These conditions are satisfied by the hole-punching firewall example. We manually performed the transformation of the transition relation to a vocabulary over the history in Ivy, and successfully proved the inductive invariant with Bounded-Horizon of bound 1.


\subsection{Evaluation}
\TODO{R1: don't call it ``initial''. Make it clearer that the goal of the evaluation is to show low overhead and not better performance}
In this section, we attempt to answer the following questions regarding the applicability of Bounded-Horizon of bound 1:
\begin{itemize}
	\item Is the method sufficiently \textbf{powerful} to prove correct inductive invariants of interesting programs? (\Cref{sec:eval-power})

	\item Can our implementation achieve quick \textbf{termination} (with a partial counterexample) when standard (complete) methods diverge due to an infinite counterexample? (\Cref{sec:eval-sat})

	\item Is the \textbf{overhead} associated with Bounded-Horizon on correct invariants reasonable compared to a baseline implementation (which does not bound instantiation a-priori)? (\Cref{sec:eval-overhead})
\end{itemize}

\noindent
\Cref{tab:implementation-table} compares the running time of our implementation of Bounded-Horizon of bound 1 with Z3 as a baseline. ``---'' means the solver did not terminate in a \timeout{} timeout.
We elaborate on each of the questions referring to the results of this table.

\noindent\makebox[\textwidth]{
\begin{threeparttable}[ht]
\centering
\footnotesize
\caption{Experimental results.\label{tab:implementation-table}}
\begin{tabular}{lcccccccccccc}
{Program} & {\#$\forall$} & {\#Func} & {\#Consts} & {\#$\forall^{\downarrow}$} & {B1 Total} & {B1 Solve} & {Baseline Z3}
\\
\toprule
Client-server & 14 & 1 & 15 & 2 & 58 ms & 3 ms & 3 ms
\\
List reverse & 47 & 3 & 15 & 4 & 319 ms & 211 ms & 50 ms
\\
Learning switch & 70 & 1 & 7 & 37 & 2004 ms & 83 ms & 91 ms
\\
Hole-punching firewall with ghost & 15 & 1 & 18 & 3 & 354 ms & 14 ms & 14 ms
\\
Hole-punching firewall $\exists$ time & 32 & 2 & 21 & 3 & 485 ms & 14 ms & 14 ms
\\
Trusted chain (correct) & 27 & 1 & 19 & 2 & 435 ms & 30 ms & 37 ms
\\
Leader-election in a ring (correct) & 41 & 1 & 21 & 1 & 517 ms & 33 ms & 47 ms
\\
\midrule
Trusted chain (incorrect) & 23 & 15 & 1 & 2 & 393 ms & 95 ms & --- 
\\
Leader-election in a ring (incorrect) & 40 & 1 & 20 & 1 & 899 ms & 417 ms & --- 
\\
\bottomrule
\end{tabular}
        \begin{tablenotes}
        	\item \textbf{B1 Total} is the time in milliseconds for the bound 1 implementation.
        	It is compared to \textbf{Baseline Z3} which is the solving time in milliseconds of $\inducformula$ as is (with quantifier alternation) by Z3.
        	``---'' means the solver did not terminate in a \timeout{} timeout.
        	\textbf{B1 Solve} measures the solving time of the formula restricted to bound 1, which demonstrates that most of the overhead occurs when constructing the formula.
        	\textbf{\#$\forall$} is the number of universal quantifiers in $\inducformula$, \textbf{\#Func} the number of different Skolem function symbols, and \textbf{\#Consts} the number of constants.
        	{\#$\forall^{\downarrow}$} is the number of universally quantified variables that were restricted in the bound 1 check.
        	Measurements were performed on a Linux 4.15.0 VM running on a 2.9GHz Intel i7-7500U CPU, except for client-server and list reverse which were measured on a 3.5GHz Intel i5-4690 CPU with 8GB RAM running Linux 3.13 x86\_64. 
        \end{tablenotes}
\end{threeparttable}}

\bigskip

\para{Summary of Results.}
The results are encouraging because they suggest that the termination strategy of Bounded-Horizon, at least for bound 1, can be combined with existing instantiation techniques to ensure termination with only a slight performance penalty. Bounded-Horizon successfully terminating on incorrect examples with a partial example suggests that the Bounded-Horizon termination criterion may indeed be useful for ``sat'' instances on which the solver may diverge.

\subsubsection{Bound 1 Suffices for Interesting Programs and Invariants}%
\label{sec:eval-power}

Bound 1 successfully proves the invariant inductive in all correct examples list in \Cref{tab:implementation-table}. This demonstrates in practice the power of bounded instantiations shown theoretically in \Cref{sec:bounded-horizon-instrumentation}.

\subsubsection{Bound 1 Terminates on Incorrect Invariants}%
\label{sec:eval-sat}

In this section we consider the incorrect invariants of leader election termination in a ring and trusted chain, outlined in \Cref{sec:partial-models}.
On these examples, our encoding of bounded instantiations allowed Z3 to terminates in seconds.
In contrast, baseline Z3 does not converge in \timeout{} on either of these examples.
The same happens with the first-order theorem prover Vampire~\cite{DBLP:conf/cav/KovacsV13}, invoked on an SMTLIB2 translation of the verification condition.
This demonstrates that the encoding of bounded instantiations outlined in \Cref{sec:implementation-encoding} allows Z3 to terminate efficiently, although to do so the solver must in theory exhaust all possible bounded instantiations.

\subsubsection{Bound 1 Has Reasonable Overhead}%
\label{sec:eval-overhead}

In this section we consider the correct invariants in \Cref{tab:implementation-table}. Both Bounded-Horizon and baseline Z3 successfully prove these inductive correct.
Bounded-Horizon of bound 1 introduces an overhead in this computation; this is to be expected since the chief benefit of our approach is the termination guarantee on incorrect examples.
We deem this overhead a reasonable price to ensure termination (in general the user does not know in advance whether the invariant is inductive or not before querying the solver). Comparing the columns \textbf{B1 Total} and \textbf{B1 Solve} with \textbf{Baseline Z3} in \Cref{tab:implementation-table}, it is apparent that most of the overhead is due to the transformation of the formula before it is sent to the solver, which may be improved by further engineering.

\section{Related Work}%
\label{sec:related}

\para{Quantifier Instantiation}
The importance of formulas with quantifier-alternation for program
verification
has led to many developments in the SMT and theorem-proving communities that aim to allow automated
reasoning with such formulas.
The Simplify system~\cite{simplifyJACM05} promoted the practical usage of quantifier
triggers, which let the user affect the quantifier instantiation procedure in a
clever way. 
Similar methods are integrated into modern SMT solvers such as Z3~\cite{de2008z3}.
Recently, a method for annotating the source code with triggers has
been developed for Dafny~\cite{leinotriggerCAV16}.
The notion of instantiation depth is related to the notions of matching-depth~\cite{simplifyJACM05} and instantiation-level~\cite{DBLP:journals/amai/GeBT09} which are used for prioritization within the trigger-based instantiation procedure.

In addition to user-provided triggers, many automated heuristics for
quantifier instantiation have been developed, such as Model-Based Quantifier Instantiation~\cite{ge2009complete}.
Even when quantifier instantiation is refutation-complete, it is still important and challenging to handle
the SAT cases, which are especially important for program verification.
Accordingly, many
works (e.g.,~\cite{ReynoldsCAV13}) consider the problem of model finding.

Local Theory Extensions and Psi-Local Theories~\cite{DBLP:conf/cade/Sofronie-Stokkermans05,IhlemannJS08TACAS08,decidinglocalCAV15}
identify settings in which limited quantifier instantiations are complete.
They show that completeness is achieved exactly when every partial model can be extended to a (total) model.
In such settings Bounded-Horizon instantiations are complete for invariant checking.
However, Bounded-Horizon can also be useful when completeness cannot be guaranteed.

Classes of SMT formulas that are decidable by complete instantiations have been studied by~\cite{ge2009complete}.
In the uninterpreted fragment, a refined version of Herbrand's Theorem generates a finite set of instantiations when the dependencies are stratified.
Bounded-Horizon is a way to bound unstratified dependencies.

Finally, first-order theorem provers, such as Vampire~\cite{DBLP:conf/cav/KovacsV13} and SPASS~\cite{DBLP:books/el/RV01/Weidenbach01}, employ other proof techniques such as resolution. These techniques are also prone to divergence on formulas with quantifier-alternation and in general unable to generate infinite counterexamples. \yotam{Sharon, (2) can you review?}

\para{Natural Proofs}
Natural proofs~\cite{Qiu0SM13} provide a sound and incomplete proof technique for deductive verification.
The key idea is to instantiate recursive definitions over the terms appearing in the program.
Bounded-Horizon is motivated by a similar intuition, but focuses on instantiating quantifiers in a way that is appropriate for the EPR setting.

\para{Decidable Logics} Different decidable logics can be used to check inductive invariants.
For example, Monadic second-order logic~\cite{DBLP:conf/tacas/HenriksenJJKPRS95} obtains decidability by limiting the underlying domain to consist of trees only, and in particular does not allow arbitrary relations, which are useful to describe properties of programs.
There are also many decidable fragments of first-order logic~\cite{borger2001classical}.
Our work aims to transcend the class of invariants checkable by a reduction to the decidable logic EPR\@.
We note that the example of \Cref{sec:partial-models} does not fall under the Loosely-Guarded Fragment of first-order logic~\cite{hodkinson2002loosely}
due to a use of a transitivity axiom, and does not enjoy the finite-model property.

\para{Abstractions for verification of infinite-state systems}
Our work is closely related to abstractions of infinite-state systems.
These abstractions aim at automatically inferring inductive invariants in a sound way.
We are interested in checking if a given invariant is inductive either for automatic and semi-automatic verification.

The View-Abstraction approach~\cite{DBLP:conf/vmcai/AbdullaHH13,abdulla2014block,abdulla2015parameterized}
defines a useful abstraction for the verification of parameterized
systems.
This abstraction is closely related to universally
quantified invariants. An extension of this approach~\cite{abdulla2014block} adds contexts to the abstraction, which
are used to capture $\AE$ invariants in a restricted setting where nodes have
finite-state and are only related by specific topologies.
Our work is
in line with the need to use $\AE$ invariants for verification, but
applies in a more general setting (with unrestricted high-arity
relations) at the cost of losing completeness of invariant checking.

Our work is related to the TVLA system~\cite{Lev-AmiS00,TOPLAS:SRW02} which allows the programmers
to define instrumentation relations.
TVLA also employs finite differencing to infer sound
update code for updating instrumentation relations~\cite{TOPLAS:RepsSL10}, but generates non-EPR formulas and does not guarantee completeness.
The focus operation in TVLA implements materialization which resembles quantifier-instantiation.
TVLA shows that very few built-in instrumentation relations can be used to verify many different programs.
\commentout{
The local instantiations in \Cref{sec:instrumentation-revisited} are somewhat analogous to
the materialization triggered by focus.
}

\para{Instrumentation and Update Formulas}
The idea of using instrumentation relations and generating update formulas is not limited to TVLA
and was also used for more predictable SMT verification~\cite{LahiriQ06,LahiriQ08}.

\para{Instrumentation for Decidable Logics}
The technique of instrumentation by derived relations to allow decidable reasoning is further discussed in a recent work~\cite{paxosEpr}. Several variants of Paxos are proved safe using models in the decidable logic of EPR with stratified functions, whose use is enabled by instrumentation.
Efficiently implementing a bounded instantiations scheme that bounds instantiations only where they are unstratified is an interesting challenge, in hope of relieving the need for instrumentation also for such cases.



\section{Conclusion}%
\label{sec:conclusion}
We have provided an initial study of the power of bounded instantiations for tackling quantifier alternation.
This paper shows that quantifier instantiation with small bounds can simulate instrumentation.
This is a step in order to eliminate the need for instrumenting the program, which can be error-prone.
\iflonglong%
The other direction, i.e.\ simulating quantifier instantiation with instrumentation, was also presented for conceptual purposes, although it is less appealing from a practical point of view.
\else
\iflong%
The other direction, i.e.\ simulating quantifier instantiation with instrumentation, is also possible but is less appealing from a practical point of view, and is presented in \Cref{sec:appendix-instrumentation-revisited}.
\else
The other direction, i.e.\ simulating quantifier instantiation with instrumentation, is also possible but is less appealing from a practical point of view, and is presented in the extended version~\cite{extendedVersion}.
\fi
\fi

We are encouraged by our initial experience that shows that various protocols can be proven with small instantiation bounds, and that partial models are useful for understanding the failures of the solver to prove inductiveness. Some of these failures correspond to non-inductive claims, especially those due to infinite counterexamples.
In the future we hope to leverage this in effective deductive verification tools, and explore meaningful ways to display infinite counterexamples to the user.
Other interesting directions include further investigation into the automation of program transformations for the purpose of verification (of which instrumentation is an example), including types of ghost code, and the use of Bounded-Horizon for automatically inferring invariants with quantifier-alternation.

\para{Acknowledgments.}
We would like to thank Nikolaj Bj{\o}rner, Shachar Itzhaky, and Bryan Parno for helpful discussions, Gilit Zohar-Oren for help and feedback, and the anonymous referees whose comments have improved the paper.
The research leading to these results has received funding from the European Research Council
under the European Union's Seventh Framework Programme (FP7/2007--2013) / ERC grant agreement no [321174], and the European Union's Horizon 2020 research and innovation programme (grant agreement No [759102-SVIS]).
This research was partially supported by BSF grant no. 2012259, and by Len Blavatnik and the Blavatnik Family foundation,
the Blavatnik Interdisciplinary Cyber Research Center, Tel Aviv University,
the United States-Israel Binational Science Foundation (BSF) grants No.\ 2016260 and 2012259,
and the Israeli Science Foundation (ISF) grant No.\ 2005/17.

\addcontentsline{toc}{chapter}{Bibliography}
\bibliographystyle{alpha}
\bibliography{refs}

\newcommand{\etalchar}[1]{$^{#1}$}
\begin{thebibliography}{PMP{\etalchar{+}}16}

\bibitem[Aba89]{Abadi:1989:PTP:66447.66449}
Mart\'{\i}n Abadi.
\newblock The power of temporal proofs.
\newblock {\em Theor. Comput. Sci.}, 65(1):35--83, June 1989.

\bibitem[add]{additionalMaterials}
Full code materials.
\newblock
  \url{http://www.cs.tau.ac.il/research/yotam.feldman/papers/tacas17/examples_code.zip}.

\bibitem[AHH13]{DBLP:conf/vmcai/AbdullaHH13}
Parosh~Aziz Abdulla, Fr{\'{e}}d{\'{e}}ric Haziza, and Luk{\'{a}}s Hol{\'{\i}}k.
\newblock All for the price of few.
\newblock In {\em Verification, Model Checking, and Abstract Interpretation,
  14th International Conference, {VMCAI} 2013, Rome, Italy, January 20-22,
  2013. Proceedings}, pages 476--495, 2013.

\bibitem[AHH14]{abdulla2014block}
Parosh~Aziz Abdulla, Fr{\'e}d{\'e}ric Haziza, and Luk{\'a}{\v{s}} Hol{\'\i}k.
\newblock Block me if you can!
\newblock In {\em International Static Analysis Symposium}, pages 1--17.
  Springer, 2014.

\bibitem[AHH15]{abdulla2015parameterized}
Parosh Abdulla, Fr{\'e}d{\'e}ric Haziza, and Luk{\'a}{\v{s}} Hol{\'\i}k.
\newblock Parameterized verification through view abstraction.
\newblock {\em International Journal on Software Tools for Technology
  Transfer}, pages 1--22, 2015.

\bibitem[BBG{\etalchar{+}}14]{DBLP:conf/pldi/BallBGIKSSV14}
Thomas Ball, Nikolaj Bj{\o}rner, Aaron Gember, Shachar Itzhaky, Aleksandr
  Karbyshev, Mooly Sagiv, Michael Schapira, and Asaf Valadarsky.
\newblock Vericon: towards verifying controller programs in software-defined
  networks.
\newblock In {\em {ACM} {SIGPLAN} Conference on Programming Language Design and
  Implementation, {PLDI} '14, Edinburgh, United Kingdom - June 09 - 11, 2014},
  page~31, 2014.

\bibitem[BGG96]{BGG}
E.~B\"orger, E.~Gr\"adel, and Y.~Gurevich.
\newblock {\em The Classical Decision Problem}.
\newblock Springer-Verlag, 1996.

\bibitem[BGG01]{borger2001classical}
Egon B{\"o}rger, Erich Gr{\"a}del, and Yuri Gurevich.
\newblock {\em The classical decision problem}.
\newblock Springer Science \& Business Media, 2001.

\bibitem[Bj{\o}]{nikolajPersonal}
Nikolaj Bj{\o}rner.
\newblock personal communication.

\bibitem[BRK{\etalchar{+}}15]{decidinglocalCAV15}
Kshitij Bansal, Andrew Reynolds, Tim King, Clark~W. Barrett, and Thomas Wies.
\newblock Deciding local theory extensions via e-matching.
\newblock In {\em Computer Aided Verification - 27th International Conference,
  {CAV} 2015, San Francisco, CA, USA, July 18-24, 2015, Proceedings, Part
  {II}}, pages 87--105, 2015.

\bibitem[CR79]{chang1979improved}
Ernest Chang and Rosemary Roberts.
\newblock An improved algorithm for decentralized extrema-finding in circular
  configurations of processes.
\newblock {\em Communications of the ACM}, 22(5):281--283, 1979.

\bibitem[Dij76]{Dijkstra76}
Edsger~W. Dijkstra.
\newblock {\em A Discipline of Programming}.
\newblock Prentice-Hall, 1976.

\bibitem[Dij82]{EWD:EWD821}
Edsger~W. Dijkstra.
\newblock From predicate transformers to predicates ({D}edicated by the
  {T}uesday {A}fternoon {C}lub to {C.A.R. Hoare} at the occasion of his being
  elected {F}ellow of the {R}oyal {S}ociety.).
\newblock circulated privately, April 1982.

\bibitem[DMB08]{de2008z3}
L.~De~Moura and N.~Bj{\o}rner.
\newblock Z3: An efficient {SMT} solver.
\newblock In {\em TACAS}, 2008.

\bibitem[DNS05]{simplifyJACM05}
David Detlefs, Greg Nelson, and James~B. Saxe.
\newblock Simplify: a theorem prover for program checking.
\newblock {\em J. {ACM}}, 52(3):365--473, 2005.

\bibitem[FLL{\etalchar{+}}02]{DBLP:conf/pldi/FlanaganLLNSS02}
Cormac Flanagan, K.~Rustan~M. Leino, Mark Lillibridge, Greg Nelson, James~B.
  Saxe, and Raymie Stata.
\newblock Extended static checking for java.
\newblock In {\em Proceedings of the 2002 {ACM} {SIGPLAN} Conference on
  Programming Language Design and Implementation (PLDI), Berlin, Germany, June
  17-19, 2002}, pages 234--245, 2002.

\bibitem[FPI{\etalchar{+}}17]{originalPaper}
Yotam M.~Y. Feldman, Oded Padon, Neil Immerman, Mooly Sagiv, and Sharon Shoham.
\newblock Bounded quantifier instantiation for checking inductive invariants.
\newblock In {\em TACAS}, 2017.

\bibitem[GBT09]{DBLP:journals/amai/GeBT09}
Yeting Ge, Clark~W. Barrett, and Cesare Tinelli.
\newblock Solving quantified verification conditions using satisfiability
  modulo theories.
\newblock {\em Ann. Math. Artif. Intell.}, 55(1-2):101--122, 2009.

\bibitem[GM09]{ge2009complete}
Yeting Ge and Leonardo~De Moura.
\newblock Complete instantiation for quantified formulas in satisfiabiliby
  modulo theories.
\newblock In {\em International Conference on Computer Aided Verification},
  pages 306--320. Springer, 2009.

\bibitem[HHK{\etalchar{+}}15]{IronFleet}
Chris Hawblitzel, Jon Howell, Manos Kapritsos, Jacob~R. Lorch, Bryan Parno,
  Michael~L. Roberts, Srinath T.~V. Setty, and Brian Zill.
\newblock Ironfleet: proving practical distributed systems correct.
\newblock In {\em Proceedings of the 25th Symposium on Operating Systems
  Principles, {SOSP}}, pages 1--17, 2015.

\bibitem[HJJ{\etalchar{+}}95]{DBLP:conf/tacas/HenriksenJJKPRS95}
Jesper~G. Henriksen, Jakob~L. Jensen, Michael~E. J{\o}rgensen, Nils Klarlund,
  Robert Paige, Theis Rauhe, and Anders Sandholm.
\newblock Mona: Monadic second-order logic in practice.
\newblock In {\em Tools and Algorithms for Construction and Analysis of
  Systems, First International Workshop, {TACAS} '95, Aarhus, Denmark, May
  19-20, 1995, Proceedings}, pages 89--110, 1995.

\bibitem[Hoa69]{Hoare69}
C.~A.~R. Hoare.
\newblock An axiomatic basis for computer programming.
\newblock {\em Commun. {ACM}}, 12(10):576--580, 1969.

\bibitem[Hod02]{hodkinson2002loosely}
Ian Hodkinson.
\newblock Loosely guarded fragment of first-order logic has the finite model
  property.
\newblock {\em Studia Logica}, 70(2):205--240, 2002.

\bibitem[IBI{\etalchar{+}}13]{CAV:IBINS13}
Shachar Itzhaky, Anindya Banerjee, Neil Immerman, Aleksandar Nanevski, and
  Mooly Sagiv.
\newblock Effectively-propositional reasoning about reachability in linked data
  structures.
\newblock In {\em CAV}, volume 8044 of {\em LNCS}, pages 756--772, 2013.

\bibitem[IBR{\etalchar{+}}14]{DBLP:conf/cav/ItzhakyBRST14}
Shachar Itzhaky, Nikolaj Bj{\o}rner, Thomas~W. Reps, Mooly Sagiv, and Aditya~V.
  Thakur.
\newblock Property-directed shape analysis.
\newblock In {\em Computer Aided Verification - 26th International Conference,
  {CAV} 2014, Held as Part of the Vienna Summer of Logic, {VSL} 2014, Vienna,
  Austria, July 18-22, 2014. Proceedings}, pages 35--51, 2014.

\bibitem[IJS08]{IhlemannJS08TACAS08}
Carsten Ihlemann, Swen Jacobs, and Viorica Sofronie{-}Stokkermans.
\newblock On local reasoning in verification.
\newblock In {\em Tools and Algorithms for the Construction and Analysis of
  Systems, 14th International Conference, {TACAS} 2008, Held as Part of the
  Joint European Conferences on Theory and Practice of Software, {ETAPS} 2008,
  Budapest, Hungary, March 29-April 6, 2008. Proceedings}, pages 265--281,
  2008.

\bibitem[Imm99]{DBLP:books/daglib/0095988}
Neil Immerman.
\newblock {\em Descriptive complexity}.
\newblock Graduate texts in computer science. Springer, 1999.

\bibitem[IRR{\etalchar{+}}04]{CSL:ImmermanRRSY04}
Neil Immerman, Alexander~Moshe Rabinovich, Thomas~W. Reps, Shmuel Sagiv, and
  Greta Yorsh.
\newblock The boundary between decidability and undecidability for
  transitive-closure logics.
\newblock In {\em CSL}, 2004.

\bibitem[KBI{\etalchar{+}}17]{KarbyshevBIRS17}
Aleksandr Karbyshev, Nikolaj Bj{\o}rner, Shachar Itzhaky, Noam Rinetzky, and
  Sharon Shoham.
\newblock Property-directed inference of universal invariants or proving their
  absence.
\newblock {\em J. {ACM}}, 64(1):7:1--7:33, 2017.

\bibitem[KV13]{DBLP:conf/cav/KovacsV13}
Laura Kov{\'{a}}cs and Andrei Voronkov.
\newblock First-order theorem proving and vampire.
\newblock In {\em Computer Aided Verification - 25th International Conference,
  {CAV} 2013, Saint Petersburg, Russia, July 13-19, 2013. Proceedings}, pages
  1--35, 2013.

\bibitem[Lei10]{dafny}
K~Rustan~M Leino.
\newblock Dafny: An automatic program verifier for functional correctness.
\newblock In {\em Logic for Programming, Artificial Intelligence, and
  Reasoning}, pages 348--370. Springer, 2010.

\bibitem[LP16]{leinotriggerCAV16}
K.~Rustan~M. Leino and Cl{\'{e}}ment Pit{-}Claudel.
\newblock Trigger selection strategies to stabilize program verifiers.
\newblock In {\em Computer Aided Verification - 28th International Conference,
  {CAV} 2016, Toronto, ON, Canada, July 17-23, 2016, Proceedings, Part {I}},
  pages 361--381, 2016.

\bibitem[LQ06]{LahiriQ06}
Shuvendu~K. Lahiri and Shaz Qadeer.
\newblock Verifying properties of well-founded linked lists.
\newblock In {\em Proceedings of the 33rd {ACM} {SIGPLAN-SIGACT} Symposium on
  Principles of Programming Languages, {POPL} 2006, Charleston, South Carolina,
  USA, January 11-13, 2006}, pages 115--126, 2006.

\bibitem[LQ08]{LahiriQ08}
Shuvendu~K. Lahiri and Shaz Qadeer.
\newblock Back to the future: revisiting precise program verification using
  {SMT} solvers.
\newblock In {\em Proceedings of the 35th {ACM} {SIGPLAN-SIGACT} Symposium on
  Principles of Programming Languages, {POPL} 2008, San Francisco, California,
  USA, January 7-12, 2008}, pages 171--182, 2008.

\bibitem[LS00]{Lev-AmiS00}
Tal Lev{-}Ami and Shmuel Sagiv.
\newblock {TVLA:} {A} system for implementing static analyses.
\newblock In {\em Static Analysis, 7th International Symposium, {SAS} 2000,
  Santa Barbara, CA, USA, June 29 - July 1, 2000, Proceedings}, pages 280--301,
  2000.

\bibitem[MP18]{DBLP:conf/sas/McMillanP18}
Kenneth~L. McMillan and Oded Padon.
\newblock Deductive verification in decidable fragments with ivy.
\newblock In {\em Static Analysis - 25th International Symposium, {SAS} 2018,
  Freiburg, Germany, August 29-31, 2018, Proceedings}, pages 43--55, 2018.

\bibitem[Nel89]{DBLP:journals/toplas/Nelson89}
Greg Nelson.
\newblock A generalization of dijkstra's calculus.
\newblock {\em {ACM} Trans. Program. Lang. Syst.}, 11(4):517--561, 1989.

\bibitem[Pad18]{oded_thesis}
Oded Padon.
\newblock {\em Deductive Verification of Distributed Protocols in First-Order
  Logic}.
\newblock PhD thesis, Tel Aviv University, 2018.

\bibitem[PLSS17]{paxosEpr}
Oded Padon, Giuliano Losa, Mooly Sagiv, and Sharon Shoham.
\newblock Paxos made epr: Decidable reasoning about distributed protocols.
\newblock {\em Proc. ACM Program. Lang.}, 1(OOPSLA):108:1--108:31, October
  2017.

\bibitem[PMP{\etalchar{+}}16]{pldi/PadonMPSS16}
Oded Padon, Kenneth~L. McMillan, Aurojit Panda, Mooly Sagiv, and Sharon Shoham.
\newblock Ivy: safety verification by interactive generalization.
\newblock In {\em Proceedings of the 37th {ACM} {SIGPLAN} Conference on
  Programming Language Design and Implementation, {PLDI} 2016, Santa Barbara,
  CA, USA, June 13-17, 2016}, pages 614--630, 2016.

\bibitem[QGSM13]{Qiu0SM13}
Xiaokang Qiu, Pranav Garg, Andrei Stefanescu, and Parthasarathy Madhusudan.
\newblock Natural proofs for structure, data, and separation.
\newblock In {\em {ACM} {SIGPLAN} Conference on Programming Language Design and
  Implementation, {PLDI} '13, Seattle, WA, USA, June 16-19, 2013}, pages
  231--242, 2013.

\bibitem[Ram30]{Ramsey01011930}
F.~P. Ramsey.
\newblock On a problem of formal logic.
\newblock {\em Proceedings of the London Mathematical Society},
  s2-30(1):264--286, 1930.

\bibitem[RS67]{Shoenfield67}
Joseph R.~Shoenfield.
\newblock {\em {Mathematical Logic}}.
\newblock Addison-Wesley, Reading, Massachusetts, 1967.

\bibitem[RSL10]{TOPLAS:RepsSL10}
Thomas~W. Reps, Mooly Sagiv, and Alexey Loginov.
\newblock Finite differencing of logical formulas for static analysis.
\newblock {\em ACM Trans. Program. Lang. Syst.}, 32(6), 2010.

\bibitem[RTG{\etalchar{+}}13]{DBLP:conf/cade/ReynoldsTGKDB13}
Andrew Reynolds, Cesare Tinelli, Amit Goel, Sava Krstic, Morgan Deters, and
  Clark Barrett.
\newblock Quantifier instantiation techniques for finite model finding in
  {SMT}.
\newblock In {\em Automated Deduction - {CADE-24} - 24th International
  Conference on Automated Deduction, Lake Placid, NY, USA, June 9-14, 2013.
  Proceedings}, pages 377--391, 2013.

\bibitem[RTGK13]{ReynoldsCAV13}
Andrew Reynolds, Cesare Tinelli, Amit Goel, and Sava Krstic.
\newblock Finite model finding in {SMT}.
\newblock In {\em Computer Aided Verification - 25th International Conference,
  {CAV} 2013, Saint Petersburg, Russia, July 13-19, 2013. Proceedings}, pages
  640--655, 2013.

\bibitem[Sof05]{DBLP:conf/cade/Sofronie-Stokkermans05}
Viorica Sofronie{-}Stokkermans.
\newblock Hierarchic reasoning in local theory extensions.
\newblock In {\em Automated Deduction - CADE-20, 20th International Conference
  on Automated Deduction, Tallinn, Estonia, July 22-27, 2005, Proceedings},
  pages 219--234, 2005.

\bibitem[SRW02]{TOPLAS:SRW02}
Shmuel Sagiv, Thomas~W. Reps, and Reinhard Wilhelm.
\newblock Parametric shape analysis via 3-valued logic.
\newblock {\em ACM Trans. Program. Lang. Syst.}, 24(3):217--298, 2002.

\bibitem[Wei01]{DBLP:books/el/RV01/Weidenbach01}
Christoph Weidenbach.
\newblock Combining superposition, sorts and splitting.
\newblock In {\em Handbook of Automated Reasoning (in 2 volumes)}, pages
  1965--2013. Elsevier, 2001.

\end{thebibliography}


\clearpage
\appendix
\def\maxcol{\textit{max}}

\section{Undecidability}%
\label{sec:undecidability}

\TODO{R1: check potential confusion pointed out by R1}

For a universal formula $I \in \forall^*(\Sigma)$, the formula $I \land \delta \land \neg I'$ is in EPR (recall that $\delta$ is specified in EPR). Hence, checking inductiveness amounts to checking the unsatisfiability of an EPR formula, and is therefore decidable. The same holds for $I \in AF(\Sigma)$. However, this is no longer true when quantifier alternation is introduced. For example, checking inductiveness of $I \in \forall^* \exists^* (\Sigma)$ amounts to checking unsatisfiability of a formula in a fragment for which satisfiability is undecidable.
In this section we show that checking inductiveness of $\AE$ formulas is indeed undecidable, even when the transition relation is restricted to $\EPR$.
The undecidability of the problem justifies sound but incomplete algorithms for checking inductiveness, one of which is the Bounded-Horizon algorithm (defined in \Cref{sec:bounded-horizon}) which we study in this paper.

\ultpara{Finite and infinite structures}
We begin by showing that the problem is undecidable when structures, or program states, are assumed to be finite.
This is the intention in most application domains~\cite{DBLP:books/daglib/0095988} (including the examples in \Cref{sec:implementation}), especially when the program does not involve numerical computations.
Nevertheless, in this paper we mostly concern ourselves with the problem of checking inductiveness when structures may also be infinite.
This is because SMT-based deductive verification relies on proof techniques from standard first-order logic, whose semantics are defined over \emph{general} structures, i.e.\ both finite and infinite.
We thus establish an undecidability result for this setting as well.
It is interesting to note that the discrepancy between the intended finiteness of the domain and the proof techniques, which cannot incorporate this assumption, re-emerges in \Cref{sec:partial-models}.

We refer to \emph{inductiveness over finite structures} when the validity of $I \land \delta \rightarrow I'$ is considered over finite structures, and
to \emph{inductiveness over general structures} when it is considered over both finite and infinite structures.

\ultpara{Scope of the proofs.}
The undecidability proofs of this section are by reductions from tiling problems.
Although technically it is also possible to prove the results by a trivial reduction from the satisfiability of $\AE$ formulas, since invariants for the transition relation $\true$ are necessarily either valid or unsatisfiable, we believe that the proofs presented here demonstrate the intuition behind the inherent difficulty of checking inductiveness of $\AE$ formulas in a more profound and robust way; the reduction using the transition relation $\true$ leaves more room for questions about decidability of the problem w.r.t.\ classes of transition relations more realistic and structured than $\true$. In contrast, the reductions we present here use transition systems (based on tiling problems) that we consider rather realistic, so the undecidability of checking $\AE$ invariants seems inherent rather than an artifact of just degenerate transition relations.

To further provide intuition, we prove the undecidability of the overarching problem of checking not only inductiveness of a candidate $I$, but also that $I$ holds initially and implies a given safety property, as is typically done when inductive invariants are used as a means to verify safety of a transition system.
%
Formally, the problem of \emph{checking inductive invariants for safety of transition systems} is defined as follows. Given a transition relation $\delta$ (over $\Sigma \uplus \Sigma'$), a sentence $\Init$ (over $\Sigma$) describing the set of initial states, a sentence $\PhiProp$ (over $\Sigma$) describing the safety property, and a sentence $I$ (the candidate inductive invariant), the problem is to check whether $\Init \rightarrow I$ (initiation), $I \land \delta \rightarrow I'$ (consecution), and $I \rightarrow \PhiProp$ (safety) are valid (over finite or general structures).
We will consider this problem when $I \in \AE$, $\Init \in \Univ$, $\PhiProp \in \Univ$ and our reductions will
generate instances where $\Init \rightarrow I$ and $I \rightarrow \PhiProp$ are valid (so it only remains to check whether $I \land \delta \rightarrow I'$ is valid).
With these restrictions, the undecidability of the problem of checking inductive invariants for safety of transition systems over general structures implies the undecidability of the problem of checking inductiveness as used elsewhere in this paper.

We now proceed to state the undecidability results and their proofs. We present different reductions for the problems of inductiveness over finite structures and over general structures.
Both reductions are from the halting problem,
albeit in opposite directions:
Over finite structures, the invariant is inductive iff there is no halting tiling iff the machine does not halt.
Over general structures, the invariant is inductive iff there is no infinite lower triangular tiling iff the machine halts.
This reflects Trakhtenbrot's theorem: over finite structures, model enumeration is possible but not proof enumeration, while over general structures the situation is reversed.

\subsection{Inductiveness Over Finite Structures}
\begin{theorem}%
\label{thm:undecidability-finite}
It is undecidable to check given $I \in \AE(\Sigma)$ and $\delta \in \EPR(\Sigma)$ whether $I$ is inductive for $\delta$ over finite structures.
\end{theorem}
The proof is based on a reduction from a variant of tiling problems. We start by defining the specific tiling problem used in the proof of this theorem:
\begin{defiC}[\cite{CSL:ImmermanRRSY04}]%
\label{def:tilability-halting}
A \emph{halting-tiling problem} consists of a finite set of tiles $T$ with designated tiles $T_{\text{start}}, T_{\text{halt}} \in T$, along with horizontal and vertical adjacency relations $\mathcal{H},\mathcal{V} \subseteq T \times T$.
A \emph{solution} to a halting-tiling problem is an arrangement of instances of the tiles in a finite rectangular grid (``board'') such that the tile $T_{\text{start}}$ appears in the top left position, the tileƒ $T_{\text{halt}}$ appears in the end of a row (the rightmost position in some row), and the adjacency relationships $\mathcal{H},\mathcal{V}$ are respected, meaning: if a tile $t_2$ appears immediately to the right of $t_1$ it must hold that $(t_1, t_2) \in \mathcal{H}$, and if a tile $t_2$ appears immediately below $t_1$ it must hold that $(t_1, t_2) \in \mathcal{V}$.
\end{defiC}
The problem is undecidable~\cite{CSL:ImmermanRRSY04}.
The proof is by a reduction from the halting problem: given a Turing machine we can compute a halting-tiling problem such that the problem has a solution iff the machine halts (on the empty input).
In the reduction, rows represent the tape of the Turing machine as it evolves over time (computation steps).
The tiles encode the location of the head and the current (control) state of the machine.
The horizontal and vertical constraints ensure that successive tiled rows correspond to a correct step of the machine, and the locality of constratins is possible by the locality of computation in a Turing machine. See~\cite{BGG} for further details.

\begin{proof}[Proof of \Cref{thm:undecidability-finite}]
The proof is by a reduction from non-tilability in the halting-tiling problem (\Cref{def:tilability-halting})
to the problem of checking inductive invariants for safety of a transition system over finite structures where the initiation and safety requirements are valid.
We think of the transition relation $\delta$ as incrementally placing tiles in a rectangle.

\ultpara{Vocabulary.}
To express locations on the board, we use a total order for both the horizontal and vertical dimensions of the board.
We add an immediate predecessor relation $j = i-1$ which is true if $j < i$ and there is no element of the order between $j,i$.
We use a constant $0$ for the minimal element of the order. These notions can be defined using a universally quantified formula.\footnote{ The assumption that there \emph{exists} a predecessor is left for the invariant to state explicitly when necessary, as this is the heart of the $\AE$ quantification in the proof.}
A \emph{location} is a pair of elements of the order, a \emph{vertical} and \emph{horizontal} component. We sometimes use the term \emph{board order} to refer to the lexicographic order of pairs of elements of the order.

The transition system keeps track of the last tile placed on the board by a relation $M(i,j)$ which is true only for the last updated location. Since the placing of tiles occurs in a sequential manner we also call this board location \emph{maximal}, and a location \emph{active} if it comes before the maximal location in the board order.
The \emph{active area} is the set of active locations.

The state of tiles on the board is represented by a set of relations $\{T_k\}$, one for each tile type, encoding the locations on the board where a tile of type $T_k$ is placed.

In this proof we also use a constant $\maxcol$ to be an element of the total order, representing the width of the rectangle.

\ultpara{Transitions.}
In every step the transition system places a valid tile in the next board location.
The next board location is considered while respecting the width of the rectangle, moving to the next row if the horizontal component of the current tile is $\maxcol$.

Placing a tile of type $T_{\text{next}}$ on the board is done by an EPR update of the (two-vocabulary) form
\begin{equation}\begin{split}
\label{eqn:undecidability-tile-update}
&\forall i,j. \ M'(i,j) \leftrightarrow \left(\left(j = 0 \land M(i-1, \maxcol)\right) \lor \left(j \neq 0 \land j \leq \maxcol \land M(i,j-1)\right)\right)
\\
&\forall i,j. \ T'_{\text{next}}(i, j) \leftrightarrow \left( T_{\text{next}}(i,j) \lor M'(i,j) \right)
\\
&\forall i,j. \ T'_{k}(i, j) \leftrightarrow T_{k}(i, j) \qquad \forall T_k \neq T_{\text{next}}.
\end{split}\end{equation}
The transition system nondeterministically chooses a tile $T_{\text{next}}$ type that respects the adjacency relationships.
These relationships are with the tiles in the board location preceding the current location in the horizontal and vertical components, expressible using the immediate predecessor relation and existential quantification on these predecessors. (Note that the existential quantifiers do not need to reside in the scope of universal quantifiers --- they depend only on the current location.)
Because the set of tile types $T$ is finite, expressing the allowed tile types given the two adjacent locations can be done by a quantifier-free formula.
Overall the EPR formula describing a step of the system consists of a disjunction between choices for $T_{\text{next}}$.
Each of these possible choices is described via a conjunction of the guard that makes sure that it is legal to place  $T_{\text{next}}$, and a corresponding update to the relation that is a conjunction of the formulas in \Cref{eqn:undecidability-tile-update}.

\ultpara{Initial state.}
Initially we only have $T_{\text{start}}$ placed in the upper-left corner, so $\forall i,j. \ T_{\text{start}}(i,j) \leftrightarrow \left( i=0 \land j=0 \right)$ and $\forall i,j. \ \neg T_k(i,j)$ for every other tile type $T_k$.

\ultpara{Safety property.}
The safety property states that the special tile $T_{\text{halt}}$,
is not placed on the board in the end of a row (in a $\maxcol$ position) in the active area.

\ultpara{Invariant.}
The invariant states that in the active area we have a valid partial tiling. We require this by a $\AE$ formula saying that for every tile placed in an active location (except for the maximal location) there is a successor tile, placed in the next board location, that conforms to the (local) ajdacency relations.\footnote{
  We specify the requirement in this forward fashion, rather than requiring that every tile has a valid predecessor, in order to easily reuse this invariant in the proof of \Cref{thm:undecidability-infinite}.
}
%
We also conjoin the safety property to the invariant.




\ultpara{Reduction argument.}
The invariant holds for the initial state, and trivially implies the safety property.

If there exists a valid tiling with $T_{\text{halt}}$ in the end of a row, a counterexample to induction can be obtained by encoding this valid tiling in the post-state and that same tiling without $T_{\text{halt}}$, which is the last-placed tile, in the pre-state.

For the converse, assume that the invariant is not inductive over finite structures, i.e., there exists a finite counterexample to induction, and show that there exists a solution to the halting-tiling problem. The reasoning is as follows:
A finite state satisfying the invariant induces a valid finite partial tiling (defined by the active area of the board in the structure).
Since the transition system always places a tile that respects the adjacency relations, it is easy to see that a counterexample to induction must be such that the transition places $T_{\text{halt}}$ on the board in the end of a row, and that this also induces a valid partial finite tiling in the post-state.
Thus a finite counterexample to induction implies the existence of a valid finite tiling with $T_{\text{halt}}$ in the end of a row, 
which is a solution to the halting-tiling problem.

Thus the invariant is inductive iff the halting-tiling problem does not have a solution.
\end{proof}

\subsection{Inductiveness Over General Structures}
\begin{theorem}%
\label{thm:undecidability-infinite}
It is undecidable to check given $I \in \AE(\Sigma)$ and $\delta \in \EPR(\Sigma)$ whether $I$ is inductive for $\delta$ over general (finite and infinite) structures.
\end{theorem}
The proof is based on a reduction from a variant of tiling problems. We start by defining the specific tiling problem used in the proof of this theorem:
\begin{definition}%
\label{def:tilability-infinite}
A \emph{lower-triangular infinite-tiling problem} consists of a finite set of tiles $T$,
along with horizontal and vertical adjacency relations $\mathcal{H},\mathcal{V} \subseteq T \times T$.
A \emph{solution} to a lower-triangular infinite-tiling problem is an arrangement of instances of the tiles in the lower-triangular plane
(i.e., a total function $\{(i, j) \in \Nat \times \Nat \; | \; i \leq j\} \to T$)
where the adjacency relationships $\mathcal{H},\mathcal{V}$ are respected, meaning: if a tile $t_2$ appears immediately to the right of $t_1$ it must hold that $(t_1, t_2) \in \mathcal{H}$, and if a tile $t_2$ appears immediately below $t_1$ it must hold that $(t_1, t_2) \in \mathcal{V}$.
\end{definition}
The problem is undecidable. 
The proof is by a reduction from the non-halting problem: given a Turing machine we can compute a lower-triangular infinite-tiling problem such that the problem has a solution iff the machine does not halt (on the empty input). 
The encoding is similar to~\cite{CSL:ImmermanRRSY04}.

\begin{proof}[Proof of \Cref{thm:undecidability-infinite}]
The proof is by a reduction from non-tilability in the lower-triangular infinite-tiling problem (\Cref{def:tilability-infinite})
to the problem of checking inductive invariants for safety of a transition system (over general structures) where the initiation and safety requirements are valid.

We construct a transition relation similar to the one in the proof of \Cref{thm:undecidability-finite}, with some changes, as described below.

\ultpara{Discussion and motivation.}
\TODO{R2: intuitive explanation unhelpful (if this is where they are referring to)}
To provide some intuition to the difference between the reductions, we remark that both of the proofs in this section are in essence a reduction from the halting (or non-halting) problem.
The proof of \Cref{thm:undecidability-finite} encodes runs of the machine as finite tilings, and asks whether a tiling that represents a terminating computation, encoded by $T_{\text{halt}}$, is possible.
This reduction is no longer adequate when structures may be infinite.
The reason is that an infinite valid partial tiling may not correspond to reachable configurations of the Turing machine, so there may be such an infinite tiling with $T_{\text{halt}}$ even though the Turing machine never halts.\footnote{
    One way to construct such a tiling, using a tile in row $\omega$ of the board, is utilized in the proof that follows.
}

In fact, the reduction in this proof must be in the opposite direction: the invariant should be inductive iff the machine terminates, whereas in the proof of \Cref{thm:undecidability-finite} the invariant is inductive iff the machine does not terminate. This is because satisfiability is recursively-enumerable over finite structures and co-recursively-enumerable over general structures (due to the existence of proofs), which reflects on checking inductiveness through the satisfiability check of the formula $I \land \delta \land \neg I'$.

Thus, we would like to have a counterexample to induction when the machine never halts, i.e.\ has an infinite run. As before, runs of the machine are encoded via tiling, only that now an infinite structure can encode an infinite run of the machine. (It is not necessary that an infinite tiling represents a valid infinite run of the machine, but every infinite run can be represented by such a structure.)
We would like to ``detect'' this situation.
Our way to do this is by the observation that induction on the number of rows, or execution steps, must hold when the number of rows is finite (but unbounded), as in \Cref{thm:undecidability-finite}, but does not necessarily hold when there may be an infinite number of rows.
This idea is implemented by a relation $P$ with the invariant that it is preserved under successive board locations. In an infinite structure this does not imply that $P$ is true for all locations. A flag $f$ is used to express a transition that is aware of $P$ not being globally true.
\commentout{
	A technical complication arises because to make sure that the infinite number of rows truely represents an infinite number of computation steps, we also need to make sure that the number of columns is infinite (representing the entire tape).
	This is done using a relation $P$ that ``detects'' the infinite number of columns similarly to $H$.
	We are assured in the existence of an infinite execution when both $H$ and $P$ ``detect'' infinity of the tiling, and we turn the flag $f$ to $\false$ to express this.
}

Another technical detail is the lower-triangular formulation of the tiling problem, which is used to construct the infinite computation of the transition system by placing a single tile in each step.

Returning to the proof, we describe the reduction and highlight its differences from the reduction in \Cref{thm:undecidability-finite}.
Following the lower-triangular formulation of the tiling problem, we restrict the board order to the lower-triangular part (locations $(i,j)$ such that $i \leq j$)  and ignore other locations when considering successor in the board order.

\ultpara{Vocabulary.}
We add 
a relation $P$ over board locations, and a Boolean flag (nullary predicate)~$f$.

\ultpara{Transitions.}
In each step the transition system places a valid tile in the next board location, similar to the proof of \Cref{thm:undecidability-finite}.
The difference is that the criterion for moving to place tiles in the next tile is when the current location $(i,j)$ has $i=j$ (whereas in \Cref{thm:undecidability-finite} the criterion was $j=\maxcol$).
\commentout{
	To maintain the invariant that $H$ is preserved under successor of active rows (see below), when we place a new tile, if $H$ holds for the current line (the element of the vertical order of the maximal active location) after placing the tile, set $H$ to true for the row after (next element of the vertical order) as well.
	(For ease of expression, in this proof we let the transition system place the first tile itself.)
}

To maintain the invariant that $P$ is preserved under successor of active tiles in the board, when we place a new tile, if $P$ holds for the maximal location before the step, set $P$ to true for the new maximal location.

If $P$ does not hold for the new maximal location, turn $f$ to $\false$.


\ultpara{Initial state.}
$P$ is true for the first location $(0,0)$ only, and $f$ is $\true$.
In this proof, initially the board is empty.

\ultpara{Safety property.}
The safety property now asserts that $f$ is $\true$.

\ultpara{Invariant.}
As before, the invariant states that the active board represents a valid partial tiling, i.e.\ every active tile except for the maximal one has a valid successor. 

\commentout{In addition, the invariant states that $H$ is preserved under successor of active board rows, i.e.: If $i_1,i_2$ are active rows --- meaning $i_1,i_2 \leq i$ where $M(i,j)$ --- and $i_2$ is the successor of $i_1$ w.r.t.\ the vertical order, then if $H$ holds for $i_1$ it must also hold for $i_2$.
}

The invariant also states that $P$ is preserved under successor of board\commentout{ locations in any line}, i.e., if $(i_1,j_1)$ and $(i_2,j_2)$ are successive active board locations w.r.t.\ the board order, then if $P$ holds for $(i_1,j_1)$ is must also hold for $(i_2,j_2)$.
We also conjoin the safety property to the invariant.

\ultpara{Reduction argument.}
The invariant holds for the initial state, and trivially implies the safety property.

Assume that there is no solution to the lower-triangular infinite-tiling problem, and show that the invariant is inductive.
The reasoning is as follows:
A state satisfying the invariant induces a partial valid tiling --- either finite or infinite --- over the active area of the board.
Since there is no valid partial tiling with an infinite number of rows, the number of active locations must be finite (the number of columns in the active domain is bounded by the number of rows, since we are discussing lower-triangular tilings).
Because $P$ is preserved under successor of the board order\commentout{in every line}, by induction on the number of locations\commentout{(in the current line)}, $P$ must hold for the maximal location.
\commentout{Either way, a}After a transition is taken, $f$ remains $\true$.
Since the transition system always places a tile that respects the horizontal and vertical ajdacency relations and sets $P$ to true for the new maximal location, it is easy to see that the rest of the invariant is preserved by a transition as well.

For the converse direction, if there is a solution to the lower-triangular infinite-tiling problem, then there is an infinite structure encoding this tiling.
The transition begins with the infinite valid tiling, with a new additional row \emph{after} this infinite sequence of tiled rows.
(Recall that the board dimensions are axiomatized using a total order; the additional row index corresponds to ordinal $\omega$ of vertical order.)
We place some tile in the first column of this row as in some valid row in the tiling.
Note that when placing tiles in this row we need not worry about vertical constraints, because they were expressed in a forward fashion, and this row is not a successor of any other row.
The first leftmost location in the new row is set to be the maximal active one, and we set $P$ to be $\false$ for this location.
Note that this does not violate the invariant: $P$ is preserved under successor of the board location, but nonetheless does not hold for the location in the additional row (it is not the successor of any location).
The transition will now place a new tile and turn $f$ to $\false$, $P$ does not hold for the current maximal location, thereby violating the invariant.

\commentout{
	The transition begins with the infinite valid tiling, with a new valid row in an additional row \emph{after} this infinite sequence of tiled rows.
	(Recall that the board dimensions are axiomatized to be total orders; such a tile is placed in a location corresponding to ordinal $\omega$ of the vertical order.)
	We can tile this row exactly the as some other arbitrary line in the tiling and need not worry about vertical constraints, because they were expressed in a forward fashion, and this row is not a successor of any other row.
	We set $H$ to be $\false$ for this additional row.

	We also add two successive columns (successive elements of the horizontal order) \emph{after} the infinite tiling (such elements correspond to ordinals $\omega$ and $\omega + 1$ of the horizontal order). Since the constraints were expressed in a forward fashion, there are no horizontal constraints on the first of these columns, as it is not the successor of any column. Therefore we use two successive tiled columns from the valid tiling and use the same tiling for the new columns.
	We set $P$ to be $\false$ for the location in these columns in the new row.

	The location in the newly added row and the first new column (i.e. $(\omega, \omega)$) is set to be the maximal active one.
	Note that this does not violate the invariant: $H$ is preserved under successor of active rows, but nonetheless does not hold for all active rows (in essence, the induction fails).
	Similarly, $P$ is preserved under successor of columns in the last line, but nonetheless does not hold for all columns.
	The transition will now place a new tile and turn $f$ to $\false$, because $H$ does not hold for the current line and $P$ does not hold for the current maximal location, thereby violating the invariant.
}

Thus the invariant is inductive iff the infinite tiling problem does not have a solution.
\end{proof}

\end{document}